\documentclass[12pt]{article}

\usepackage{amsmath,amsfonts,amssymb}
\usepackage{graphicx,psfrag,epsf,color}
\usepackage{enumerate}
\usepackage{natbib}
\usepackage{url}
\usepackage{epstopdf}
\usepackage{amsthm}
\usepackage{bm}
\usepackage{rotating}
\usepackage{booktabs}
\usepackage{multirow}
\usepackage{threeparttable}
\usepackage{caption}

\usepackage{lscape}
\usepackage{titlesec}
\usepackage{tabularx}
\usepackage{array}
\usepackage{ragged2e}
\usepackage{longtable}

\usepackage[colorlinks=true,bookmarksopen=true,bookmarksnumbered=true,
            citecolor=red,urlcolor=red,linkcolor=black]{hyperref}

\newtheorem{assumption}{Assumption}
\newtheorem{lemma}{Lemma}
\newtheorem{theorem}{Theorem}

\newcolumntype{L}[1]{>{\RaggedRight\arraybackslash}p{#1}}
\newcolumntype{R}[1]{>{\RaggedLeft\arraybackslash}p{#1}}
\newcolumntype{C}[1]{>{\Centering\arraybackslash}p{#1}}
\newcolumntype{Y}{>{\RaggedRight\arraybackslash}X}

\newcommand{\blind}{0}

\begin{document}

% Keep sections 1-7 out of the printed table of contents (which lists only
% the appendices); the switch back happens at the supplement.
\addtocontents{toc}{\protect\setcounter{tocdepth}{-1}}

\def\spacingset#1{\renewcommand{\baselinestretch}%
{#1}\small\normalsize} \spacingset{1}

%%%%%%%%%%%%%%%%%%%%%%%%%%%%%%%%%%%%%%%%%%%%%%%%%%%%%%%%%%%%%%%%%%%%%%%%%%%%%%

\if0\blind
{
\title{\bf Simultaneously accounting for the winner's curse and sample structure in Mendelian randomization: bivariate rerandomized inverse variance weighted estimator}

\author{Xin Liu,$\rm{^{a,\dagger}}$ Youpeng Su,$\rm{^{a,\dagger}}$ Chuchu Zeng,$\rm{^{a}}$ Yilei Ma,$\rm{^{b}}$ \\
Siqi Xu,$\rm{^{c}}$ Ping Yin,$\rm{^{a}}$ and Peng Wang$\rm{^{a,*}}$
\hspace{.2cm}\\
\ \\
$\rm{^{a}}$Department of Epidemiology and Biostatistics, School of Public Health,\\
Tongji Medical College, Huazhong University of Science and Technology,\\
Wuhan, China\\[0.35em]
$\rm{^{b}}$Tongji Hospital, Tongji Medical College,\\
Huazhong University of Science and Technology, Wuhan, China\\[0.35em]
$\rm{^{c}}$Clinical Research Centre, Guangzhou First People's Hospital,\\
Guangzhou Medical University, Guangzhou, China
}

\maketitle

\begingroup
\renewcommand{\thefootnote}{\fnsymbol{footnote}}
\footnotetext[2]{These authors contributed equally to this work.}
\footnotetext[1]{Correspondence: pengwang\_stat@hust.edu.cn}
\endgroup
}
\fi

\newpage
\begin{abstract}
The recently developed rerandomized inverse variance weighted (RIVW) estimator provides a simple and efficient framework to break the winner's curse in two-sample Mendelian randomization (MR).
However, this method does not account for sample structure (e.g., residual population stratification and sample overlap), a common source of confounding in MR studies.
Sample structure can not only distort SNP--exposure and SNP--outcome association estimates but also induce correlation between them, leading exposure-side instrument selection to propagate bias to the outcome side.
To address this challenge, we propose the bivariate RIVW (BRIVW) estimator to simultaneously account for the winner's curse and sample structure.
The BRIVW estimator extends the RIVW framework by modeling the joint distribution of SNP--exposure and SNP--outcome association estimates, first adjusting their covariance matrix via linkage disequilibrium score regression to account for sample structure and then applying randomized instrument selection and bivariate Rao--Blackwellization to obtain unbiased post-selection association estimates together with an estimator of their covariance matrix.
Under mild conditions, we show that the BRIVW estimator is consistent and asymptotically normal.
The finite-sample performance of the proposed estimator is evaluated through extensive simulations and real data analyses.
\end{abstract}
{\it Keywords:}  causal inference, inverse variance weighting, Mendelian randomization, sample structure, winner's curse

\vfill

\newpage
\spacingset{1.45} % DON'T change the spacing!
\section{Introduction}
Establishing causal relationships remains a central goal in biomedical research (Rothman and Greenland 2005). 
Mendelian randomization (MR) has become a widely used framework for this purpose, using single nucleotide polymorphisms (SNPs) as instrumental variables (IVs) to infer causal effects from observational data (Smith and Ebrahim 2003; Didelez and Sheehan 2007; Burgess et al. 2012).
The growing availability of genome-wide association study (GWAS) summary statistics has made two-sample MR increasingly scalable across complex traits.
In practice, however, such a design often faces two complications: SNPs may have weak to moderate associations with the exposure, making exposure-based IV selection routine, 
and exposure and outcome GWASs may retain sample structure due to sample overlap or residual population stratification.

The weak IV problem has motivated the development of methods such as the debiased inverse variance weighted (dIVW) estimator (Ye, Shao, and Kang 2021) and its extensions (Xu et al. 2023; Su et al. 2024). 
Nevertheless, exposure-based IV selection remains common in MR studies as a means of improving IV strength and estimation precision.
However, using the same data for IV selection and effect estimation can induce winner's curse bias in SNP-exposure association estimates (Gkatzionis and Burgess 2019; Jiang et al. 2023). 
The recently proposed rerandomized IVW (RIVW) estimator (Ma, Wang, and Wu 2023) addresses this issue by combining randomized IV selection with Rao--Blackwellization, thereby providing a simple and effective post-selection correction.

Sample structure is another common concern in MR studies, particularly in large biobanks and GWAS consortia, where shared participants or cohorts often lead to sample overlap between exposure and outcome GWASs. 
Meanwhile, residual population stratification may persist even after adjustment with principal components or linear mixed models (Price et al. 2006; Loh et al. 2015).
Such sample structure can alter the joint sampling distribution of the SNP-exposure and SNP-outcome association estimates by inflating their marginal variances and inducing correlation between their estimation errors (Burgess, Davies, and Thompson 2016; Sanderson et al. 2021; Hu et al. 2022).

The preceding two issues have mostly been discussed separately,
but the key challenge arises when exposure-based IV selection is performed in the presence of sample structure.
Because SNP--exposure and SNP--outcome association estimates are correlated, exposure-based IV selection not only affects the selected SNP-exposure association estimates but also shifts the conditional distribution of the corresponding SNP-outcome association estimates as well as their post-selection covariance. 
The winner's curse can therefore propagate from the exposure side to the outcome side, creating a bivariate post-selection problem. 
Consequently, the RIVW estimator may lose its validity because it does not account for either the induced outcome-side shift or the post-selection covariance.
Although such a bivariate post-selection problem has been recognized previously (Sadreev et al. 2021; Jiang et al. 2023), statistical methods that can jointly address the winner's curse and sample structure remain limited.
A notable exception is MR-APSS, which provides a likelihood-based framework that jointly models SNP-exposure and SNP-outcome association estimates, using linkage disequilibrium score regression (LDSC) to account for sample structure and a truncated bivariate normal model to adjust for winner's curse bias (Hu et al. 2022). 
However, its reliance on parametric modeling assumptions and variational inference motivates the development of a complementary, simple, closed-form approach.

To fill this gap, we propose the bivariate RIVW (BRIVW) estimator, which extends the RIVW framework to explicitly model the joint distribution of SNP--exposure and SNP--outcome association estimates.
We first adjust the covariance matrix of these estimates using LDSC to account for sample structure (Bulik-Sullivan et al. 2015a, b; Hu et al. 2022).
Rather than simply applying RIVW after LDSC calibration, we develop a bivariate Rao--Blackwellization procedure under randomized IV selection to construct unbiased post-selection association estimates together with an estimator of their covariance matrix.
The resulting estimator retains the closed-form simplicity and computational scalability of IVW-type methods while simultaneously addressing the winner's curse and sample structure.
We establish the asymptotic properties of the proposed estimator under mild conditions.
Extensive simulations and real data analyses show that the proposed estimator can properly account for both the winner's curse and sample structure.

The remainder of the article is organized as follows. 
Section 2 introduces the notation and assumptions. 
Section 3 briefly reviews the RIVW estimator.
Section 4 characterizes the bivariate post-selection bias problem under sample structure and develops the BRIVW estimator.
Sections 5 and 6 report simulation results and real data applications, respectively. Section 7 concludes with a discussion.

\section{Notation and Assumptions}
We assume that, after LD pruning or clumping (Purcell et al. 2007; Hemani et al. 2018), $p$ independent SNPs
$G_1,\ldots,G_p$ are available as candidate IVs.
Let $\beta$ denote the causal effect of an exposure $X$ on an outcome $Y$.
We adopt the following linear structural model (Pierce and Burgess 2013; Burgess, Small, and Thompson 2017):
\begin{align*}
X &= \sum_{j=1}^p \gamma_j G_j + \eta_X U + E_X, \\
Y &= \beta X + \sum_{j=1}^{p} \alpha_j G_j + \eta_Y U + E_Y,
\end{align*}
where $\gamma_j$ denotes the effect of $G_j$ on $X$, $\alpha_j$ denotes the direct pleiotropic effect of $G_j$ on $Y$,
and $U$ is an unmeasured confounder independent of all $G_j$. $E_X$ and $E_Y$ are mutually independent error terms that are also independent of $(G_1,\ldots,G_p,U)$.
Let $\Gamma_j$ denote the effect of $G_j$ on $Y$. Under the above model,
$\Gamma_j = \beta \gamma_j + \alpha_j$.

Two-sample MR is commonly conducted using GWAS summary statistics
(Burgess et al. 2015).
For each \(G_j\), let \(\hat{\gamma}_j\) and \(\hat{\Gamma}_j\)
denote its marginal association estimates with the exposure and outcome, respectively.
Although they are expected to be independent under the standard two-sample design,
sample structure may violate this condition
(Burgess, Davies, and Thompson 2016; Hu et al. 2022; Zhang et al. 2025).
We therefore adopt the following assumptions:
\begin{assumption}
The exposure and outcome GWAS sample sizes, denoted by $n_X$ and $n_Y$,
satisfy $n_X,n_Y \to \infty$ at the same order, and the number of
independent SNPs $p \to \infty$.
\end{assumption}

\begin{assumption}
\begin{enumerate}[(i)]
    \item For any $j \neq j'$, the pairs $(\hat{\gamma}_j, \hat{\Gamma}_j)$ and $(\hat{\gamma}_{j'}, \hat{\Gamma}_{j'})$ are independent.
    \item For each SNP $j$,     
    \[
        \begin{bmatrix}
        \hat{\gamma}_j \\[2pt]
        \hat{\Gamma}_j
        \end{bmatrix} \sim 
        N \left(
        \begin{bmatrix}
        \gamma_j \\[2pt]
        \beta\gamma_j+\alpha_j
        \end{bmatrix},
        \begin{bmatrix}
        \sigma_{\hat{\gamma}_j}^2 & \rho \, \sigma_{\hat{\gamma}_j} \sigma_{\hat{\Gamma}_j} \\[2pt]
        \rho \, \sigma_{\hat{\gamma}_j} \sigma_{\hat{\Gamma}_j} & \sigma_{\hat{\Gamma}_j}^2
        \end{bmatrix}
        \right),
    \]
where \(\rho\in(-1,1)\) captures the correlation induced by sample structure.
In addition, there exists some $v \to 0 $ such that $\{\sigma_{\hat{\gamma}_j}/v,\; \sigma_{\hat{\Gamma}_j}/v: 1 \le j \le p\}$ are bounded and bounded away from zero.

\end{enumerate}
\end{assumption}

\noindent\textbf{Remark.}
The standard errors \(\sigma_{\hat{\gamma}_j}\) and \(\sigma_{\hat{\Gamma}_j}\) in Assumption 2 are the working standard errors used for inference. 
In the absence of sample structure, the working standard errors coincide with the reported GWAS standard errors, and there is no correlation between $\hat{\gamma}_j$ and $\hat{\Gamma}_j$ (i.e., $\rho=0$).
Otherwise, the single-trait LDSC intercept estimates \(c_1\) and \(c_2\) are used to scale the marginal variances for the exposure and outcome GWASs, respectively, while the cross-trait LDSC intercept estimate \(c_{12}\) captures the correlation induced by sample structure, with $\rho = c_{12}/\sqrt{c_1 c_2}$.
Notably, $c_1$, $c_2$, and $c_{12}$ are treated as known rather than estimated in the subsequent analysis to avoid cumbersome notation. (Liu and Lin 2018; Hu et al. 2022).

We next state the balanced pleiotropy and InSIDE (Instrument Strength Independent of Direct Effect)  conditions required for the validity of the BRIVW estimator.
\begin{assumption}
The pleiotropic effects \(\alpha_1,\ldots,\alpha_p\) are mutually independent
across SNPs and independent of their corresponding instrument strengths
\(\gamma_1,\ldots,\gamma_p\). The former have mean \(0\), variance
\(\tau_\alpha^2\), and a bounded third moment. In addition,
\(\tau_\alpha/v\) is bounded.
\end{assumption}

Assumption 3 allows balanced pleiotropy but rules out directional and correlated pleiotropy.
The special case \(\alpha_j=0\) for all \(j\) corresponds to no horizontal pleiotropy.
The boundedness of \(\tau_\alpha/v\), as also assumed for RIVW,
is a technical condition that keeps the scale of the pleiotropic effects no larger than the common sampling-error scale, 
preventing their contribution from dominating the asymptotic variance.

Finally, we introduce some notation for probabilistic ordering. We use $\xrightarrow{p}$ and $\xrightarrow{d}$ to denote convergence in probability and
convergence in distribution, respectively.
For random variables $A$ and $B$, $A = O_p(B)$ indicates that
$A/B$ is bounded in probability, whereas $A = o_p(B)$ implies that
$A/B \xrightarrow{p} 0$. We write $A=\Theta_p(B)$ if both $A=O_p(B)$ and $B=O_p(A)$.

\section{Overview of the RIVW Estimator}
For a prespecified threshold $\lambda$, selecting IVs with
$\lvert \hat{\gamma}_j / \sigma_{\hat{\gamma}_j} \rvert > \lambda$
truncates the distribution of $\hat{\gamma}_j$ and thereby induces winner's curse bias.
The RIVW estimator (Ma, Wang, and Wu 2023) removes this bias by introducing independent pseudo SNP--exposure associations \(Z_j\sim N(0,\eta^2)\) and performing randomized IV selection:
\[
\mathcal S_{\lambda}
=
\left\{
j : S_j > 0,\; j = 1, 2, \ldots, p
\right\},
\qquad
\text{where }
S_j
=
\left|
\frac{\hat{\gamma}_j}{\sigma_{\hat{\gamma}_j}}
+
Z_j
\right|
-
\lambda.
\]
Here, \(\eta\) is a prespecified tuning parameter that is required to be bounded away from zero and infinity.
RIVW first constructs an initial estimator of
\(\gamma_j\) that is unbiased and independent of the selection event:
\[
\hat{\gamma}_{j,\mathrm{ini}}
=
\hat{\gamma}_j
-
\frac{\sigma_{\hat{\gamma}_j}}{\eta^2}Z_j.
\]
Then, Rao--Blackwellizing this initial estimator improves its efficiency and yields
\[
\hat{\gamma}_{j,\mathrm{RB}}
    =
    \hat{\gamma}_j 
    - 
    \frac{\sigma_{\hat{\gamma}_j}}{\eta}\,
    \frac{\phi(A_{j,+}) - \phi(A_{j,-})}
         {1 - \Phi(A_{j,+}) + \Phi(A_{j,-})},
\quad
\text{where }
A_{j,\pm}
    = -\frac{\hat{\gamma}_j}{\sigma_{\hat{\gamma}_j}\eta} 
      \pm \frac{\lambda}{\eta},
\]
and $\phi(\cdot)$ and $\Phi(\cdot)$ denote the probability density and cumulative distribution functions of the standard normal distribution, respectively.
The estimator $\hat{\gamma}_{j,\mathrm{RB}}$ is unbiased for \(\gamma_j\) conditional on the selection event and achieves minimum variance for a given $\eta$.
Incorporating $\hat{\gamma}_{j,\mathrm{RB}}$ and its variance estimator $\hat{\sigma}^2_{\hat{\gamma}_{j,\mathrm{RB}}}$, the RIVW estimator is defined as
\[
\hat{\beta}_{\mathrm{RIVW}}
    = 
    \frac{ \sum_{j \in \mathcal S_{\lambda}}
            \sigma_{\hat{\Gamma}_j}^{-2}
            \hat{\Gamma}_j \hat{\gamma}_{j,\mathrm{RB}} }
         { \sum_{j \in \mathcal S_{\lambda}}
            \sigma_{\hat{\Gamma}_j}^{-2}
            (\hat{\gamma}_{j,\mathrm{RB}}^2
             - \hat{\sigma}^2_{\hat{\gamma}_{j,\mathrm{RB}}}) },
\]
which corrects for both winner's curse and weak IV bias.

However, the RIVW estimator is developed under the standard setting in which \(\rho=0\). 
When sample structure induces correlation between \(\hat{\gamma}_j\) and \(\hat{\Gamma}_j\),
selection based on \(\hat{\gamma}_j\) also distorts the post-selection distribution of \(\hat{\Gamma}_j\), motivating the development of the BRIVW estimator in Section~4.

\section{Methods}
\subsection{Bivariate post-selection bias under sample structure}
When \(\rho=0\), the SNP--outcome association estimate
\(\hat{\Gamma}_j\) is independent of the randomized selection event
\(j\in\mathcal S_\lambda\), so only the selected
SNP--exposure association estimates need to be adjusted.
However, when sample structure induces a nonzero $\rho$, this independence no longer holds:
selection based on \(\hat{\gamma}_j\) also changes the conditional
distribution of \(\hat{\Gamma}_j\).
The following lemma makes the outcome-side selection effect explicit.

\begin{lemma}
Assume that Assumption~1--2 hold. Then
\[
\mathbb{E}[\hat{\Gamma}_j \mid S_j>0]
=
\Gamma_j
+
\underbrace{
\dfrac{\rho \, \sigma_{\hat{\Gamma}_j}}
{\Pr(S_j>0)\,\sqrt{1+\eta^2}}
\Bigg[
\phi\Bigg(
\dfrac{\lambda - \dfrac{\gamma_j}{\sigma_{\hat{\gamma}_j}}}{\sqrt{1+\eta^2}}
\Bigg)
-
\phi\Bigg(
\dfrac{-\lambda - \dfrac{\gamma_j}{\sigma_{\hat{\gamma}_j}}}{\sqrt{1+\eta^2}}
\Bigg)
\Bigg]
}_{\text{bias}}
\]
where
\[
\Pr(S_j > 0)
=
\Phi\Bigg(
\dfrac{
-\lambda - \dfrac{\gamma_j}{\sigma_{\hat{\gamma}_j}}
}{
\sqrt{1 + \eta^2}
}
\Bigg)
+
1 -
\Phi\Bigg(
\dfrac{
\lambda - \dfrac{\gamma_j}{\sigma_{\hat{\gamma}_j}}
}{
\sqrt{1 + \eta^2}
}
\Bigg).
\]
\end{lemma}
\noindent
The proof of Lemma 1 is provided in Web Appendix A.1. It shows that exposure-based IV selection induces a conditional mean
shift in \(\hat{\Gamma}_j\) when \(\rho\neq0\).
The direction of the shift depends on \(\rho\) and \(\gamma_j\),
while its magnitude increases linearly with \(\lvert\rho\rvert\).
This shift vanishes when \(\rho=0\).
Beyond this mean shift, the selection event can also
alter the covariance between \(\hat{\gamma}_j\) and
\(\hat{\Gamma}_j\), so an exposure-side correction alone is generally
insufficient in the presence of sample structure.

\subsection{The BRIVW estimator}
Motivated by the outcome-side mean shift and covariance distortion identified above, we extend RIVW to the bivariate post-selection setting under sample structure.
Throughout this subsection, \(\sigma_{\hat{\gamma}_j}\) and \(\sigma_{\hat{\Gamma}_j}\) denote the LDSC-calibrated working standard errors, and $\rho$ denotes the correlation coefficient between the estimation errors, as specified in Assumption~2.

We first correct the outcome-side post-selection bias in \(\hat{\Gamma}_j\). 
Under Assumption~2,
\[
\operatorname{Cov}\left(
\hat{\Gamma}_j,\,
\hat{\gamma}_j/\sigma_{\hat{\gamma}_j}+Z_j
\right)
=
\rho\sigma_{\hat{\Gamma}_j},
\]
so \(\hat{\Gamma}_j\) is no longer independent of the selection event
\(j\in\mathcal S_\lambda\) when \(\rho\neq0\).
Consider the unbiased initial estimator
\[
\hat{\Gamma}_{j,\mathrm{ini}}
=
\hat{\Gamma}_j
-
\frac{\rho\sigma_{\hat{\Gamma}_j}}{\eta^2}Z_j.
\]
Its covariance with
\(\hat{\gamma}_j/\sigma_{\hat{\gamma}_j}+Z_j\) is zero.
Because the two quantities are jointly normally distributed,
\(\hat{\Gamma}_{j,\mathrm{ini}}\) is independent of the selection event and
therefore remains unbiased for \(\Gamma_j\) after selection.
Next, we apply the Rao--Blackwellization procedure to reduce the variance of $\hat{\Gamma}_{j,\mathrm{ini}}$ by conditioning on the observed summary statistics and the selection event:
\[
\hat{\Gamma}_{j,\mathrm{RB}}
=
\mathbb{E}[\hat{\Gamma}_{j,\mathrm{ini}} \mid \hat{\Gamma}_j, \hat{\gamma}_j, j \in \mathcal S_\lambda]
=\hat{\Gamma}_j
-
\frac{\rho\,\sigma_{\hat{\Gamma}_j}}{\eta}
\,
\frac{
\phi\!\left(A_{j,+}\right)
-
\phi\!\left(A_{j,-}\right)
}{
1-\Phi\!\left(A_{j,+}\right)
+
\Phi\!\left(A_{j,-}\right)
}.
\]
As shown in Web Appendix~A.2, the law of iterated expectations yields
$
\mathbb{E}\left[
\hat{\Gamma}_{j,\mathrm{RB}}
\mid j\in\mathcal S_\lambda
\right]
=
\Gamma_j.
$
Numerical illustrations in Web Figure~S1 further show that \(\hat{\Gamma}_{j,\mathrm{RB}}\) remains well centered around \(\Gamma_j\) under sample structure.

We next correct the post-selection covariance between 
\(\hat{\gamma}_{j,\mathrm{RB}}\) and \(\hat{\Gamma}_{j,\mathrm{RB}}\). 
Even though \(\hat{\Gamma}_{j,\mathrm{RB}}\) is unbiased after selection, the product 
\(\hat{\gamma}_{j,\mathrm{RB}}\hat{\Gamma}_{j,\mathrm{RB}}\) still contains a post-selection covariance contribution:
\[
\mathbb{E}\left(
\hat{\gamma}_{j,\mathrm{RB}}\hat{\Gamma}_{j,\mathrm{RB}}
\mid j\in\mathcal S_\lambda
\right)
=
\gamma_j\Gamma_j
+
\sigma_{\hat{\gamma}_{j,\mathrm{RB}} \hat{\Gamma}_{j,\mathrm{RB}}},
\]
where
$\sigma_{\hat{\gamma}_{j,\mathrm{RB}} \hat{\Gamma}_{j,\mathrm{RB}}}$ is the post-selection covariance between \(\hat{\gamma}_{j,\mathrm{RB}}\) and \(\hat{\Gamma}_{j,\mathrm{RB}}\).
Thus, this covariance term must be removed from the numerator of the IVW-type estimator. 
It differs from the original covariance
\(\rho\sigma_{\hat{\gamma}_j}\sigma_{\hat{\Gamma}_j}\) because randomized IV selection and Rao--Blackwellization alter the joint post-selection distribution.
Using the conditional covariance decomposition, we obtain the following closed-form estimator:
\[
\hat{\sigma}_{\hat{\gamma}_{j,\mathrm{RB}} \hat{\Gamma}_{j,\mathrm{RB}}}
= \rho \, \sigma_{\hat{\gamma}_j} \, \sigma_{\hat{\Gamma}_j}
\Bigg[
1
- \frac{1}{\eta^{2}}
\frac{A_{j,+}\,\phi(A_{j,+}) - A_{j,-}\,\phi(A_{j,-})}
{1 - \Phi(A_{j,+}) + \Phi(A_{j,-})}
+ \frac{1}{\eta^{2}}
\left(
\frac{\phi(A_{j,+}) - \phi(A_{j,-})}
{1 - \Phi(A_{j,+}) + \Phi(A_{j,-})}
\right)^{2}
\Bigg].
\]
As shown in Web Appendix~B.1,
$\hat{\sigma}_{\hat{\gamma}_{j,\mathrm{RB}} \hat{\Gamma}_{j,\mathrm{RB}}}$
is unbiased for
\(\sigma_{\hat{\gamma}_{j,\mathrm{RB}}\hat{\Gamma}_{j,\mathrm{RB}}}\) conditional on the
selection event.
Its aggregate estimation error over the selected IVs is asymptotically negligible under the conditions stated in Web Appendix~B.2.

Combining the above components, we construct the BRIVW estimator as
\begin{equation*}
\hat{\beta}_{\mathrm{BRIVW}}
=
\frac{\sum_{j \in \mathcal{S}_\lambda}
\sigma_{\hat{\Gamma}_j}^{-2}\big(\hat{\Gamma}_{j,\mathrm{RB}}\hat{\gamma}_{j,\mathrm{RB}}-\hat{\sigma}_{\hat{\gamma}_{j,\mathrm{RB}} \hat{\Gamma}_{j,\mathrm{RB}}}\big)}
{\sum_{j \in \mathcal{S}_\lambda}
\sigma_{\hat{\Gamma}_j}^{-2}\big(\hat{\gamma}_{j,\mathrm{RB}}^{2}-\hat{\sigma}^{\,2}_{\hat{\gamma}_{j,\mathrm{RB}}}\big)}.
\end{equation*}
Here, we retain the conventional weights
\(\sigma_{\hat{\Gamma}_j}^{-2}\) because
weighting by the inverse estimated post-selection variance of
\(\hat{\Gamma}_{j,\mathrm{RB}}\) would introduce additional estimation uncertainty.
Compared with the classical IVW estimator, BRIVW jointly corrects for the winner's curse on both the exposure and outcome sides, sample-structure–induced covariance, and weak IV bias, while retaining a closed-form IVW-type estimator.
When $\rho=0$ and $c_1=c_2=1$ (i.e., no sample structure), the BRIVW estimator reduces to the RIVW estimator, and thus the former can be viewed as an extension of the latter under sample structure.

To estimate the variance of $\hat{\beta}_{\mathrm{BRIVW}}$, we consider the following decomposition:
\[
\hat{\beta}_{\mathrm{BRIVW}}
=
\beta
+
\frac{
\sum\nolimits_{j \in \mathcal{S}_\lambda} \sigma_{\hat{\Gamma}_j}^{-2} u_{j,\mathrm{BRIVW}}
}{
\sum\nolimits_{j \in \mathcal{S}_\lambda}
\sigma_{\hat{\Gamma}_j}^{-2}\big(\hat{\gamma}_{j,\mathrm{RB}}^{2}-\hat{\sigma}^{\,2}_{\hat{\gamma}_{j,\mathrm{RB}}}\big)
},
\]
where
\[
u_{j,\mathrm{BRIVW}}
=
\left(
\hat{\Gamma}_{j,\mathrm{RB}}\hat{\gamma}_{j,\mathrm{RB}}
-
\hat{\sigma}_{\hat{\gamma}_{j,\mathrm{RB}} \hat{\Gamma}_{j,\mathrm{RB}}}
\right)
-
\beta
\left(
\hat{\gamma}_{j,\mathrm{RB}}^{2}
-
\hat{\sigma}^{\,2}_{\hat{\gamma}_{j,\mathrm{RB}}}
\right).
\]
Web Appendix~C.1 shows that the denominator of the BRIVW estimator is asymptotically equivalent to its leading
term $\sum_{j\in\mathcal S_\lambda}\sigma_{\hat{\Gamma}_j}^{-2}\gamma_j^2$, 
so the dominant source of variability in $\hat{\beta}_{\mathrm{BRIVW}}$ arises from the numerator.
Under Assumptions~2 and 3, the terms \(u_{j,\mathrm{BRIVW}}\) are independent across SNPs and have conditional mean zero after selection.
This motivates estimating the conditional variance of the numerator by
the corresponding sum of squared residuals $\sum_{j \in \mathcal{S}_\lambda} \sigma_{\hat{\Gamma}_j}^{-4} u_{j,\mathrm{BRIVW}}^2$.
Replacing the unknown \(\beta\) with
\(\hat{\beta}_{\mathrm{BRIVW}}\) yields
\begin{equation*}
\widehat{V}_{\mathrm{BRIVW}}
=
\frac{\sum_{j \in \mathcal{S}_\lambda}
\sigma_{\hat{\Gamma}_j}^{-4} \big[
\hat{\Gamma}_{j,\mathrm{RB}} \hat{\gamma}_{j,\mathrm{RB}}
-
\hat{\sigma}_{\hat{\gamma}_{j,\mathrm{RB}} \hat{\Gamma}_{j,\mathrm{RB}}}
-
\hat{\beta}_{\mathrm{BRIVW}}
\big(\hat{\gamma}_{j,\mathrm{RB}}^{2}
-
\hat{\sigma}^{\,2}_{\hat{\gamma}_{j,\mathrm{RB}}}\big)
\big]^{2}}
{\left[
\sum_{j \in \mathcal{S}_\lambda}
\sigma_{\hat{\Gamma}_j}^{-2}\big(\hat{\gamma}_{j,\mathrm{RB}}^{2}
-
\hat{\sigma}^{\,2}_{\hat{\gamma}_{j,\mathrm{RB}}}\big)
\right]^{2}}.
\end{equation*}
Under the balanced pleiotropy and InSIDE conditions, 
pleiotropic effects contribute only to the residual variation. 
The residual-based variance estimator therefore remains valid without modification (Ma, Wang, and Wu 2023).

We next establish the asymptotic properties of the BRIVW estimator.
Let \(p_\lambda=|\mathcal S_\lambda|\) denote the number of selected IVs, and define the post-selection instrument strength as \[ \kappa_\lambda = \frac{1}{p_\lambda} \sum_{j\in\mathcal S_\lambda} \left( \frac{\gamma_j}{\sigma_{\hat{\gamma}_j}} \right)^2 . \] 
For the theoretical analysis, we impose the following additional assumption.

\begin{assumption}
The true instrument effects satisfy
\[
\frac{
\max_{j \in \mathcal{S}_\lambda} \gamma_j^2
}{
\sum_{j \in \mathcal{S}_\lambda} \gamma_j^2
}
\;\xrightarrow{p}\; 0.
\]
\end{assumption}
Assumption~4 rules out dominance by one or a few selected instruments and is used to establish asymptotic normality.

\begin{theorem}
Suppose that Assumptions~1--3 hold,
\(p_\lambda\xrightarrow{p}\infty\), and
$
\kappa_\lambda\sqrt{p_\lambda}
\xrightarrow{p}\infty.
$
Then, conditional on the selected instrument set
\(\mathcal S_\lambda\),
\[
\hat{\beta}_{\mathrm{BRIVW}}
\xrightarrow{p}
\beta.
\]
If Assumption~4 also holds, then
\[
\widehat V_{\mathrm{BRIVW}}^{-1/2}
\left(
\hat{\beta}_{\mathrm{BRIVW}}-\beta
\right)
\xrightarrow{d}
N(0,1).
\]
\end{theorem}
\noindent
The proof of Theorem 1 is given in Web Appendix~C.
Compared with the RIVW estimator, the BRIVW estimator is established under substantially weaker
asymptotic conditions: \(\lambda\to\infty\) is no longer required, and
the condition
\(\kappa_\lambda/\lambda^2\xrightarrow{p}\infty\) is replaced by
\(
\kappa_\lambda\sqrt{p_\lambda}\xrightarrow{p}\infty.
\)
The first relaxation follows from Web Appendix~A.3, which establishes
uniform upper and lower bounds for the Rao--Blackwellized variances
for fixed \(\lambda\). For the second relaxation, the orthogonal Gaussian
representation in Web Appendix~B.2 yields uniform fixed-order moment
bounds and sharpens the aggregate estimation errors of the variance
and covariance corrections to
\(O_p(v^2\sqrt{p_\lambda})\), with no dependence on \(\lambda\).
These bounds allow the higher-order Rao--Blackwellized terms to be
incorporated directly into the asymptotic analysis, leading to the
weaker strength condition above. 
Notably, this condition can accommodate many weak
instrument scenarios: \(\kappa_\lambda\) need not diverge and may even tend to
zero, provided that \(p_\lambda\) grows sufficiently fast (Ye, Shao, and Kang 2021).
The theorem establishes the consistency and asymptotic normality of the BRIVW estimator.
Therefore, a two-sided $100(1-\alpha)\%$ confidence interval (CI) for $\beta$ can be constructed as
 \[
\left[
\hat{\beta}_{\mathrm{BRIVW}}
-
\Phi^{-1}(1-\alpha/2)\sqrt{\widehat{V}_{\mathrm{BRIVW}}},
\;
\hat{\beta}_{\mathrm{BRIVW}}
+
\Phi^{-1}(1-\alpha/2)\sqrt{\widehat{V}_{\mathrm{BRIVW}}}
\right].
\]

\section{Simulation Studies}

\subsection{Simulation settings}

In the main simulations, we set \(\beta=0.2\).
The number of independent SNPs is fixed at \(p=200{,}000\).
The parameters $(\gamma_j, \alpha_j)$ are generated from the following mixture distribution:
\[
\begin{aligned}
\begin{pmatrix}
\gamma_j\\
\alpha_j
\end{pmatrix}
\sim\;&
\pi_x(1-w_1)
\begin{pmatrix}
N(0,\varepsilon_x^2)\\
\delta_0
\end{pmatrix}
+
\pi_x w_1
\begin{pmatrix}
N(0,\varepsilon_x^2)\\
N(0,\tau^2)
\end{pmatrix}
+
\pi_y
\begin{pmatrix}
\delta_0\\
N(0,\tau^2)
\end{pmatrix}
+
(1-\pi_x-\pi_y)
\begin{pmatrix}
\delta_0\\
\delta_0
\end{pmatrix},
\end{aligned}
\]
where \(\delta_0\) denotes the Dirac measure centered at zero.
The four mixture components represent valid IVs, balanced pleiotropic IVs, outcome-only IVs, and null IVs, respectively. 
We then set
$
\Gamma_j=\beta\gamma_j+\alpha_j
$
and generate the GWAS summary statistics from the bivariate normal model
in Assumption~2.
Here, \(\rho\) controls the sample-structure-induced correlation between the estimation errors of
\(\hat{\gamma}_j\) and \(\hat{\Gamma}_j\).
Throughout this section, we set \(\rho\in\{-0.3,-0.1,0,0.1,0.3\}\), \(n_X=n_Y=100{,}000\), 
\(\sigma_{\hat{\gamma}_j}=1/\sqrt{n_X}\), 
\(\sigma_{\hat{\Gamma}_j}=1/\sqrt{n_Y}\), 
\(\varepsilon_x^2=\tau^2=5\times10^{-5}\), 
\(\pi_x=0.02\), and \(\pi_y=0.01\).

We consider \(w_1\in\{0,0.5,1\}\) to vary the proportion of balanced pleiotropic IVs.
We compare the BRIVW estimator with four competing methods: IVW (Didelez and Sheehan 2007), RIVW (Ma, Wang, and Wu 2023), MR-Egger (Bowden, Davey Smith, and Burgess 2015), and MR-APSS (Hu et al. 2022).
For IV selection, IVW and MR-Egger use the conventional genome-wide
significance threshold \(\lambda=\Phi^{-1}(1-5\times10^{-8}/2)=5.45\), whereas MR-APSS uses its default
screening threshold \(\lambda=\Phi^{-1}(1-5\times10^{-5}/2)=4.06\).
Following Ma, Wang, and Wu (2023), RIVW and BRIVW use randomized selection
with \(\lambda=4.06\) and \(\eta=0.5\).
Additional simulations in Web Appendix~D.2 show that BRIVW is also robust to the choice of $\eta$.
All methods are evaluated in terms of empirical bias, standard deviation (SD), and coverage probability (CP) for the nominal \(95\%\) CI.
Because MR-APSS is computationally intensive,
the simulation results are based on 2,000 replicates.

\subsection{Simulation results}

\begin{figure}[!htbp]
    \centering
    \includegraphics[width=0.95\textwidth]{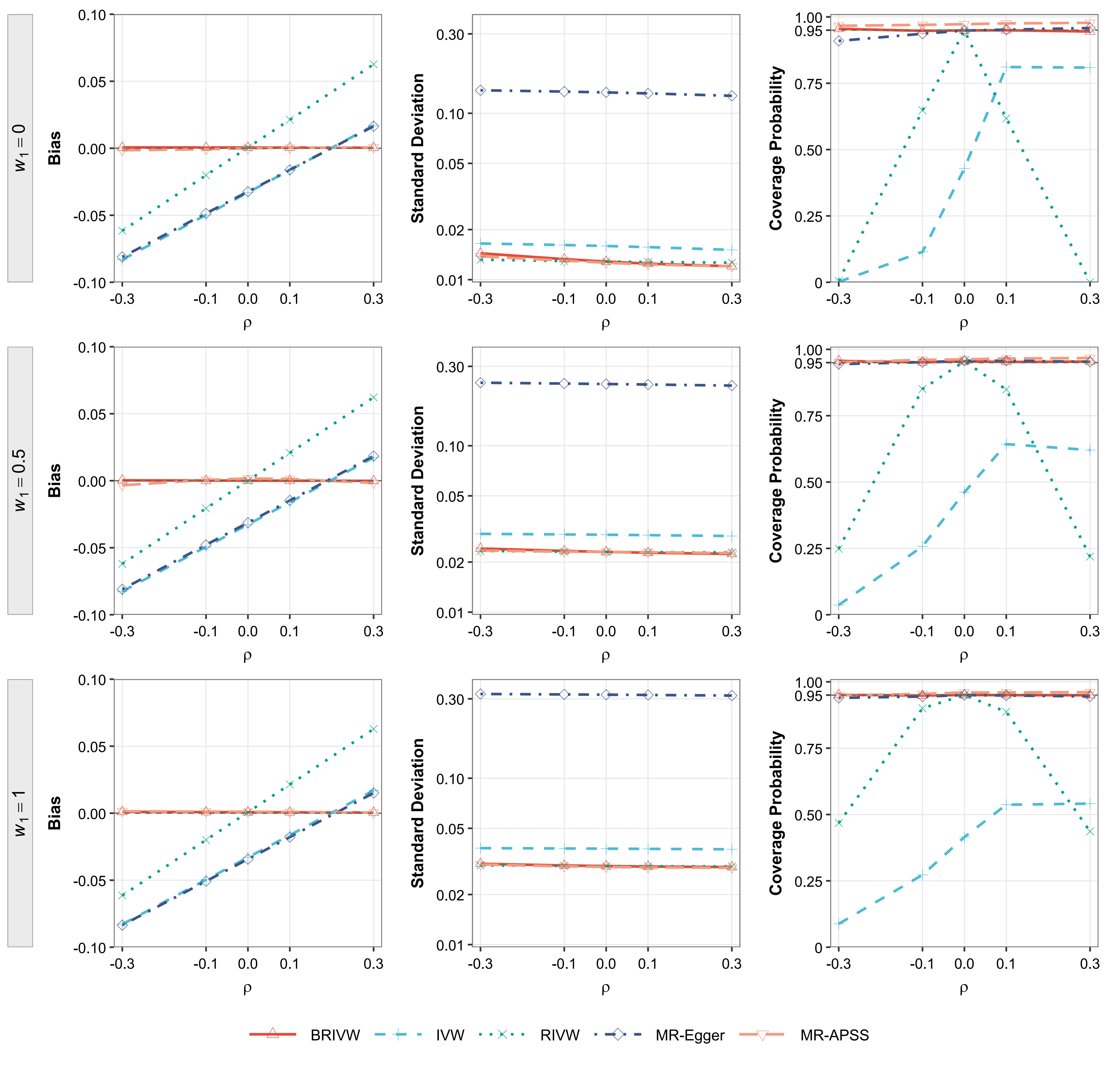}
    \caption{\small
    Empirical bias, standard deviation, and coverage probability of BRIVW
    and the competing methods across various values of \(\rho\) and \(w_1\).
    }
    \label{fig1}
\end{figure}

Figure~1 summarizes the simulation results across different levels of sample structure and balanced pleiotropy. 
In the absence of horizontal pleiotropy (i.e., \(w_1=0\)), the BRIVW estimator remains unbiased over the entire range of \(\rho\), with small SDs and accurate CPs.
The RIVW estimator has the same performance as the BRIVW estimator when \(\rho=0\), but its bias increases in magnitude as \(\rho\) departs from zero, leading to substantial undercoverage.
The IVW and MR-Egger estimators exhibit downward bias at \(\rho=0\) because of the weak IV problem and winner's curse.
As \(\rho\) varies, negative values of $\rho$ reinforce this bias, whereas positive values partially offset it and may even reverse its direction. 
Thus, the improved performance of IVW and MR-Egger estimators under positive \(\rho\) reflects bias cancellation rather than robustness.
With increasing levels of balanced pleiotropy (i.e., \(w_1=0.5\) and \(1\)), the SDs of all estimators increase, while the overall performance patterns remain essentially unchanged.
Although MR-APSS performs comparably to the BRIVW estimator under the normal-mixture model, its performance deteriorates when the mixture distribution is misspecified (see Web Appendix~D.5).
Additional sensitivity analyses with different causal effect sizes and unequal GWAS sample sizes yield similar overall conclusions (see Web Appendices~D.3-D.4).

\section{Real Data Applications}

We conduct three real data analyses. 
The negative control outcome analysis assesses the empirical type I error of BRIVW for causal-effect testing. 
The same trait analysis serves as a benchmark for estimation accuracy, with the target effect equal to 1. 
Finally, we apply BRIVW to investigate the causal effects of BMI on cardiovascular outcomes.

\subsection{Data processing}
For each exposure--outcome pair, we first apply the quality-control procedures of Hu et al.~(2022) to both GWAS datasets (see Web Appendix~E.1).
The GWAS summary statistics are then harmonized to align effect alleles, and the sample-structure parameters \((c_1,c_2,c_{12})\) are estimated by LDSC.
For BRIVW, these estimates are used to calibrate the marginal standard errors and cross-trait correlation.
We select approximately independent SNPs using the revised sigma-based pruning procedure (Ma, Wang, and Wu 2023) to avoid selection bias, with pairwise LD \(R^2<0.001\) within 10,000 kb.
Among correlated SNPs, the variant with the smallest SNP--exposure standard error is retained.
The IV selection thresholds for all methods are the same as in the simulation studies.

\subsection{Negative control outcome analysis}
To assess the type I error control of BRIVW, we use the negative control outcomes proposed by Sanderson et al. (2021), for which no causal effects are expected. 
We consider five UK Biobank outcomes—four categories of natural hair color before graying (black, blonde, light brown, and dark brown) and skin tanning ability—each paired with 53 exposure traits, yielding 265 exposure-outcome pairs (see Web Table~S2).

\begin{figure}[t]
  \centering
  \includegraphics[width=\linewidth]{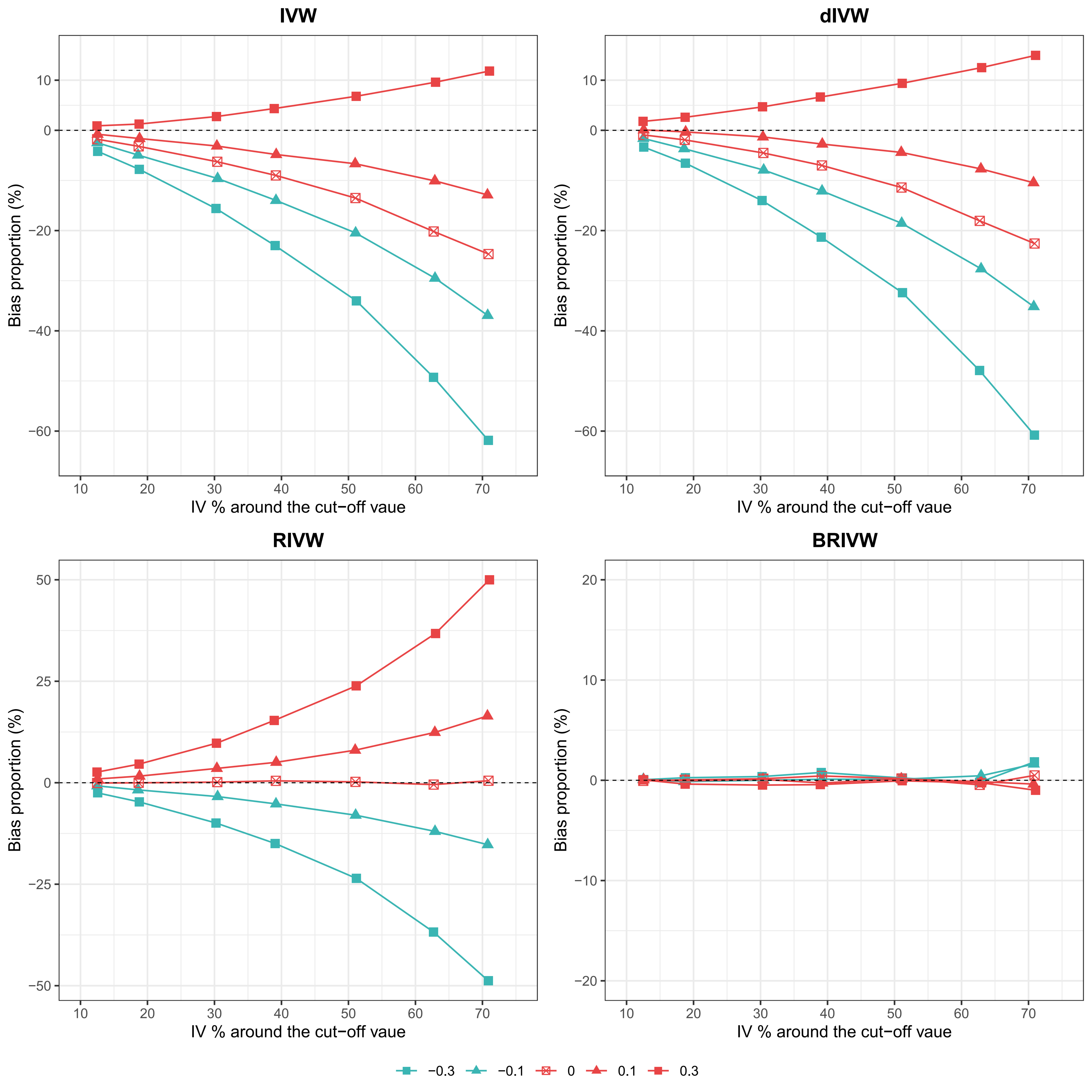}
  \caption{\small
  Negative control outcome analyses.
  QQ-plots of causal-effect \(P\)-values from BRIVW and competing MR methods
  for 265 exposure--outcome pairs. Red dots represent all pairs and blue
  triangles the subset with nonsignificant \(\rho\).
  }
  \label{fig2}
\end{figure}

Figure~2 shows the QQ-plots of causal-effect \(P\)-values from BRIVW and the competing methods.
BRIVW yields well-calibrated \(P\)-values across the 265 negative-control pairs.
In contrast, IVW and RIVW show substantial inflation, whereas MR-Egger and MR-APSS remain broadly well calibrated.
Of the 265 pairs, 172 (64.9\%) have \(\rho\) values significantly different from zero (see Web Figure~S10), and restricting the analysis to pairs with nonsignificant \(\rho\) markedly reduces the inflation, suggesting that sample structure contributes materially to the distortion.

\subsection{Same trait analyses}
We next follow the same trait design of Ma, Wang, and Wu~(2023) to benchmark estimation accuracy, using the same body mass index (BMI) and high-density lipoprotein cholesterol (HDL) GWAS datasets.
Each trait is analyzed in both directions, yielding four exposure--outcome configurations (see Web Table~S3).

\begin{figure}[t]
  \centering
  \includegraphics[width=\linewidth]{Fig3.pdf}
  \caption{\small Same trait results for different methods. The vertical reference line denotes the benchmark effect of 1.}
  \label{fig3}
\end{figure}

Figure~3 reports the estimated causal effects and two-sided 95\% CIs across methods.
Across all four configurations, BRIVW yields estimates close to 1 with relatively narrow 95\% CIs.
In contrast, IVW systematically attenuates the estimates toward zero, as expected in the presence of the weak IV problem and winner's curse.
RIVW reduces this attenuation, but its estimates exceed 1 in some configurations (e.g., BMI-2 $\rightarrow$ BMI-1), suggesting that residual sample structure may still affect the estimates, as indicated by the nonzero \(\rho\) (see Web Table~S4).
MR-Egger is less precise, while MR-APSS deviates from the benchmark in some configurations.

\subsection{BMI effects on cardiovascular outcomes}
Elevated BMI is an established modifiable risk factor for cardiovascular disease, supported by extensive epidemiologic and Mendelian randomization evidence (Prospective Studies Collaboration 2009; Holmes et al. 2014; Lincoff et al. 2023). 
Although the causal contribution of elevated BMI to cardiovascular risk is well established, accurately quantifying the magnitude of these effects remains important for clinical and public health interpretation.
We therefore apply the BRIVW estimator to estimate the causal effects of BMI
(Bycroft et al. 2018; ukb-b-19953) on systolic blood pressure
(SBP; Bycroft et al. 2018; ukb-b-20175), diastolic blood pressure
(DBP; Bycroft et al. 2018; ukb-b-7992), and coronary heart disease
(CHD; CARDIoGRAMplusC4D Consortium et al. 2013; ieu-a-9).
The BMI, SBP, and DBP GWASs are all based on UK Biobank, a widely used resource for MR analyses. 
Because large biobank cohorts are often reused across multiple GWASs, participant overlap can arise between exposure and outcome samples, making sample structure a non-negligible consideration in MR analyses.

Table~1 summarizes the analysis results.
For SBP and DBP, the BRIVW estimator indicates positive effects of BMI, with estimates of 0.106 (95\% CI: 0.059, 0.154) and 0.199 (95\% CI: 0.147, 0.251), respectively. 
The corresponding \(\rho\) values are 0.199 and 0.278, both highly significant, 
and the RIVW estimator yields substantially larger estimates than the BRIVW estimator, 
suggesting that sample structure can materially affect causal-effect estimation.
For CHD, the BRIVW estimator also indicates a positive effect of BMI (\(\hat\beta=0.395\), 95\% CI: 0.267, 0.524), 
while \(\rho\) is lower
(\(\rho=0.070\)) and the BRIVW and RIVW estimators are correspondingly closer.

\begin{table}[t]
\centering

\captionsetup{
  labelfont=bf,
  labelsep=colon,
  skip=2pt
}

\scriptsize
\setlength{\tabcolsep}{4pt}
\renewcommand{\arraystretch}{1.08}

\caption{Causal-effect estimates of BMI on cardiovascular outcomes.}
\label{tab:bmi_cv_results}

\begin{threeparttable}

\begin{tabular*}{\textwidth}{@{\extracolsep{\fill}}cccccc@{}}

\toprule

Trait pair &
\(\rho\) (\(P\)) &
Method &
IVs &
\(\hat\beta\) (SE) &
95\% CI \\

\midrule

\multirow{5}{*}{BMI \(\rightarrow\) SBP}
& \multirow{5}{*}{0.199 (\(1.93\times10^{-119}\))}
& BRIVW    & 767 & 0.106 (0.024) & (0.059, 0.154) \\

& & RIVW     & 863 & 0.186 (0.024) & (0.140, 0.232) \\

& & IVW      & 422 & 0.129 (0.008) & (0.113, 0.144) \\

& & MR-Egger & 422 & 0.117 (0.078) & (-0.036, 0.270) \\

& & MR-APSS  & 771 & 0.061 (0.057) & (-0.050, 0.172) \\

\midrule

\multirow{5}{*}{BMI \(\rightarrow\) DBP}
& \multirow{5}{*}{0.278 (\(4.41\times10^{-171}\))}
& BRIVW    & 767 & 0.199 (0.027) & (0.147, 0.251) \\

& & RIVW     & 863 & 0.296 (0.026) & (0.246, 0.346) \\

& & IVW      & 422 & 0.207 (0.008) & (0.191, 0.222) \\

& & MR-Egger & 422 & 0.174 (0.082) & (0.013, 0.336) \\

& & MR-APSS  & 771 & 0.151 (0.060) & (0.032, 0.269) \\

\midrule

\multirow{5}{*}{BMI \(\rightarrow\) CHD}
& \multirow{5}{*}{0.070 (0.015)}
& BRIVW    & 402 & 0.395 (0.066) & (0.267, 0.524) \\

& & RIVW     & 436 & 0.419 (0.065) & (0.292, 0.546) \\

& & IVW      & 252 & 0.469 (0.051) & (0.369, 0.569) \\

& & MR-Egger & 252 & 1.104 (0.259) & (0.596, 1.611) \\

& & MR-APSS  & 413 & 0.215 (0.051) & (0.114, 0.315) \\

\bottomrule

\end{tabular*}

\begin{tablenotes}[flushleft]
\footnotesize
\item Note:
\(P\) represents the \(P\)-value of the test for \(\rho=0\).
IVs denote the number of selected instrumental variables.
\end{tablenotes}

\end{threeparttable}

\end{table}

\section{Discussion}
In this article, we have developed the BRIVW estimator to address the bivariate post-selection problem that arises in two-sample MR when sample structure induces correlation between SNP--exposure and SNP--outcome association estimates.
BRIVW uses bivariate Rao--Blackwellization to remove post-selection bias while preserving the closed form of an IVW-type estimator.
Although MR-APSS also addresses the winner's curse and sample structure, it relies on specific model assumptions and computationally intensive variational inference.
Simulation studies and real data analyses also support the advantages of the BRIVW estimator.

Notably, the RIVW estimator is a special case of the BRIVW estimator, because the latter reduces to the former in the absence of sample structure. 
From an applied perspective, BRIVW is particularly suited to modern MR analyses involving highly polygenic exposures and large GWAS datasets with residual sample structure. 
Its bias correction permits less stringent IV-selection thresholds, allowing more variants to be retained when genome-wide significant instruments are scarce, as is often the case for highly polygenic traits such as psychiatric disorders (Sullivan et al. 2018). 
Moreover, its closed-form expression is computationally efficient and readily scalable to large-scale analyses involving many exposure-outcome pairs while accounting for sample overlap and shared population stratification that commonly arise in biobank and consortium GWASs.

Several limitations warrant further investigation. 
First, our method relies on LDSC-based estimates of the sample-structure parameters \((c_1,c_2,c_{12})\), and its performance may be affected when these estimates are unstable or LDSC assumptions are violated. 
Future work could develop alternative sample-structure estimators or procedures that are less sensitive to LDSC inputs. 
Second, as an IVW-type estimator, BRIVW relies on balanced pleiotropy and the InSIDE conditions, and remains vulnerable to directional or correlated pleiotropy. 
Integrating the BRIVW estimator with pleiotropy-robust methods, such as CARE (Xie et al. 2026), is therefore an important direction for future work.
Finally, while the current framework focuses on two-sample MR, extending it to multivariable MR represents a natural direction for future research.

\section*{Acknowledgements}
We gratefully acknowledge the participants and investigators of the UK Biobank and the GWAS consortia that made the summary statistics used in this study publicly available.

% ============================================================================
%                        SUPPORTING INFORMATION
% ============================================================================
\clearpage

% ---- restore the supplement's own page geometry (6.5in x 9.0in), so the
%      appendices are typeset exactly as in the standalone supplement ----
\setlength{\textwidth}{6.5in}
\setlength{\textheight}{9.0in}
\setlength{\oddsidemargin}{0in}
\setlength{\evensidemargin}{0in}
\setlength{\topmargin}{0in}
\setlength{\hsize}{\textwidth}
\setlength{\linewidth}{\textwidth}
\setlength{\columnwidth}{\textwidth}

\spacingset{1.2}

% ---- switch numbering over to the supplementary material ----
\appendix
\setcounter{figure}{0}
\setcounter{table}{0}
\setcounter{lemma}{0}
\setcounter{theorem}{0}
\setcounter{assumption}{0}

\makeatletter
\@addtoreset{lemma}{section}
\@addtoreset{theorem}{section}
\@addtoreset{assumption}{section}
\@addtoreset{equation}{section}
% reproduce the supplement's indentfirst behaviour from here on only
\let\@afterindentfalse\@afterindenttrue
\@afterindenttrue
\makeatother

\renewcommand{\thelemma}{\thesection.\arabic{lemma}}
\renewcommand{\thetheorem}{\thesection.\arabic{theorem}}
\renewcommand{\theassumption}{\thesection.\arabic{assumption}}
\renewcommand{\theequation}{\Alph{section}\arabic{equation}}
\renewcommand{\thetable}{\Alph{section}\arabic{table}}
\renewcommand{\thefigure}{\Alph{section}\arabic{figure}}

\titleformat{\section}
  {\normalsize\bfseries}
  {Appendix~\Alph{section}.}
  {1em}{}
\titleformat{\subsection}
  {\normalsize\bfseries}
  {\Alph{section}.\arabic{subsection}}
  {1em}{}

\begin{center}
{\bf\large Supporting Information for ``Simultaneously accounting for the
winner's curse and sample structure in Mendelian randomization:
bivariate rerandomized inverse variance weighted estimator''}\\[1.2em]
Xin Liu,$\rm{^{a,\dagger}}$ Youpeng Su,$\rm{^{a,\dagger}}$ Chuchu Zeng,$\rm{^{a}}$ Yilei Ma,$\rm{^{b}}$\\
Siqi Xu,$\rm{^{c}}$ Ping Yin,$\rm{^{a}}$ and Peng Wang$\rm{^{a,*}}$\\[0.8em]
$\rm{^{a}}$Department of Epidemiology and Biostatistics, School of Public Health,\\
Tongji Medical College, Huazhong University of Science and Technology,\\
Wuhan, China\\[0.35em]
$\rm{^{b}}$Tongji Hospital, Tongji Medical College,\\
Huazhong University of Science and Technology, Wuhan, China\\[0.35em]
$\rm{^{c}}$Clinical Research Centre, Guangzhou First People's Hospital,\\
Guangzhou Medical University, Guangzhou, China\\[0.8em]
$^\dagger$These authors contributed equally to this work.\\
$^*$Correspondence: pengwang\_stat@hust.edu.cn
\end{center}

\vspace{1em}

% from here on the table of contents lists entries again
\addtocontents{toc}{\protect\setcounter{tocdepth}{2}}
\setcounter{tocdepth}{2}
\tableofcontents
\clearpage

\section{Bivariate post-selection bias and Rao--Blackwellization}
\label{sec:app_bivariate_RB}

\subsection{Post-selection distribution and bias of the outcome-association estimate}
\label{sec:app_bias_propagation}

We first characterize the conditional distribution of the outcome-association estimate after randomized IV selection. This result shows how the post-selection bias induced in \(\hat{\gamma}_j\) propagates to \(\hat{\Gamma}_j\) through the sampling-error correlation.

\begin{lemma}
\label{lem:conditional_Gamma}
Suppose that Assumptions~1--2 hold. Then the conditional density of
\(\hat{\Gamma}_j/\sigma_{\hat{\Gamma}_j}\) given
\(j\in\mathcal S_\lambda\) is
\[
\begin{aligned}
&f_{\hat{\Gamma}_j/\sigma_{\hat{\Gamma}_j}\mid j\in\mathcal S_\lambda}(y) \\
&=
\frac{\phi\left(y-\dfrac{\Gamma_j}{\sigma_{\hat{\Gamma}_j}}\right)}
{\Pr(j\in\mathcal S_\lambda)}
\Bigg[
\Phi\Bigg(
\frac{-\lambda-\dfrac{\gamma_j}{\sigma_{\hat{\gamma}_j}}
+\dfrac{\rho\Gamma_j}{\sigma_{\hat{\Gamma}_j}}-\rho y}
{\sqrt{1-\rho^2+\eta^2}}
\Bigg)
+1-\Phi\Bigg(
\frac{\lambda-\dfrac{\gamma_j}{\sigma_{\hat{\gamma}_j}}
+\dfrac{\rho\Gamma_j}{\sigma_{\hat{\Gamma}_j}}-\rho y}
{\sqrt{1-\rho^2+\eta^2}}
\Bigg)
\Bigg],
\end{aligned}
\]
where
\[
\Pr(j\in\mathcal S_\lambda)
=
\Phi\Bigg(
\frac{-\lambda-\dfrac{\gamma_j}{\sigma_{\hat{\gamma}_j}}}
{\sqrt{1+\eta^2}}
\Bigg)
+
1-\Phi\Bigg(
\frac{\lambda-\dfrac{\gamma_j}{\sigma_{\hat{\gamma}_j}}}
{\sqrt{1+\eta^2}}
\Bigg).
\]
Moreover,
\[
\begin{aligned}
\mathbb E\left(\hat{\Gamma}_j\mid j\in\mathcal S_\lambda\right)
&=
\Gamma_j+
\frac{\rho\sigma_{\hat{\Gamma}_j}}
{\sqrt{1+\eta^2}\Pr(j\in\mathcal S_\lambda)}
\Bigg[
\phi\Bigg(
\frac{\lambda-\dfrac{\gamma_j}{\sigma_{\hat{\gamma}_j}}}
{\sqrt{1+\eta^2}}
\Bigg)
-
\phi\Bigg(
\frac{\lambda+\dfrac{\gamma_j}{\sigma_{\hat{\gamma}_j}}}
{\sqrt{1+\eta^2}}
\Bigg)
\Bigg].
\end{aligned}
\]
\end{lemma}

\begin{proof}
Consider the joint probability
\[
\begin{aligned}
&\Pr\left(
\frac{\hat{\Gamma}_j}{\sigma_{\hat{\Gamma}_j}}\leq y,\,
j\in\mathcal S_\lambda
\right)
=
\Pr\left(
\frac{\hat{\Gamma}_j}{\sigma_{\hat{\Gamma}_j}}\leq y,\,
\left|
\frac{\hat{\gamma}_j}{\sigma_{\hat{\gamma}_j}}+Z_j
\right|>\lambda
\right).
\end{aligned}
\]
Conditioning on \(Z_j=w\), we have
\[
\begin{aligned}
&\Pr\left(
\frac{\hat{\Gamma}_j}{\sigma_{\hat{\Gamma}_j}}\leq y,\,
j\in\mathcal S_\lambda
\right) \\
&\quad=
\int_{-\infty}^{\infty}
\Bigg[
\Pr\left(
\frac{\hat{\Gamma}_j}{\sigma_{\hat{\Gamma}_j}}\leq y,\,
\frac{\hat{\gamma}_j}{\sigma_{\hat{\gamma}_j}}>\lambda-w
\right)
+
\Pr\left(
\frac{\hat{\Gamma}_j}{\sigma_{\hat{\Gamma}_j}}\leq y,\,
\frac{\hat{\gamma}_j}{\sigma_{\hat{\gamma}_j}}<-\lambda-w
\right)
\Bigg]
\frac{1}{\eta}\phi\left(\frac{w}{\eta}\right)\,dw.
\end{aligned}
\]

Let \(U\) and \(V\) be independent standard normal random variables. Under Assumption~2, we can write
\[
\frac{\hat{\Gamma}_j-\Gamma_j}{\sigma_{\hat{\Gamma}_j}}=U,
\qquad
\frac{\hat{\gamma}_j-\gamma_j}{\sigma_{\hat{\gamma}_j}}
=
\rho U+\sqrt{1-\rho^2}\,V.
\]
It follows that
\[
\begin{aligned}
&\Pr\left(
\frac{\hat{\Gamma}_j}{\sigma_{\hat{\Gamma}_j}}\leq y,\,
j\in\mathcal S_\lambda
\right) \\
&=
\int_{-\infty}^{y-\Gamma_j/\sigma_{\hat{\Gamma}_j}}
\phi(u)
\int_{-\infty}^{\infty}
\Bigg[
\Phi\Bigg(
\frac{-\lambda-w-\dfrac{\gamma_j}{\sigma_{\hat{\gamma}_j}}-\rho u}
{\sqrt{1-\rho^2}}
\Bigg)
+
1-\Phi\Bigg(
\frac{\lambda-w-\dfrac{\gamma_j}{\sigma_{\hat{\gamma}_j}}-\rho u}
{\sqrt{1-\rho^2}}
\Bigg)
\Bigg]
\frac{1}{\eta}\phi\left(\frac{w}{\eta}\right)\,dw\,du.
\end{aligned}
\]
Using the identity
\[
\int_{-\infty}^{\infty}\phi(x)\Phi(a+bx)\,dx
=
\Phi\left(\frac{a}{\sqrt{1+b^2}}\right),
\]
we obtain
\[
\begin{aligned}
&\int_{-\infty}^{\infty}
\frac{1}{\eta}\phi\left(\frac{w}{\eta}\right)
\Phi\Bigg(
\frac{-\lambda-w-\dfrac{\gamma_j}{\sigma_{\hat{\gamma}_j}}-\rho u}
{\sqrt{1-\rho^2}}
\Bigg)\,dw =
\Phi\Bigg(
\frac{-\lambda-\dfrac{\gamma_j}{\sigma_{\hat{\gamma}_j}}-\rho u}
{\sqrt{1-\rho^2+\eta^2}}
\Bigg),
\end{aligned}
\]
and
\[
\begin{aligned}
&\int_{-\infty}^{\infty}
\frac{1}{\eta}\phi\left(\frac{w}{\eta}\right)
\Phi\Bigg(
\frac{\lambda-w-\dfrac{\gamma_j}{\sigma_{\hat{\gamma}_j}}-\rho u}
{\sqrt{1-\rho^2}}
\Bigg)\,dw =
\Phi\Bigg(
\frac{\lambda-\dfrac{\gamma_j}{\sigma_{\hat{\gamma}_j}}-\rho u}
{\sqrt{1-\rho^2+\eta^2}}
\Bigg).
\end{aligned}
\]
Therefore,
\[
\begin{aligned}
&\Pr\left(
\frac{\hat{\Gamma}_j}{\sigma_{\hat{\Gamma}_j}}\leq y,\,
j\in\mathcal S_\lambda
\right) =
\int_{-\infty}^{y-\Gamma_j/\sigma_{\hat{\Gamma}_j}}
\phi(u)
\Bigg[
\Phi\Bigg(
\frac{-\lambda-\dfrac{\gamma_j}{\sigma_{\hat{\gamma}_j}}-\rho u}
{\sqrt{1-\rho^2+\eta^2}}
\Bigg)
+
1-\Phi\Bigg(
\frac{\lambda-\dfrac{\gamma_j}{\sigma_{\hat{\gamma}_j}}-\rho u}
{\sqrt{1-\rho^2+\eta^2}}
\Bigg)
\Bigg]\,du.
\end{aligned}
\]

Next, we compute the derivatives with respect to \(y\), which gives
\[
\begin{aligned}
&\frac{d}{dy}
\Pr\left(
\frac{\hat{\Gamma}_j}{\sigma_{\hat{\Gamma}_j}}\leq y,\,
j\in\mathcal S_\lambda
\right) \\
&=
\phi\left(y-\frac{\Gamma_j}{\sigma_{\hat{\Gamma}_j}}\right)
\Bigg[
\Phi\Bigg(
\frac{-\lambda-\dfrac{\gamma_j}{\sigma_{\hat{\gamma}_j}}
+\dfrac{\rho\Gamma_j}{\sigma_{\hat{\Gamma}_j}}-\rho y}
{\sqrt{1-\rho^2+\eta^2}}
\Bigg)
+
1-\Phi\Bigg(
\frac{\lambda-\dfrac{\gamma_j}{\sigma_{\hat{\gamma}_j}}
+\dfrac{\rho\Gamma_j}{\sigma_{\hat{\Gamma}_j}}-\rho y}
{\sqrt{1-\rho^2+\eta^2}}
\Bigg)
\Bigg].
\end{aligned}
\]
Since
\[
\frac{\hat{\gamma}_j}{\sigma_{\hat{\gamma}_j}}+Z_j
\sim
N\left(
\frac{\gamma_j}{\sigma_{\hat{\gamma}_j}},
1+\eta^2
\right),
\]
the probability of selection is
\[
\Pr(j\in\mathcal S_\lambda)
=
\Phi\Bigg(
\frac{-\lambda-\dfrac{\gamma_j}{\sigma_{\hat{\gamma}_j}}}
{\sqrt{1+\eta^2}}
\Bigg)
+
1-\Phi\Bigg(
\frac{\lambda-\dfrac{\gamma_j}{\sigma_{\hat{\gamma}_j}}}
{\sqrt{1+\eta^2}}
\Bigg).
\]
Dividing the derivative above by \(\Pr(j\in\mathcal S_\lambda)\) gives the stated conditional density.

We next calculate the conditional expectation directly from this density:
\[
\mathbb E\left(
\left.
\frac{\hat{\Gamma}_j}{\sigma_{\hat{\Gamma}_j}}
\,\right|\,
j\in\mathcal S_\lambda
\right)
=
\int_{-\infty}^{\infty}
y\,
f_{\hat{\Gamma}_j/\sigma_{\hat{\Gamma}_j}\mid j\in\mathcal S_\lambda}(y)\,dy.
\]
Making the change of variable
\[
x=y-\frac{\Gamma_j}{\sigma_{\hat{\Gamma}_j}},
\]
we obtain
\[
\begin{aligned}
&\mathbb E\left(
\left.
\frac{\hat{\Gamma}_j}{\sigma_{\hat{\Gamma}_j}}
\,\right|\,
j\in\mathcal S_\lambda
\right) \\
&=
\frac{\Gamma_j}{\sigma_{\hat{\Gamma}_j}}
+
\frac{1}{\Pr(j\in\mathcal S_\lambda)}
\int_{-\infty}^{\infty}
x\phi(x)
\Bigg[
\Phi\Bigg(
\frac{-\lambda-\dfrac{\gamma_j}{\sigma_{\hat{\gamma}_j}}-\rho x}
{\sqrt{1-\rho^2+\eta^2}}
\Bigg)
+
1-\Phi\Bigg(
\frac{\lambda-\dfrac{\gamma_j}{\sigma_{\hat{\gamma}_j}}-\rho x}
{\sqrt{1-\rho^2+\eta^2}}
\Bigg)
\Bigg]\,dx.
\end{aligned}
\]
For any \(a\in\mathbb R\), integration by parts gives
\[
\int_{-\infty}^{\infty}
x\phi(x)
\Phi\Bigg(
\frac{a-\rho x}{\sqrt{1-\rho^2+\eta^2}}
\Bigg)\,dx
=
-\frac{\rho}{\sqrt{1+\eta^2}}
\phi\left(
\frac{a}{\sqrt{1+\eta^2}}
\right).
\]
Since \(\int_{-\infty}^{\infty}x\phi(x)\,dx=0\), it follows that
\[
\begin{aligned}
&\int_{-\infty}^{\infty}
x\phi(x)
\Bigg[
\Phi\Bigg(
\frac{-\lambda-\dfrac{\gamma_j}{\sigma_{\hat{\gamma}_j}}-\rho x}
{\sqrt{1-\rho^2+\eta^2}}
\Bigg)
+
1-\Phi\Bigg(
\frac{\lambda-\dfrac{\gamma_j}{\sigma_{\hat{\gamma}_j}}-\rho x}
{\sqrt{1-\rho^2+\eta^2}}
\Bigg)
\Bigg]\,dx \\
&=
\frac{\rho}{\sqrt{1+\eta^2}}
\Bigg[
\phi\Bigg(
\frac{\lambda-\dfrac{\gamma_j}{\sigma_{\hat{\gamma}_j}}}
{\sqrt{1+\eta^2}}
\Bigg)
-
\phi\Bigg(
\frac{\lambda+\dfrac{\gamma_j}{\sigma_{\hat{\gamma}_j}}}
{\sqrt{1+\eta^2}}
\Bigg)
\Bigg].
\end{aligned}
\]
Consequently,
\[
\begin{aligned}
\mathbb E\left(\hat{\Gamma}_j\mid j\in\mathcal S_\lambda\right)
&=
\Gamma_j+
\frac{\rho\sigma_{\hat{\Gamma}_j}}
{\sqrt{1+\eta^2}\Pr(j\in\mathcal S_\lambda)}
\Bigg[
\phi\Bigg(
\frac{\lambda-\dfrac{\gamma_j}{\sigma_{\hat{\gamma}_j}}}
{\sqrt{1+\eta^2}}
\Bigg)
-
\phi\Bigg(
\frac{\lambda+\dfrac{\gamma_j}{\sigma_{\hat{\gamma}_j}}}
{\sqrt{1+\eta^2}}
\Bigg)
\Bigg].
\end{aligned}
\]
This completes the proof.
\end{proof}

Lemma~\ref{lem:conditional_Gamma} shows that exposure-based IV
selection generally induces post-selection bias in
\(\hat{\Gamma}_j\) when \(\rho\neq0\).
This bias propagation motivates a joint Rao--Blackwell correction of the exposure- and outcome-association estimates, which we develop in the next subsection.

\subsection{Rao--Blackwellization of the outcome-association estimate}
\label{sec:app_RB_Gamma}

For the outcome-side Rao--Blackwell correction, consider the initial estimator
\[
\hat\Gamma_{j,\mathrm{ini}}
=
\hat\Gamma_j
-
\frac{\rho\sigma_{\hat\Gamma_j}}{\eta^2}Z_j.
\]
Since \(Z_j\) is independent of the GWAS summary statistics and has
mean zero,
\[
\mathbb E\left(
\hat\Gamma_{j,\mathrm{ini}}
\right)
=
\Gamma_j.
\]
Moreover,
\[
\begin{aligned}
\operatorname{Cov}\left(
\hat\Gamma_{j,\mathrm{ini}},
\frac{\hat\gamma_j}{\sigma_{\hat\gamma_j}}+Z_j
\right)
=
\operatorname{Cov}\left(
\hat\Gamma_j,
\frac{\hat\gamma_j}{\sigma_{\hat\gamma_j}}
\right)
-
\frac{\rho\sigma_{\hat\Gamma_j}}{\eta^2}
\operatorname{Var}(Z_j)
=
\rho\sigma_{\hat\Gamma_j}
-
\rho\sigma_{\hat\Gamma_j}
=
0.
\end{aligned}
\]
Because these two quantities are jointly bivariate normal,
\(\hat\Gamma_{j,\mathrm{ini}}\) is independent of the randomized
selection statistic and hence of the event
\(j\in\mathcal S_\lambda\). Thus,
\[
\mathbb E\left(
\hat\Gamma_{j,\mathrm{ini}}
\mid
j\in\mathcal S_\lambda
\right)
=
\Gamma_j.
\]

We then Rao--Blackwellize the initial estimator with respect to the
observed GWAS summary statistics:
\[
\hat\Gamma_{j,\mathrm{RB}}
=
\mathbb E\left(
\hat\Gamma_{j,\mathrm{ini}}
\mid
\hat\Gamma_j,\hat\gamma_j,
j\in\mathcal S_\lambda
\right).
\]
The following lemma gives its explicit form and conditional
unbiasedness.

\begin{lemma}
\label{lem:RB_Gamma}
Suppose that Assumptions~1--2 hold. Then
\[
\hat\Gamma_{j,\mathrm{RB}}
=
\hat\Gamma_j
-
\frac{\rho\sigma_{\hat\Gamma_j}}{\eta}
\frac{
\phi(A_{j,+})-\phi(A_{j,-})
}{
1-\Phi(A_{j,+})+\Phi(A_{j,-})
},
\]
where
\[
A_{j,\pm}
=
-\frac{\hat\gamma_j}
{\sigma_{\hat\gamma_j}\eta}
\pm
\frac{\lambda}{\eta}.
\]
Moreover,
\[
\mathbb E\left(
\hat\Gamma_{j,\mathrm{RB}}
\mid
j\in\mathcal S_\lambda
\right)
=
\Gamma_j.
\]
\end{lemma}

\begin{proof}
By definition,
\[
\begin{aligned}
\hat\Gamma_{j,\mathrm{RB}}
=
\hat\Gamma_j
-
\frac{\rho\sigma_{\hat\Gamma_j}}{\eta^2}
\mathbb E\left(
Z_j
\mid
\hat\Gamma_j,\hat\gamma_j,
j\in\mathcal S_\lambda
\right)
=
\hat\Gamma_j
-
\frac{\rho\sigma_{\hat\Gamma_j}}{\eta^2}
\mathbb E\left(
Z_j
\mid
\hat\gamma_j,
j\in\mathcal S_\lambda
\right),
\end{aligned}
\]
where the second equality follows because \(Z_j\) is independent of
the GWAS summary statistics and the selection event depends only on
\((\hat\gamma_j,Z_j)\).
By Lemma~S.8 of Ma et al.~(2023),
\[
\mathbb E\left(
Z_j
\mid
\hat\gamma_j,
j\in\mathcal S_\lambda
\right)
=
\eta
\frac{
\phi(A_{j,+})-\phi(A_{j,-})
}{
1-\Phi(A_{j,+})+\Phi(A_{j,-})
}.
\]
Substituting this expression gives the stated closed form.
Finally, by the law of iterated expectations and the independence of
\(\hat\Gamma_{j,\mathrm{ini}}\) from the selection event,
\[
\begin{aligned}
\mathbb E\left(
\hat\Gamma_{j,\mathrm{RB}}
\mid
j\in\mathcal S_\lambda
\right)
=
\mathbb E\left[
\mathbb E\left(
\hat\Gamma_{j,\mathrm{ini}}
\mid
\hat\Gamma_j,\hat\gamma_j,
j\in\mathcal S_\lambda
\right)
\middle|
j\in\mathcal S_\lambda
\right]
=
\mathbb E\left(
\hat\Gamma_{j,\mathrm{ini}}
\mid
j\in\mathcal S_\lambda
\right)
=
\Gamma_j.
\end{aligned}
\]
This completes the proof.
\end{proof}

\subsection{Conditional variance and uniform bounds}
\label{sec:app_RB_Gamma_variance}

\begin{lemma}
\label{lem:RB_Gamma_variance}
Suppose that Assumptions~1--2 hold. Then
\[
\begin{aligned}
\sigma^{2}_{\hat{\Gamma}_{j,\mathrm{RB}}}
&=
\operatorname{Var}\left(
\hat{\Gamma}_{j,\mathrm{RB}}
\mid j\in\mathcal S_\lambda
\right) \\
&=
\sigma_{\hat{\Gamma}_j}^2
-
\frac{\rho^2\sigma_{\hat{\Gamma}_j}^2}{\eta^2}
\mathbb{E}\Bigg[
\frac{
A_{j,+}\phi(A_{j,+})-A_{j,-}\phi(A_{j,-})
}{
1-\Phi(A_{j,+})+\Phi(A_{j,-})
}
-
\Bigg(
\frac{
\phi(A_{j,+})-\phi(A_{j,-})
}{
1-\Phi(A_{j,+})+\Phi(A_{j,-})
}
\Bigg)^2
\,\Bigg|\,
j\in\mathcal S_\lambda
\Bigg].
\end{aligned}
\]
\end{lemma}

\begin{proof}
Recall from Section~\ref{sec:app_RB_Gamma} that
\[
\hat{\Gamma}_{j,\mathrm{ini}}
=
\hat{\Gamma}_j
-
\frac{\rho\sigma_{\hat{\Gamma}_j}}{\eta^2}Z_j,
\qquad
\hat{\Gamma}_{j,\mathrm{RB}}
=
\mathbb{E}\left(
\hat{\Gamma}_{j,\mathrm{ini}}
\mid
\hat{\Gamma}_j,\hat{\gamma}_j,\,
j\in\mathcal S_\lambda
\right).
\]
By the conditional variance decomposition,
\[
\begin{aligned}
\sigma^{2}_{\hat{\Gamma}_{j,\mathrm{RB}}}
=
\operatorname{Var}\left(
\hat{\Gamma}_{j,\mathrm{ini}}
\mid j\in\mathcal S_\lambda
\right) -
\mathbb{E}\left[
\operatorname{Var}\left(
\hat{\Gamma}_{j,\mathrm{ini}}
\mid
\hat{\Gamma}_j,\hat{\gamma}_j,\,
j\in\mathcal S_\lambda
\right)
\middle|
j\in\mathcal S_\lambda
\right].
\end{aligned}
\]
Since \(\hat{\Gamma}_{j,\mathrm{ini}}\) is independent of the selection event,
\[
\operatorname{Var}\left(
\hat{\Gamma}_{j,\mathrm{ini}}
\mid j\in\mathcal S_\lambda
\right)
=
\sigma_{\hat{\Gamma}_j}^2
\left(1+\frac{\rho^2}{\eta^2}\right).
\]
Moreover,
\[
\operatorname{Var}\left(
\hat{\Gamma}_{j,\mathrm{ini}}
\mid
\hat{\Gamma}_j,\hat{\gamma}_j,\,
j\in\mathcal S_\lambda
\right)
=
\frac{\rho^2\sigma_{\hat{\Gamma}_j}^2}{\eta^4}
\operatorname{Var}\left(
Z_j\mid\hat{\gamma}_j,\,
j\in\mathcal S_\lambda
\right).
\]
By Lemma~S.8 of Ma et al.~(2023),
\[
\begin{aligned}
\operatorname{Var}\left(
Z_j\mid\hat{\gamma}_j,\,
j\in\mathcal S_\lambda
\right)
=
\eta^2
\Bigg[
1
&+
\frac{
A_{j,+}\phi(A_{j,+})-A_{j,-}\phi(A_{j,-})
}{
1-\Phi(A_{j,+})+\Phi(A_{j,-})
} -
\Bigg(
\frac{
\phi(A_{j,+})-\phi(A_{j,-})
}{
1-\Phi(A_{j,+})+\Phi(A_{j,-})
}
\Bigg)^2
\Bigg].
\end{aligned}
\]
Substitution gives the stated result.
\end{proof}

To obtain uniform bounds for the outcome-side Rao--Blackwellized
variance, we first establish the corresponding bound for the
exposure-side estimator. Ma et al.~(2023, Lemma~S.11) proved such a
bound under a diverging selection threshold. We show below that the
same conclusion also holds for a fixed threshold, allowing the
subsequent asymptotic analysis to proceed without requiring
\(\lambda\to\infty\).

\begin{lemma}
\label{lem:RB_gamma_variance_bounds}
Suppose that Assumptions~1--2 hold and that \(\lambda\geq0\) and
\(\eta>0\) are fixed. Then there exist constants
\(0<c_1<c_2<\infty\), such that
\[
c_1\sigma_{\hat{\gamma}_j}^2
\leq
\sigma_{\hat{\gamma}_{j,\mathrm{RB}}}^2
\leq
c_2\sigma_{\hat{\gamma}_j}^2
\]
uniformly over \(1\leq j\leq p\). In other words,
$
\left\{
{
\sigma_{\hat{\gamma}_{j,\mathrm{RB}}}^2
}/{
\sigma_{\hat{\gamma}_j}^2
}
:
1\leq j\leq p
\right\}
$
is uniformly bounded and bounded away from zero.
\end{lemma}

\begin{proof}
Let
\[
X_j
=
\frac{\hat{\gamma}_j}{\sigma_{\hat{\gamma}_j}},
\qquad
\mu_j
=
\frac{\gamma_j}{\sigma_{\hat{\gamma}_j}},
\]
so that \(X_j\sim N(\mu_j,1)\). Given \(X_j=x\), define the
conditional selection probability
\[
q_\lambda(x)
=
\Pr\left(
|x+Z_j|>\lambda
\right)
=
1-\Phi\left(
\frac{\lambda-x}{\eta}
\right)
+
\Phi\left(
\frac{-\lambda-x}{\eta}
\right),
\]
with derivative
\[
q_\lambda'(x)
=
\frac{1}{\eta}
\left\{
\phi\left(
\frac{\lambda-x}{\eta}
\right)
-
\phi\left(
\frac{-\lambda-x}{\eta}
\right)
\right\}.
\]
The Rao--Blackwell correction can therefore be written as
\[
\frac{\hat{\gamma}_{j,\mathrm{RB}}}
{\sigma_{\hat{\gamma}_j}}
=
X_j-\frac{q_\lambda'(X_j)}{q_\lambda(X_j)}.
\]
By Lemma~S.5 of Ma et al.~(2023), the conditional density of \(X_j\)
given \(j\in\mathcal S_\lambda\) is
\[
f_{X_j\mid j\in\mathcal S_\lambda}(x)
=
\frac{
\phi(x-\mu_j)q_\lambda(x)
}{
\Pr(j\in\mathcal S_\lambda)
}.
\]
Since
\[
\mathbb E\left(
\left.
\frac{\hat{\gamma}_{j,\mathrm{RB}}}
{\sigma_{\hat{\gamma}_j}}
\,\right|\,
j\in\mathcal S_\lambda
\right)
=
\mu_j,
\]
we have
\[
\begin{aligned}
\operatorname{Cov}\left(
\left.
\frac{\hat{\gamma}_{j,\mathrm{RB}}}
{\sigma_{\hat{\gamma}_j}},
X_j
\,\right|\,
j\in\mathcal S_\lambda
\right)
=
\frac{1}{\Pr(j\in\mathcal S_\lambda)}
\int_{-\infty}^{\infty}
\left\{
(x-\mu_j)^2q_\lambda(x)
-
(x-\mu_j)q_\lambda'(x)
\right\}
\phi(x-\mu_j)\,dx.
\end{aligned}
\]
Noting that
\[
\begin{aligned}
\left\{
(x-\mu_j)^2q_\lambda(x)
-
(x-\mu_j)q_\lambda'(x)
\right\}
\phi(x-\mu_j)
=
q_\lambda(x)\phi(x-\mu_j)
-
\frac{d}{dx}
\left\{
(x-\mu_j)q_\lambda(x)\phi(x-\mu_j)
\right\},
\end{aligned}
\]
integration over \(\mathbb R\), together with the vanishing boundary
term, gives
\[
\begin{aligned}
\int_{-\infty}^{\infty}
\left\{
(x-\mu_j)^2q_\lambda(x)
-
(x-\mu_j)q_\lambda'(x)
\right\}
\phi(x-\mu_j)\,dx
=
\int_{-\infty}^{\infty}
q_\lambda(x)\phi(x-\mu_j)\,dx
=
\Pr(j\in\mathcal S_\lambda).
\end{aligned}
\]
Hence,
\[
\operatorname{Cov}\left(
\left.
\frac{\hat{\gamma}_{j,\mathrm{RB}}}
{\sigma_{\hat{\gamma}_j}},
X_j
\,\right|\,
j\in\mathcal S_\lambda
\right)
=
1.
\]
By the Cauchy--Schwarz inequality,
\[
\operatorname{Var}\left(
\left.
\frac{\hat{\gamma}_{j,\mathrm{RB}}}
{\sigma_{\hat{\gamma}_j}}
\,\right|\,
j\in\mathcal S_\lambda
\right)
\geq
\frac{1}{
\operatorname{Var}\left(
X_j\mid j\in\mathcal S_\lambda
\right)
}.
\]

Since \(q_\lambda(x)\) is symmetric and minimized at \(x=0\),
\[
q_\lambda(x)
\geq
q_\lambda(0)
=
2\Phi\left(-\frac{\lambda}{\eta}\right)>0.
\]
Therefore,
\[
\Pr(j\in\mathcal S_\lambda)
=
\mathbb E\{q_\lambda(X_j)\}
\geq
2\Phi\left(-\frac{\lambda}{\eta}\right),
\]
and
\[
\begin{aligned}
\operatorname{Var}\left(
X_j\mid j\in\mathcal S_\lambda
\right)
\leq
\mathbb E\left[
(X_j-\mu_j)^2
\mid
j\in\mathcal S_\lambda
\right]
=
\frac{
\mathbb E\left[
(X_j-\mu_j)^2q_\lambda(X_j)
\right]
}{
\Pr(j\in\mathcal S_\lambda)
}
\leq
\left\{
2\Phi\left(-\frac{\lambda}{\eta}\right)
\right\}^{-1}.
\end{aligned}
\]
It follows that
\[
\sigma_{\hat{\gamma}_{j,\mathrm{RB}}}^2
\geq
2\Phi\left(-\frac{\lambda}{\eta}\right)
\sigma_{\hat{\gamma}_j}^2
\]
uniformly over \(j\).

For the upper bound, recall the initial exposure estimator
\[
\hat{\gamma}_{j,\mathrm{ini}}
=
\hat{\gamma}_j
-
\frac{\sigma_{\hat{\gamma}_j}}{\eta^2}Z_j,
\]
which is independent of the selection event and satisfies
\[
\operatorname{Var}\left(
\hat{\gamma}_{j,\mathrm{ini}}
\right)
=
\sigma_{\hat{\gamma}_j}^2
\left(1+\eta^{-2}\right).
\]
Since \(\hat{\gamma}_{j,\mathrm{RB}}\) is its Rao--Blackwellization,
the conditional variance decomposition yields
\[
\sigma_{\hat{\gamma}_{j,\mathrm{RB}}}^2
\leq
\operatorname{Var}\left(
\hat{\gamma}_{j,\mathrm{ini}}
\mid
j\in\mathcal S_\lambda
\right)
=
\sigma_{\hat{\gamma}_j}^2
\left(1+\eta^{-2}\right).
\]
Thus, one may take
\[
c_1
=
2\Phi\left(-\frac{\lambda}{\eta}\right),
\qquad
c_2
=
1+\eta^{-2},
\]
which proves the result.
\end{proof}

For completeness, when \(\lambda\to\infty\),
Lemma~S.11 of Ma et al.~(2023) establishes the corresponding uniform
variance bounds. Thus, the nondegeneracy of the exposure-side
Rao--Blackwellized variance does not intrinsically require a diverging
selection threshold. This exposure-side bound, together with the
variance identity above, immediately yields the corresponding result
for the outcome-side Rao--Blackwellized estimator.

\begin{lemma}
\label{lem:RB_Gamma_variance_bounds}
Suppose that Assumption~1--2 hold. Then there exist constants
\(0<c_1^\ast<c_2^\ast<\infty\) such that
\[
c_1^\ast\sigma_{\hat{\Gamma}_j}^2
\leq
\sigma^{2}_{\hat{\Gamma}_{j,\mathrm{RB}}}
\leq
c_2^\ast\sigma_{\hat{\Gamma}_j}^2
\]
uniformly over \(1\leq j\leq p\). In other words,
$
\left\{
{
\sigma_{\hat{\Gamma}_{j,\mathrm{RB}}}^2
}/{
\sigma_{\hat{\Gamma}_j}^2
}
:
1\leq j\leq p
\right\}
$
is uniformly bounded and bounded away from zero.
\end{lemma}

\begin{proof}
Combining Lemma~\ref{lem:RB_Gamma_variance} with the corresponding
variance expression for \(\hat{\gamma}_{j,\mathrm{RB}}\) gives
\[
\frac{
\sigma^{2}_{\hat{\Gamma}_{j,\mathrm{RB}}}
}{
\sigma_{\hat{\Gamma}_j}^2
}
=
(1-\rho^2)
+
\rho^2
\frac{
\sigma^{2}_{\hat{\gamma}_{j,\mathrm{RB}}}
}{
\sigma_{\hat{\gamma}_j}^2
}.
\]
By Lemma~\ref{lem:RB_gamma_variance_bounds}, there exist constants
\(0<c_1<c_2<\infty\) such that
\[
c_1
\leq
\frac{
\sigma^{2}_{\hat{\gamma}_{j,\mathrm{RB}}}
}{
\sigma_{\hat{\gamma}_j}^2
}
\leq
c_2
\]
uniformly over \(j\). Since \(0\leq\rho^2<1\), the outcome-side
variance ratio is a convex combination of \(1\) and the exposure-side
variance ratio. Hence,
\[
\min(1,c_1)
\leq
\frac{
\sigma^{2}_{\hat{\Gamma}_{j,\mathrm{RB}}}
}{
\sigma_{\hat{\Gamma}_j}^2
}
\leq
\max(1,c_2),
\]
which proves the result.
\end{proof}

\section{Conditional covariance after Rao--Blackwellization}
\label{sec:app_RB_covariance}

\subsection{Conditional covariance and its uniform bounds}
\label{sec:app_RB_covariance_properties}

We next derive the conditional covariance between
\(\hat{\gamma}_{j,\mathrm{RB}}\) and
\(\hat{\Gamma}_{j,\mathrm{RB}}\) given the selection event. Recall the
initial estimators
\[
\hat{\gamma}_{j,\mathrm{ini}}
=
\hat{\gamma}_j
-
\frac{\sigma_{\hat{\gamma}_j}}{\eta^2}Z_j,
\qquad
\hat{\Gamma}_{j,\mathrm{ini}}
=
\hat{\Gamma}_j
-
\frac{\rho\sigma_{\hat{\Gamma}_j}}{\eta^2}Z_j,
\]
whose Rao--Blackwellized counterparts are
\[
\hat{\gamma}_{j,\mathrm{RB}}
=
\mathbb{E}\left(
\hat{\gamma}_{j,\mathrm{ini}}
\mid
\hat{\gamma}_j,\,
j\in\mathcal S_\lambda
\right),
\qquad
\hat{\Gamma}_{j,\mathrm{RB}}
=
\mathbb{E}\left(
\hat{\Gamma}_{j,\mathrm{ini}}
\mid
\hat{\Gamma}_j,\hat{\gamma}_j,\,
j\in\mathcal S_\lambda
\right).
\]
For convenience, write
\[
\sigma_{\hat{\gamma}_{j,\mathrm{RB}}\hat{\Gamma}_{j,\mathrm{RB}}}
=
\operatorname{Cov}\left(
\hat{\gamma}_{j,\mathrm{RB}},
\hat{\Gamma}_{j,\mathrm{RB}}
\mid j\in\mathcal S_\lambda
\right).
\]

\begin{lemma}
\label{lem:RB_covariance}
Suppose that Assumption~1--2 hold. Then
\[
\begin{aligned}
\sigma_{\hat{\gamma}_{j,\mathrm{RB}}\hat{\Gamma}_{j,\mathrm{RB}}}
=
\rho\sigma_{\hat{\gamma}_j}\sigma_{\hat{\Gamma}_j}
-
\frac{
\rho\sigma_{\hat{\gamma}_j}\sigma_{\hat{\Gamma}_j}
}{\eta^2}
\mathbb{E}\Bigg[
\frac{
A_{j,+}\phi(A_{j,+})
-
A_{j,-}\phi(A_{j,-})
}{
1-\Phi(A_{j,+})+\Phi(A_{j,-})
} 
-
\Bigg(
\frac{
\phi(A_{j,+})-\phi(A_{j,-})
}{
1-\Phi(A_{j,+})+\Phi(A_{j,-})
}
\Bigg)^2
\,\Bigg|\,
j\in\mathcal S_\lambda
\Bigg].
\end{aligned}
\]
\end{lemma}

\begin{proof}
The two initial estimators satisfy
\[
\operatorname{Cov}\left(
\hat{\gamma}_{j,\mathrm{ini}},
\hat{\Gamma}_{j,\mathrm{ini}}
\right)
=
\rho\sigma_{\hat{\gamma}_j}\sigma_{\hat{\Gamma}_j}
\left(1+\frac{1}{\eta^2}\right).
\]
Moreover, both are uncorrelated with the randomized selection statistic:
\[
\operatorname{Cov}\left(
\hat{\gamma}_{j,\mathrm{ini}},
\frac{\hat{\gamma}_j}{\sigma_{\hat{\gamma}_j}}+Z_j
\right)
=
\operatorname{Cov}\left(
\hat{\Gamma}_{j,\mathrm{ini}},
\frac{\hat{\gamma}_j}{\sigma_{\hat{\gamma}_j}}+Z_j
\right)
=0.
\]
By joint normality,
\((\hat{\gamma}_{j,\mathrm{ini}},\hat{\Gamma}_{j,\mathrm{ini}})\)
is therefore independent of the selection event, and hence
\[
\operatorname{Cov}\left(
\hat{\gamma}_{j,\mathrm{ini}},
\hat{\Gamma}_{j,\mathrm{ini}}
\mid j\in\mathcal S_\lambda
\right)
=
\rho\sigma_{\hat{\gamma}_j}\sigma_{\hat{\Gamma}_j}
\left(1+\frac{1}{\eta^2}\right).
\]
Applying the conditional covariance decomposition gives
\[
\begin{aligned}
\sigma_{\hat{\gamma}_{j,\mathrm{RB}}\hat{\Gamma}_{j,\mathrm{RB}}}
=
\operatorname{Cov}\left(
\hat{\gamma}_{j,\mathrm{ini}},
\hat{\Gamma}_{j,\mathrm{ini}}
\mid j\in\mathcal S_\lambda
\right)  -
\mathbb{E}\left[
\operatorname{Cov}\left(
\hat{\gamma}_{j,\mathrm{ini}},
\hat{\Gamma}_{j,\mathrm{ini}}
\mid
\hat{\Gamma}_j,\hat{\gamma}_j,\,
j\in\mathcal S_\lambda
\right)
\middle|
j\in\mathcal S_\lambda
\right].
\end{aligned}
\]
Conditional on \(\hat{\Gamma}_j\), \(\hat{\gamma}_j\), and
\(j\in\mathcal S_\lambda\), the remaining randomness in both initial
estimators comes from \(Z_j\), so
\[
\operatorname{Cov}\left(
\hat{\gamma}_{j,\mathrm{ini}},
\hat{\Gamma}_{j,\mathrm{ini}}
\mid
\hat{\Gamma}_j,\hat{\gamma}_j,\,
j\in\mathcal S_\lambda
\right)
=
\frac{
\rho\sigma_{\hat{\gamma}_j}\sigma_{\hat{\Gamma}_j}
}{\eta^4}
\operatorname{Var}\left(
Z_j
\mid
\hat{\gamma}_j,\,
j\in\mathcal S_\lambda
\right).
\]
Using the truncated-normal variance formula in the proof of
Lemma~\ref{lem:RB_Gamma_variance} and simplifying gives the stated result.
\end{proof}

The conditional covariance admits a simpler representation that directly
yields its uniform bounds.

\begin{lemma}
\label{lem:RB_covariance_bounds}
Suppose that Assumption~1-2 hold and that \(\eta>0\) is fixed. For
\(\rho\neq0\), there exist constants \(0<c_1<c_2<\infty\) such that
\[
c_1
\leq
\frac{
\sigma_{\hat{\gamma}_{j,\mathrm{RB}}\hat{\Gamma}_{j,\mathrm{RB}}}
}{
\rho\sigma_{\hat{\gamma}_j}\sigma_{\hat{\Gamma}_j}
}
\leq
c_2
\]
uniformly over \(1\leq j\leq p\). Consequently,
\[
c_1|\rho|\sigma_{\hat{\gamma}_j}\sigma_{\hat{\Gamma}_j}
\leq
\left|
\sigma_{\hat{\gamma}_{j,\mathrm{RB}}\hat{\Gamma}_{j,\mathrm{RB}}}
\right|
\leq
c_2|\rho|\sigma_{\hat{\gamma}_j}\sigma_{\hat{\Gamma}_j}.
\]
Moreover, the conditional covariance has the same sign as \(\rho\)
and equals zero when \(\rho=0\).
\end{lemma}

\begin{proof}
Combining Lemma~\ref{lem:RB_covariance} with the corresponding variance
expression for \(\hat{\gamma}_{j,\mathrm{RB}}\) gives
\[
\sigma_{\hat{\gamma}_{j,\mathrm{RB}}\hat{\Gamma}_{j,\mathrm{RB}}}
=
\rho
\frac{\sigma_{\hat{\Gamma}_j}}
{\sigma_{\hat{\gamma}_j}}
\sigma^{2}_{\hat{\gamma}_{j,\mathrm{RB}}}.
\]
Hence, for \(\rho\neq0\),
\[
\frac{
\sigma_{\hat{\gamma}_{j,\mathrm{RB}}\hat{\Gamma}_{j,\mathrm{RB}}}
}{
\rho\sigma_{\hat{\gamma}_j}\sigma_{\hat{\Gamma}_j}
}
=
\frac{
\sigma^{2}_{\hat{\gamma}_{j,\mathrm{RB}}}
}{
\sigma_{\hat{\gamma}_j}^{2}
}.
\]
By Lemma~\ref{lem:RB_gamma_variance_bounds}, the ratio on the
right-hand side is uniformly bounded and bounded away from zero.
The stated bounds therefore follow immediately. Since this ratio is
strictly positive, the conditional covariance has the same sign as
\(\rho\). When \(\rho=0\), the covariance is zero directly from the
identity above.
\end{proof}

To estimate the conditional covariance, define
\[
\begin{aligned}
\hat{\sigma}_{\hat{\gamma}_{j,\mathrm{RB}}\hat{\Gamma}_{j,\mathrm{RB}}}
&=
\rho\sigma_{\hat{\gamma}_j}\sigma_{\hat{\Gamma}_j}
\Bigg[
1-
\frac{1}{\eta^2}
\frac{
A_{j,+}\phi(A_{j,+})-A_{j,-}\phi(A_{j,-})
}{
1-\Phi(A_{j,+})+\Phi(A_{j,-})
}
+
\frac{1}{\eta^2}
\Bigg(
\frac{
\phi(A_{j,+})-\phi(A_{j,-})
}{
1-\Phi(A_{j,+})+\Phi(A_{j,-})
}
\Bigg)^2
\Bigg].
\end{aligned}
\]
By Lemma~\ref{lem:RB_covariance},
\[
\mathbb{E}\left(
\hat{\sigma}_{\hat{\gamma}_{j,\mathrm{RB}}\hat{\Gamma}_{j,\mathrm{RB}}}
\mid j\in\mathcal S_\lambda
\right)
=
\sigma_{\hat{\gamma}_{j,\mathrm{RB}}\hat{\Gamma}_{j,\mathrm{RB}}}.
\]
Thus, the covariance estimator is conditionally unbiased, although it
need not be consistent for an individual SNP. For the subsequent
asymptotic analysis, it suffices to control its aggregate estimation
error over the selected instruments.

\subsection{Sharper aggregate error bounds}
\label{sec:app_RB_covariance_estimation}

For the exposure-side variance estimator, Ma et al.~(2023) established
an aggregate error bound of order
\(O_p(v^2\lambda\sqrt{p_\lambda})\).
They noted that the additional factor \(\lambda\) arose from their
higher-order moment bound for the Rao--Blackwellized quantities and
conjectured that the sharper rate
\(O_p(v^2\sqrt{p_\lambda})\) should hold.
We verify this conjecture using an orthogonal Gaussian representation
that avoids direct higher-order moment expansions of the
truncated-normal terms. Because the covariance estimator depends on
the same conditional variance of the randomization variable \(Z_j\),
the same argument also yields the sharper rate for the covariance
correction introduced here.

\begin{lemma}
\label{lem:RB_variance_covariance_estimation}
Suppose that Assumption~1--2 hold and that \(\eta>0\) is bounded and
bounded away from zero. Then
\[
\begin{aligned}
\left|
\sum_{j\in\mathcal S_\lambda}
\left(
\hat{\sigma}^{\,2}_{\hat{\gamma}_{j,\mathrm{RB}}}
-
\sigma^{2}_{\hat{\gamma}_{j,\mathrm{RB}}}
\right)
\right|
&=
O_p\left(v^2\sqrt{p_\lambda}\right),\\
\left|
\sum_{j\in\mathcal S_\lambda}
\left(
\hat{\sigma}_{\hat{\gamma}_{j,\mathrm{RB}}
\hat{\Gamma}_{j,\mathrm{RB}}}
-
\sigma_{\hat{\gamma}_{j,\mathrm{RB}}
\hat{\Gamma}_{j,\mathrm{RB}}}
\right)
\right|
&=
O_p\left(v^2\sqrt{p_\lambda}\right).
\end{aligned}
\]
\end{lemma}

\begin{proof}
For each \(j\), define
\[
V_j
=
\operatorname{Var}\left(
Z_j
\mid
\hat{\gamma}_j,\,
j\in\mathcal S_\lambda
\right).
\]
From the definitions of the Rao--Blackwellized variance and covariance
estimators,
\[
\begin{aligned}
\hat{\sigma}^{\,2}_{\hat{\gamma}_{j,\mathrm{RB}}}
-
\sigma^{2}_{\hat{\gamma}_{j,\mathrm{RB}}}
&=
-\frac{\sigma_{\hat{\gamma}_j}^2}{\eta^4}
\left[
V_j-
\mathbb E\left(
V_j\mid j\in\mathcal S_\lambda
\right)
\right],\\
\hat{\sigma}_{\hat{\gamma}_{j,\mathrm{RB}}
\hat{\Gamma}_{j,\mathrm{RB}}}
-
\sigma_{\hat{\gamma}_{j,\mathrm{RB}}
\hat{\Gamma}_{j,\mathrm{RB}}}
&=
-\frac{
\rho\sigma_{\hat{\gamma}_j}\sigma_{\hat{\Gamma}_j}
}{\eta^4}
\left[
V_j-
\mathbb E\left(
V_j\mid j\in\mathcal S_\lambda
\right)
\right].
\end{aligned}
\]
Thus, both estimation errors are centered conditional on selection.
It remains to show that
\(\operatorname{Var}(V_j\mid j\in\mathcal S_\lambda)\)
is uniformly bounded.
Define the orthogonal Gaussian variable
\[
M_j
=
\frac{\hat{\gamma}_j-\gamma_j}
{\sigma_{\hat{\gamma}_j}}
-
\frac{Z_j}{\eta^2}.
\]
Since
\[
\operatorname{Cov}\left(
M_j,\,
\frac{\hat{\gamma}_j}{\sigma_{\hat{\gamma}_j}}+Z_j
\right)
=
1-\frac{\operatorname{Var}(Z_j)}{\eta^2}
=
0,
\]
joint normality implies that \(M_j\) is independent of the randomized
selection statistic, and hence of the event
\(j\in\mathcal S_\lambda\). Moreover,
\[
M_j\sim N\left(0,1+\eta^{-2}\right).
\]
Conditional on \(\hat{\gamma}_j\) and \(j\in\mathcal S_\lambda\),
\[
M_j-
\mathbb E\left(
M_j\mid\hat{\gamma}_j,\,
j\in\mathcal S_\lambda
\right)
=
-\frac{1}{\eta^2}
\left[
Z_j-
\mathbb E\left(
Z_j\mid\hat{\gamma}_j,\,
j\in\mathcal S_\lambda
\right)
\right].
\]
Therefore,
\[
\begin{aligned}
&\mathbb E\left[
\left\{
Z_j-
\mathbb E\left(
Z_j\mid\hat{\gamma}_j,\,
j\in\mathcal S_\lambda
\right)
\right\}^4
\middle|
j\in\mathcal S_\lambda
\right] \\
&\qquad =
\eta^8
\mathbb E\left[
\left\{
M_j-
\mathbb E\left(
M_j\mid\hat{\gamma}_j,\,
j\in\mathcal S_\lambda
\right)
\right\}^4
\middle|
j\in\mathcal S_\lambda
\right] \\
&\qquad \leq
16\eta^8
\mathbb E\left(
M_j^4\mid j\in\mathcal S_\lambda
\right)
=
48\eta^8\left(1+\eta^{-2}\right)^2
=
O(1),
\end{aligned}
\]
where the inequality follows from
\(|a-b|^4\leq8(a^4+b^4)\) and Jensen's inequality, and the
last equality uses the independence of \(M_j\) and the selection event.

Recall that
\[
V_j
=
\mathbb E\left[
\left\{
Z_j-
\mathbb E\left(
Z_j
\mid
\hat{\gamma}_j,\,
j\in\mathcal S_\lambda
\right)
\right\}^2
\middle|
\hat{\gamma}_j,\,
j\in\mathcal S_\lambda
\right].
\]
Therefore, by Jensen's inequality,
\[
\begin{aligned}
V_j^2
&\leq
\mathbb E\left[
\left\{
Z_j-
\mathbb E\left(
Z_j
\mid
\hat{\gamma}_j,\,
j\in\mathcal S_\lambda
\right)
\right\}^4
\middle|
\hat{\gamma}_j,\,
j\in\mathcal S_\lambda
\right].
\end{aligned}
\]
Taking conditional expectations given \(j\in\mathcal S_\lambda\) and
using the law of iterated expectations yield
\[
\begin{aligned}
\mathbb E\left(
V_j^2
\mid
j\in\mathcal S_\lambda
\right)
&\leq
\mathbb E\left[
\left\{
Z_j-
\mathbb E\left(
Z_j
\mid
\hat{\gamma}_j,\,
j\in\mathcal S_\lambda
\right)
\right\}^4
\middle|
j\in\mathcal S_\lambda
\right] \\
&\leq
48\eta^8\left(1+\eta^{-2}\right)^2
=
O(1),
\end{aligned}
\]
uniformly over \(j\), where the last bound follows from the preceding
fourth-moment argument. Consequently,
\[
\operatorname{Var}\left(
V_j
\mid
j\in\mathcal S_\lambda
\right)
\leq
\mathbb E\left(
V_j^2
\mid
j\in\mathcal S_\lambda
\right)
=
O(1)
\]
uniformly over \(j\).

By the across-SNP independence in Assumption~2, conditional on
\(\mathcal S_\lambda\),
\[
\begin{aligned}
\operatorname{Var}\left[
\sum_{j\in\mathcal S_\lambda}
\left(
\hat{\sigma}^{\,2}_{\hat{\gamma}_{j,\mathrm{RB}}}
-
\sigma^{2}_{\hat{\gamma}_{j,\mathrm{RB}}}
\right)
\middle|
\mathcal S_\lambda
\right] 
=
\sum_{j\in\mathcal S_\lambda}
\frac{\sigma_{\hat{\gamma}_j}^4}{\eta^8}
\operatorname{Var}\left(
V_j\mid j\in\mathcal S_\lambda
\right)
=
O_p\left(p_\lambda v^4\right),
\end{aligned}
\]
where we used
\(\sigma_{\hat{\gamma}_j}=O(v)\) uniformly over \(j\).
Similarly,
\[
\begin{aligned}
\operatorname{Var}\left[
\sum_{j\in\mathcal S_\lambda}
\left(
\hat{\sigma}_{\hat{\gamma}_{j,\mathrm{RB}}
\hat{\Gamma}_{j,\mathrm{RB}}}
-
\sigma_{\hat{\gamma}_{j,\mathrm{RB}}
\hat{\Gamma}_{j,\mathrm{RB}}}
\right)
\middle|
\mathcal S_\lambda
\right]
=
\sum_{j\in\mathcal S_\lambda}
\frac{
\rho^2\sigma_{\hat{\gamma}_j}^2
\sigma_{\hat{\Gamma}_j}^2
}{\eta^8}
\operatorname{Var}\left(
V_j\mid j\in\mathcal S_\lambda
\right)
=
O_p\left(p_\lambda v^4\right),
\end{aligned}
\]
since
\(\sigma_{\hat{\gamma}_j}=O(v)\) and
\(\sigma_{\hat{\Gamma}_j}=O(v)\) uniformly over \(j\).
The two stated rates then follow from Chebyshev's
inequality.
\end{proof}

The same orthogonal Gaussian representation also yields the
fixed-order conditional moment bounds required below. In particular,
letting
\[
u_{\hat{\gamma}_{j,\mathrm{RB}}}
=
\hat{\gamma}_{j,\mathrm{RB}}-\gamma_j
=
\sigma_{\hat{\gamma}_j}
\mathbb E\left(
M_j
\mid
\hat{\gamma}_j,\,
j\in\mathcal S_\lambda
\right),
\]
Jensen's inequality gives, for every fixed \(r\geq2\),
\[
\begin{aligned}
\mathbb E\left(
\left|
u_{\hat{\gamma}_{j,\mathrm{RB}}}
\right|^r
\mid
j\in\mathcal S_\lambda
\right)
\leq
\sigma_{\hat{\gamma}_j}^r
\mathbb E\left(
|M_j|^r
\mid
j\in\mathcal S_\lambda
\right) 
=
\sigma_{\hat{\gamma}_j}^r
\mathbb E\left(
|M_j|^r
\right)
=
O(v^r),
\end{aligned}
\]
uniformly over \(j\).
Thus, the orthogonal Gaussian representation establishes the sharper
\(O_p(v^2\sqrt{p_\lambda})\) aggregate rate conjectured in RIVW and extends it to the covariance correction in BRIVW. 
The removal of the additional $\lambda$ factor leads directly to the weaker instrument-strength requirement used in the subsequent asymptotic analysis.

\section{Proof of Theorem~1}
\subsection{Consistency and asymptotic normality of the BRIVW estimator}

Following Ma et al.~(2023), we present the detailed derivations under \(\alpha_j=0\) for notational simplicity. 
Under Assumption~3, the pleiotropic effects remain centered after
selection and \(\tau_\alpha=O(v)\), ensuring that their scale is no
larger than the common sampling-error scale.
Thus, the same arguments extend to  balanced horizontal pleiotropy without changing the form of the BRIVW estimator or its variance estimator.

The BRIVW estimator is
\begin{equation*}
\hat{\beta}_{\mathrm{BRIVW}}
=
\frac{\sum_{j \in \mathcal{S}_\lambda}
\sigma_{\hat{\Gamma}_j}^{-2}\big(\hat{\Gamma}_{j,\mathrm{RB}}\hat{\gamma}_{j,\mathrm{RB}}-\hat{\sigma}_{\hat{\gamma}_{j,\mathrm{RB}} \hat{\Gamma}_{j,\mathrm{RB}}}\big)}
{\sum_{j \in \mathcal{S}_\lambda}
\sigma_{\hat{\Gamma}_j}^{-2}\big(\hat{\gamma}_{j,\mathrm{RB}}^{2}-\hat{\sigma}^{\,2}_{\hat{\gamma}_{j,\mathrm{RB}}}\big)}.
\end{equation*}

We first establish concentration of the denominator. The sharper
aggregate error bounds in
Lemma~\ref{lem:RB_variance_covariance_estimation} remove the additional
\(\lambda\)-factor appearing in the corresponding result of
Ma et al.~(2023).

\begin{lemma}
\label{lem:BRIVW_denominator}
Suppose that Assumption~1--2 hold,
\(p_\lambda\xrightarrow{p}\infty\), and
$
\kappa_\lambda\sqrt{p_\lambda}
\xrightarrow{p}\infty.
$
Then
\[
\begin{aligned}
\sum_{j\in\mathcal S_\lambda}
\sigma_{\hat{\Gamma}_j}^{-2}
\left(
\hat{\gamma}_{j,\mathrm{RB}}^2
-
\hat{\sigma}_{\hat{\gamma}_{j,\mathrm{RB}}}^{\,2}
\right)
=
\left[
\sum_{j\in\mathcal S_\lambda}
\sigma_{\hat{\Gamma}_j}^{-2}\gamma_j^2
\right]
(1+o_p(1))
=
\Theta_p\left(
\kappa_\lambda p_\lambda
\right).
\end{aligned}
\]
\end{lemma}

\begin{proof}
Write
\[
\hat{\gamma}_{j,\mathrm{RB}}
=
\gamma_j+u_{\hat{\gamma}_{j,\mathrm{RB}}}.
\]
Then
\[
\begin{aligned}
\sum_{j\in\mathcal S_\lambda}
\sigma_{\hat{\Gamma}_j}^{-2}
\left(
\hat{\gamma}_{j,\mathrm{RB}}^2
-
\hat{\sigma}_{\hat{\gamma}_{j,\mathrm{RB}}}^{\,2}
\right)
&=
\sum_{j\in\mathcal S_\lambda}
\sigma_{\hat{\Gamma}_j}^{-2}\gamma_j^2
+
2\sum_{j\in\mathcal S_\lambda}
\sigma_{\hat{\Gamma}_j}^{-2}
\gamma_j u_{\hat{\gamma}_{j,\mathrm{RB}}}
\\
&\quad+
\sum_{j\in\mathcal S_\lambda}
\sigma_{\hat{\Gamma}_j}^{-2}
\left(
u_{\hat{\gamma}_{j,\mathrm{RB}}}^2
-
\sigma_{\hat{\gamma}_{j,\mathrm{RB}}}^2
\right)
+
\sum_{j\in\mathcal S_\lambda}
\sigma_{\hat{\Gamma}_j}^{-2}
\left(
\sigma_{\hat{\gamma}_{j,\mathrm{RB}}}^2
-
\hat{\sigma}_{\hat{\gamma}_{j,\mathrm{RB}}}^{\,2}
\right).
\end{aligned}
\]
By conditional independence across SNPs and the uniform
post-selection moment bounds established above,
\[
\sum_{j\in\mathcal S_\lambda}
\sigma_{\hat{\Gamma}_j}^{-2}
\gamma_j u_{\hat{\gamma}_{j,\mathrm{RB}}}
=
O_p\left(
\sqrt{\kappa_\lambda p_\lambda}
\right),
\]
and
\[
\sum_{j\in\mathcal S_\lambda}
\sigma_{\hat{\Gamma}_j}^{-2}
\left(
u_{\hat{\gamma}_{j,\mathrm{RB}}}^2
-
\sigma_{\hat{\gamma}_{j,\mathrm{RB}}}^2
\right)
=
O_p\left(
\sqrt{p_\lambda}
\right).
\]
Likewise, by Lemma~\ref{lem:RB_variance_covariance_estimation} and the
uniform standard-error bounds in Assumption~2,
\[
\sum_{j\in\mathcal S_\lambda}
\sigma_{\hat{\Gamma}_j}^{-2}
\left(
\sigma_{\hat{\gamma}_{j,\mathrm{RB}}}^2
-
\hat{\sigma}_{\hat{\gamma}_{j,\mathrm{RB}}}^{\,2}
\right)
=
O_p\left(
\sqrt{p_\lambda}
\right).
\]
Consequently,
\[
\begin{aligned}
\sum_{j\in\mathcal S_\lambda}
\sigma_{\hat{\Gamma}_j}^{-2}
\left(
\hat{\gamma}_{j,\mathrm{RB}}^2
-
\hat{\sigma}_{\hat{\gamma}_{j,\mathrm{RB}}}^{\,2}
\right)
=
\sum_{j\in\mathcal S_\lambda}
\sigma_{\hat{\Gamma}_j}^{-2}\gamma_j^2
+
O_p\left(
\sqrt{\kappa_\lambda p_\lambda}
+
\sqrt{p_\lambda}
\right).
\end{aligned}
\]
By the definition of \(\kappa_\lambda\) and the uniform
standard error bounds in Assumption~2,
\[
\sum_{j\in\mathcal S_\lambda}
\sigma_{\hat{\Gamma}_j}^{-2}\gamma_j^2
=
\Theta_p\left(
\kappa_\lambda p_\lambda
\right).
\]
Moreover,
\[
\frac{\sqrt{\kappa_\lambda p_\lambda}}
{\kappa_\lambda p_\lambda}
=
\frac{1}{\sqrt{\kappa_\lambda p_\lambda}}
=o_p(1),
\qquad
\frac{\sqrt{p_\lambda}}
{\kappa_\lambda p_\lambda}
=
\frac{1}{\kappa_\lambda\sqrt{p_\lambda}}
=o_p(1),
\]
because
\(p_\lambda\xrightarrow{p}\infty\) and
\(\kappa_\lambda\sqrt{p_\lambda}\xrightarrow{p}\infty\).
The result follows.
\end{proof}

We next turn to the numerator. Write
\[
\hat{\beta}_{\mathrm{BRIVW}}
=
\beta
+
\frac{
\sum\nolimits_{j \in \mathcal{S}_\lambda}
\sigma_{\hat{\Gamma}_j}^{-2}u_{j,\mathrm{BRIVW}}
}{
\sum\nolimits_{j \in \mathcal{S}_\lambda}
\sigma_{\hat{\Gamma}_j}^{-2}
\left(
\hat{\gamma}_{j,\mathrm{RB}}^{2}
-
\hat{\sigma}^{\,2}_{\hat{\gamma}_{j,\mathrm{RB}}}
\right)
},
\]
where
\[
u_{\hat{\gamma}_{j,\mathrm{RB}}}
=
\hat{\gamma}_{j,\mathrm{RB}}-\gamma_j,
\qquad
u_{\hat{\Gamma}_{j,\mathrm{RB}}}
=
\hat{\Gamma}_{j,\mathrm{RB}}-\Gamma_j,
\]
and
\[
\begin{aligned}
u_{j,\mathrm{BRIVW}}
&=
\gamma_j
\left(
u_{\hat{\Gamma}_{j,\mathrm{RB}}}
-
\beta u_{\hat{\gamma}_{j,\mathrm{RB}}}
\right)
+
\left(
u_{\hat{\gamma}_{j,\mathrm{RB}}}
u_{\hat{\Gamma}_{j,\mathrm{RB}}}
-
\hat{\sigma}_{\hat{\gamma}_{j,\mathrm{RB}}
\hat{\Gamma}_{j,\mathrm{RB}}}
\right)
-
\beta
\left(
u_{\hat{\gamma}_{j,\mathrm{RB}}}^2
-
\hat{\sigma}^{\,2}_{\hat{\gamma}_{j,\mathrm{RB}}}
\right).
\end{aligned}
\]
We then characterize the conditional variance of the numerator.
Following the simplifying decomposition used in RIVW would require
the quadratic component to be asymptotically negligible. Instead, we
work directly with the complete estimating residual
\(u_{j,\mathrm{BRIVW}}\), which allows the subsequent analysis to
proceed under a weaker strength condition.

\begin{lemma}
\label{lem:BRIVW_variance_order}
Suppose that Assumptions~1--2 hold and
\(p_\lambda\xrightarrow{p}\infty\). Then
\[
\operatorname{Var}\left(
\sum_{j\in\mathcal S_\lambda}
u_{j,\mathrm{BRIVW}}
\;\middle|\;
\mathcal S_\lambda
\right)
=
\Theta_p\left\{
p_\lambda v^4(\kappa_\lambda+1)
\right\}.
\]
\end{lemma}

\begin{proof}
Under Assumption~2, the orthogonal representation established above
gives
\[
u_{\hat{\Gamma}_{j,\mathrm{RB}}}
=
\rho
\frac{\sigma_{\hat{\Gamma}_j}}
{\sigma_{\hat{\gamma}_j}}
u_{\hat{\gamma}_{j,\mathrm{RB}}}
+
\sigma_{\hat{\Gamma}_j}
\sqrt{1-\rho^2}\,\varepsilon_j,
\]
where \(\varepsilon_j\sim N(0,1)\) is independent of the
exposure-side Rao--Blackwellized quantities and the selection event.
Moreover,
\[
\hat{\sigma}_{\hat{\gamma}_{j,\mathrm{RB}}
\hat{\Gamma}_{j,\mathrm{RB}}}
=
\rho
\frac{\sigma_{\hat{\Gamma}_j}}
{\sigma_{\hat{\gamma}_j}}
\hat{\sigma}_{\hat{\gamma}_{j,\mathrm{RB}}}^{\,2}.
\]
Substituting these identities into \(u_{j,\mathrm{BRIVW}}\) gives
\[
\begin{aligned}
u_{j,\mathrm{BRIVW}}
=
\left(
\rho
\frac{\sigma_{\hat{\Gamma}_j}}
{\sigma_{\hat{\gamma}_j}}
-\beta
\right)
\left\{
\gamma_j u_{\hat{\gamma}_{j,\mathrm{RB}}}
+
u_{\hat{\gamma}_{j,\mathrm{RB}}}^2
-
\hat{\sigma}_{\hat{\gamma}_{j,\mathrm{RB}}}^{\,2}
\right\}
+
\sigma_{\hat{\Gamma}_j}\sqrt{1-\rho^2}
\left(
\gamma_j+u_{\hat{\gamma}_{j,\mathrm{RB}}}
\right)\varepsilon_j.
\end{aligned}
\]
The two terms on the right-hand side are conditionally centered and
uncorrelated. The post-selection moment bounds established above,
together with the variance-estimation bounds and Assumption~2, imply
that the conditional variance of the first term is
\[
O\left(
\gamma_j^2v^2+v^4
\right).
\]
For the second term,
\[
\begin{aligned}
&\operatorname{Var}\left[
\sigma_{\hat{\Gamma}_j}\sqrt{1-\rho^2}
\left(
\gamma_j+u_{\hat{\gamma}_{j,\mathrm{RB}}}
\right)\varepsilon_j
\;\middle|\;
j\in\mathcal S_\lambda
\right]
=
(1-\rho^2)\sigma_{\hat{\Gamma}_j}^2
\left(
\gamma_j^2+
\sigma_{\hat{\gamma}_{j,\mathrm{RB}}}^2
\right).
\end{aligned}
\]
By Lemma~\ref{lem:RB_gamma_variance_bounds} and the uniform
standard-error bounds in Assumption~2, this quantity is bounded above
and below by positive constant multiples of
\(\gamma_j^2v^2+v^4\), uniformly over \(j\). Hence,
\[
\operatorname{Var}\left(
u_{j,\mathrm{BRIVW}}
\mid
j\in\mathcal S_\lambda
\right)
=
\Theta\left(
\gamma_j^2v^2+v^4
\right)
\]
uniformly over \(j\).

Conditional independence across SNPs therefore gives
\[
\begin{aligned}
\operatorname{Var}\left(
\sum_{j\in\mathcal S_\lambda}
u_{j,\mathrm{BRIVW}}
\;\middle|\;
\mathcal S_\lambda
\right)
=
\Theta_p\left(
v^2
\sum_{j\in\mathcal S_\lambda}\gamma_j^2
+
p_\lambda v^4
\right)
=
\Theta_p\left\{
p_\lambda v^4(\kappa_\lambda+1)
\right\},
\end{aligned}
\]
where the last equality follows from the definition of
\(\kappa_\lambda\) and the uniform standard-error bounds in
Assumption~2.
\end{proof}

Combining the denominator convergence with the preceding variance
order, we now establish the consistency and asymptotic normality of
the BRIVW estimator.

\begin{lemma}
\label{lem:BRIVW_asymptotic}
Suppose that Assumptions~1--3 hold,
\(p_\lambda\xrightarrow{p}\infty\), and
$
\kappa_\lambda\sqrt{p_\lambda}
\xrightarrow{p}\infty.
$
Then
\[
\hat{\beta}_{\mathrm{BRIVW}}
\xrightarrow{p}\beta.
\]
If Assumption~4 also holds, then
\[
V_{\mathrm{BRIVW}}^{-1/2}
\left(
\hat{\beta}_{\mathrm{BRIVW}}-\beta
\right)
\xrightarrow{d}N(0,1),
\]
where
\[
V_{\mathrm{BRIVW}}
=
\frac{
\operatorname{Var}\left[
\displaystyle
\sum_{j\in\mathcal S_\lambda}
\sigma_{\hat{\Gamma}_j}^{-2}
u_{j,\mathrm{BRIVW}}
\;\middle|\;
\mathcal S_\lambda
\right]
}{
\left(
\displaystyle
\sum_{j\in\mathcal S_\lambda}
\sigma_{\hat{\Gamma}_j}^{-2}\gamma_j^2
\right)^2
}
=
\Theta_p\left\{
\frac{
\kappa_\lambda+1
}{
\kappa_\lambda^2 p_\lambda
}
\right\}.
\]
\end{lemma}

\begin{proof}
We first establish consistency. By conditional unbiasedness,
\[
\mathbb E\left(
u_{j,\mathrm{BRIVW}}
\mid
j\in\mathcal S_\lambda
\right)
=
0.
\]
It follows from Lemma~\ref{lem:BRIVW_variance_order} and
Assumption~2 that
\[
\sum_{j\in\mathcal S_\lambda}
\sigma_{\hat{\Gamma}_j}^{-2}
u_{j,\mathrm{BRIVW}}
=
O_p\left\{
\sqrt{
(\kappa_\lambda+1)p_\lambda
}
\right\}.
\]
Meanwhile, Lemma~\ref{lem:BRIVW_denominator} gives
\[
\sum_{j\in\mathcal S_\lambda}
\sigma_{\hat{\Gamma}_j}^{-2}
\left(
\hat{\gamma}_{j,\mathrm{RB}}^2
-
\hat{\sigma}_{\hat{\gamma}_{j,\mathrm{RB}}}^{\,2}
\right)
=
\Theta_p\left(
\kappa_\lambda p_\lambda
\right).
\]
Therefore,
\[
\hat{\beta}_{\mathrm{BRIVW}}-\beta
=
O_p\left\{
\frac{
\sqrt{\kappa_\lambda+1}
}{
\kappa_\lambda\sqrt{p_\lambda}
}
\right\}.
\]
Since
\[
\frac{\kappa_\lambda+1}
{\kappa_\lambda^2 p_\lambda}
=
\frac{1}{\kappa_\lambda p_\lambda}
+
\frac{1}{\kappa_\lambda^2 p_\lambda}
\xrightarrow{p}0
\]
under
\(p_\lambda\xrightarrow{p}\infty\) and
\(\kappa_\lambda\sqrt{p_\lambda}\xrightarrow{p}\infty\),
we obtain
\[
\hat{\beta}_{\mathrm{BRIVW}}
\xrightarrow{p}\beta.
\]

We next establish asymptotic normality. Conditional on
\(\mathcal S_\lambda\), the terms
\(\sigma_{\hat{\Gamma}_j}^{-2}u_{j,\mathrm{BRIVW}}\) are independent
and centered. By Lemma~\ref{lem:BRIVW_variance_order},
\[
\operatorname{Var}\left[
\sum_{j\in\mathcal S_\lambda}
\sigma_{\hat{\Gamma}_j}^{-2}
u_{j,\mathrm{BRIVW}}
\;\middle|\;
\mathcal S_\lambda
\right]
=
\Theta_p\left\{
(\kappa_\lambda+1)p_\lambda
\right\}.
\]
The orthogonal Gaussian representation and the uniform post-selection
moment bounds established above further imply
\[
\mathbb E\left[
\left.
\left|
\sigma_{\hat{\Gamma}_j}^{-2}
u_{j,\mathrm{BRIVW}}
\right|^3
\,\right|\,
j\in\mathcal S_\lambda
\right]
\leq
C\left(
\frac{|\gamma_j|^3}{v^3}+1
\right)
\]
uniformly over \(j\), for some constant \(C<\infty\). Hence,
\[
\begin{aligned}
&\frac{
\displaystyle
\sum_{j\in\mathcal S_\lambda}
\mathbb E\left[
\left.
\left|
\sigma_{\hat{\Gamma}_j}^{-2}
u_{j,\mathrm{BRIVW}}
\right|^3
\,\right|\,
\mathcal S_\lambda
\right]
}{
\displaystyle
\left\{
\operatorname{Var}\left[
\sum_{j\in\mathcal S_\lambda}
\sigma_{\hat{\Gamma}_j}^{-2}
u_{j,\mathrm{BRIVW}}
\mid
\mathcal S_\lambda
\right]
\right\}^{3/2}
}
\leq
C
\left(
\frac{
\max_{j\in\mathcal S_\lambda}\gamma_j^2
}{
\sum_{j\in\mathcal S_\lambda}\gamma_j^2
}
\right)^{1/2}
\left(
\frac{
\kappa_\lambda
}{
\kappa_\lambda+1
}
\right)^{3/2}
+
O_p\left(
p_\lambda^{-1/2}
\right).
\end{aligned}
\]
By Assumption~4,
\[
\frac{
\max_{j\in\mathcal S_\lambda}\gamma_j^2
}{
\sum_{j\in\mathcal S_\lambda}\gamma_j^2
}
\xrightarrow{p}0,
\]
and \(p_\lambda\xrightarrow{p}\infty\). Thus, the conditional
Lyapunov condition of order three holds, and
\[
\frac{
\displaystyle
\sum_{j\in\mathcal S_\lambda}
\sigma_{\hat{\Gamma}_j}^{-2}
u_{j,\mathrm{BRIVW}}
}{
\displaystyle
\sqrt{
\operatorname{Var}\left[
\sum_{j\in\mathcal S_\lambda}
\sigma_{\hat{\Gamma}_j}^{-2}
u_{j,\mathrm{BRIVW}}
\mid
\mathcal S_\lambda
\right]
}
}
\xrightarrow{d}N(0,1).
\]

Together with the denominator order and
Lemma~\ref{lem:BRIVW_variance_order}, this gives
\[
\begin{aligned}
V_{\mathrm{BRIVW}}
=
\frac{
\Theta_p\left\{
(\kappa_\lambda+1)p_\lambda
\right\}
}{
\Theta_p\left(
\kappa_\lambda^2p_\lambda^2
\right)
}
=
\Theta_p\left\{
\frac{
\kappa_\lambda+1
}{
\kappa_\lambda^2 p_\lambda
}
\right\}.
\end{aligned}
\]
Finally,
\[
\begin{aligned}
V_{\mathrm{BRIVW}}^{-1/2}
\left(
\hat{\beta}_{\mathrm{BRIVW}}-\beta
\right)
=
\frac{
\displaystyle
\sum_{j\in\mathcal S_\lambda}
\sigma_{\hat{\Gamma}_j}^{-2}
u_{j,\mathrm{BRIVW}}
}{
\displaystyle
\sqrt{
\operatorname{Var}\left[
\sum_{j\in\mathcal S_\lambda}
\sigma_{\hat{\Gamma}_j}^{-2}
u_{j,\mathrm{BRIVW}}
\mid
\mathcal S_\lambda
\right]
}
}
\frac{
\displaystyle
\sum_{j\in\mathcal S_\lambda}
\sigma_{\hat{\Gamma}_j}^{-2}\gamma_j^2
}{
\displaystyle
\sum_{j\in\mathcal S_\lambda}
\sigma_{\hat{\Gamma}_j}^{-2}
\left(
\hat{\gamma}_{j,\mathrm{RB}}^2
-
\hat{\sigma}_{\hat{\gamma}_{j,\mathrm{RB}}}^{\,2}
\right)
}.
\end{aligned}
\]
The first factor converges in distribution to \(N(0,1)\), whereas the
second converges in probability to one by
Lemma~\ref{lem:BRIVW_denominator}. Slutsky's theorem therefore yields
\[
V_{\mathrm{BRIVW}}^{-1/2}
\left(
\hat{\beta}_{\mathrm{BRIVW}}-\beta
\right)
\xrightarrow{d}N(0,1).
\]
\end{proof}

Lemma~\ref{lem:BRIVW_asymptotic} establishes the consistency and
asymptotic normality parts of Theorem~1. Consistency of the variance
estimator is established in the next subsection.

\subsection{Consistency of the BRIVW variance estimator}

It remains to establish consistency of the variance estimator in
Theorem~1. We first consider the oracle estimator obtained by
evaluating the estimating residual at the true value of \(\beta\), and
then show that replacing \(\beta\) by
\(\hat{\beta}_{\mathrm{BRIVW}}\) is asymptotically negligible.

\begin{lemma}
\label{lem:BRIVW_variance_consistency}
Suppose that Assumptions~1--4 hold,
\(p_\lambda\xrightarrow{p}\infty\), and
$
\kappa_\lambda\sqrt{p_\lambda}
\xrightarrow{p}\infty.
$
Then
\[
\frac{
\widehat V_{\mathrm{BRIVW}}
}{
V_{\mathrm{BRIVW}}
}
\xrightarrow{p}1.
\]
\end{lemma}

\begin{proof}
Conditional on \(\mathcal S_\lambda\), the
\(u_{j,\mathrm{BRIVW}}\)'s are independent and centered. Hence, an
oracle moment estimator of the conditional numerator variance is
\[
\widetilde{\operatorname{Var}}\left[
\sum_{j\in\mathcal S_\lambda}
\sigma_{\hat{\Gamma}_j}^{-2}
u_{j,\mathrm{BRIVW}}
\;\middle|\;
\mathcal S_\lambda
\right]
=
\sum_{j\in\mathcal S_\lambda}
\sigma_{\hat{\Gamma}_j}^{-4}
u_{j,\mathrm{BRIVW}}^2.
\]
Let
\[
X_{j,\lambda}
=
\sigma_{\hat{\Gamma}_j}^{-2}
u_{j,\mathrm{BRIVW}},
\qquad
s_\lambda^2
=
\operatorname{Var}\left[
\sum_{j\in\mathcal S_\lambda}
X_{j,\lambda}
\;\middle|\;
\mathcal S_\lambda
\right].
\]
By Lemma~\ref{lem:BRIVW_variance_order},
\[
s_\lambda^2
=
\Theta_p\left\{
(\kappa_\lambda+1)p_\lambda
\right\}.
\]
Moreover, the third-moment bound established in
Lemma~\ref{lem:BRIVW_asymptotic} gives
\[
L_\lambda
:=
\frac{
\displaystyle
\sum_{j\in\mathcal S_\lambda}
\mathbb E\left(
|X_{j,\lambda}|^3
\mid
\mathcal S_\lambda
\right)
}{
s_\lambda^3
}
\xrightarrow{p}0.
\]
For any fixed \(\delta>0\), decompose
\[
X_{j,\lambda}^2
=
X_{j,\lambda}^2
\mathbf 1\left\{
|X_{j,\lambda}|\leq\delta s_\lambda
\right\}
+
X_{j,\lambda}^2
\mathbf 1\left\{
|X_{j,\lambda}|>\delta s_\lambda
\right\}.
\]
For the truncated part, conditional independence gives
\[
\begin{aligned}
\frac{1}{s_\lambda^4}
\operatorname{Var}\left[
\sum_{j\in\mathcal S_\lambda}
X_{j,\lambda}^2
\mathbf 1\left\{
|X_{j,\lambda}|\leq\delta s_\lambda
\right\}
\;\middle|\;
\mathcal S_\lambda
\right]
\leq
\frac{1}{s_\lambda^4}
\sum_{j\in\mathcal S_\lambda}
\mathbb E\left[
X_{j,\lambda}^4
\mathbf 1\left\{
|X_{j,\lambda}|\leq\delta s_\lambda
\right\}
\;\middle|\;
\mathcal S_\lambda
\right]
\leq
\delta L_\lambda
=
o_p(1).
\end{aligned}
\]
Hence, by Chebyshev's inequality, the centered truncated sum is
\(o_p(s_\lambda^2)\).
For the tail part,
\[
\begin{aligned}
\frac{1}{s_\lambda^2}
\sum_{j\in\mathcal S_\lambda}
\mathbb E\left[
X_{j,\lambda}^2
\mathbf 1\left\{
|X_{j,\lambda}|>\delta s_\lambda
\right\}
\;\middle|\;
\mathcal S_\lambda
\right]
\leq
\frac{1}{\delta}
L_\lambda
=
o_p(1).
\end{aligned}
\]
Markov's inequality therefore shows that the tail contribution is also
\(o_p(s_\lambda^2)\). Since
\[
\sum_{j\in\mathcal S_\lambda}
\mathbb E\left(
X_{j,\lambda}^2
\mid
\mathcal S_\lambda
\right)
=
s_\lambda^2,
\]
we obtain
\[
\frac{
\sum_{j\in\mathcal S_\lambda}
X_{j,\lambda}^2
}{
s_\lambda^2
}
\xrightarrow{p}1.
\]
Equivalently,
\[
\frac{
\widetilde{\operatorname{Var}}\left[
\displaystyle
\sum_{j\in\mathcal S_\lambda}
\sigma_{\hat{\Gamma}_j}^{-2}
u_{j,\mathrm{BRIVW}}
\;\middle|\;
\mathcal S_\lambda
\right]
}{
\operatorname{Var}\left[
\displaystyle
\sum_{j\in\mathcal S_\lambda}
\sigma_{\hat{\Gamma}_j}^{-2}
u_{j,\mathrm{BRIVW}}
\;\middle|\;
\mathcal S_\lambda
\right]
}
\xrightarrow{p}1.
\]
Together with
\(s_\lambda^2
=
\Theta_p\{(\kappa_\lambda+1)p_\lambda\}\),
this further implies
\[
\widetilde{\operatorname{Var}}\left[
\sum_{j\in\mathcal S_\lambda}
\sigma_{\hat{\Gamma}_j}^{-2}
u_{j,\mathrm{BRIVW}}
\;\middle|\;
\mathcal S_\lambda
\right]
=
\Theta_p\left\{
(\kappa_\lambda+1)p_\lambda
\right\}.
\]
We next replace \(\beta\) by
\(\hat{\beta}_{\mathrm{BRIVW}}\). Expanding the feasible squared
residuals gives
\[
\begin{aligned}
&\widehat{\operatorname{Var}}\left[
\sum_{j\in\mathcal S_\lambda}
\sigma_{\hat{\Gamma}_j}^{-2}
u_{j,\mathrm{BRIVW}}
\;\middle|\;
\mathcal S_\lambda
\right] \\
&=
\widetilde{\operatorname{Var}}\left[
\sum_{j\in\mathcal S_\lambda}
\sigma_{\hat{\Gamma}_j}^{-2}
u_{j,\mathrm{BRIVW}}
\;\middle|\;
\mathcal S_\lambda
\right]
+
\left(
\beta-\hat{\beta}_{\mathrm{BRIVW}}
\right)^2
\sum_{j\in\mathcal S_\lambda}
\sigma_{\hat{\Gamma}_j}^{-4}
\left(
\hat{\gamma}_{j,\mathrm{RB}}^2
-
\hat{\sigma}_{\hat{\gamma}_{j,\mathrm{RB}}}^{\,2}
\right)^2
\\
&\quad+
2\left(
\beta-\hat{\beta}_{\mathrm{BRIVW}}
\right)
\sum_{j\in\mathcal S_\lambda}
\sigma_{\hat{\Gamma}_j}^{-4}
u_{j,\mathrm{BRIVW}}
\left(
\hat{\gamma}_{j,\mathrm{RB}}^2
-
\hat{\sigma}_{\hat{\gamma}_{j,\mathrm{RB}}}^{\,2}
\right).
\end{aligned}
\]
Using the decomposition in the proof of
Lemma~\ref{lem:BRIVW_denominator} and the fixed-order
post-selection moment bounds established above,
\[
\begin{aligned}
&\sum_{j\in\mathcal S_\lambda}
\sigma_{\hat{\Gamma}_j}^{-4}
\left(
\hat{\gamma}_{j,\mathrm{RB}}^2
-
\hat{\sigma}_{\hat{\gamma}_{j,\mathrm{RB}}}^{\,2}
\right)^2
=
O_p\left\{
\sum_{j\in\mathcal S_\lambda}
\frac{\gamma_j^4}{v^4}
+
\kappa_\lambda p_\lambda
+
p_\lambda
\right\}.
\end{aligned}
\]
Assumption~4 gives
\[
\frac{
\displaystyle
\sum_{j\in\mathcal S_\lambda}\gamma_j^4/v^4
}{
\kappa_\lambda^2 p_\lambda^2
}
\leq
C
\frac{
\max_{j\in\mathcal S_\lambda}\gamma_j^2
}{
\sum_{j\in\mathcal S_\lambda}\gamma_j^2
}
\xrightarrow{p}0.
\]
Furthermore,
\[
\frac{\kappa_\lambda p_\lambda}
{\kappa_\lambda^2p_\lambda^2}
=
\frac{1}{\kappa_\lambda p_\lambda}
\xrightarrow{p}0,
\qquad
\frac{p_\lambda}
{\kappa_\lambda^2p_\lambda^2}
=
\frac{1}{\kappa_\lambda^2p_\lambda}
\xrightarrow{p}0
\]
under
\(p_\lambda\xrightarrow{p}\infty\) and
\(\kappa_\lambda\sqrt{p_\lambda}\xrightarrow{p}\infty\).
Hence,
\[
\sum_{j\in\mathcal S_\lambda}
\sigma_{\hat{\Gamma}_j}^{-4}
\left(
\hat{\gamma}_{j,\mathrm{RB}}^2
-
\hat{\sigma}_{\hat{\gamma}_{j,\mathrm{RB}}}^{\,2}
\right)^2
=
o_p\left(
\kappa_\lambda^2p_\lambda^2
\right).
\]
From Lemma~\ref{lem:BRIVW_asymptotic},
\[
\hat{\beta}_{\mathrm{BRIVW}}-\beta
=
O_p\left\{
\frac{
\sqrt{\kappa_\lambda+1}
}{
\kappa_\lambda\sqrt{p_\lambda}
}
\right\}.
\]
It follows that
\[
\begin{aligned}
\left(
\beta-\hat{\beta}_{\mathrm{BRIVW}}
\right)^2
\sum_{j\in\mathcal S_\lambda}
\sigma_{\hat{\Gamma}_j}^{-4}
\left(
\hat{\gamma}_{j,\mathrm{RB}}^2
-
\hat{\sigma}_{\hat{\gamma}_{j,\mathrm{RB}}}^{\,2}
\right)^2
=
o_p\left\{
(\kappa_\lambda+1)p_\lambda
\right\}.
\end{aligned}
\]
For the cross term, the Cauchy--Schwarz inequality gives
\[
\begin{aligned}
&\left|
\left(
\beta-\hat{\beta}_{\mathrm{BRIVW}}
\right)
\sum_{j\in\mathcal S_\lambda}
\sigma_{\hat{\Gamma}_j}^{-4}
u_{j,\mathrm{BRIVW}}
\left(
\hat{\gamma}_{j,\mathrm{RB}}^2
-
\hat{\sigma}_{\hat{\gamma}_{j,\mathrm{RB}}}^{\,2}
\right)
\right|
\\
&\quad\leq
\left|
\beta-\hat{\beta}_{\mathrm{BRIVW}}
\right|
\left[
\sum_{j\in\mathcal S_\lambda}
\sigma_{\hat{\Gamma}_j}^{-4}
u_{j,\mathrm{BRIVW}}^2
\right]^{1/2}
\left[
\sum_{j\in\mathcal S_\lambda}
\sigma_{\hat{\Gamma}_j}^{-4}
\left(
\hat{\gamma}_{j,\mathrm{RB}}^2
-
\hat{\sigma}_{\hat{\gamma}_{j,\mathrm{RB}}}^{\,2}
\right)^2
\right]^{1/2}
=
o_p\left\{
(\kappa_\lambda+1)p_\lambda
\right\}.
\end{aligned}
\]
Therefore,
\[
\begin{aligned}
&\widehat{\operatorname{Var}}\left[
\sum_{j\in\mathcal S_\lambda}
\sigma_{\hat{\Gamma}_j}^{-2}
u_{j,\mathrm{BRIVW}}
\;\middle|\;
\mathcal S_\lambda
\right]
=
\widetilde{\operatorname{Var}}\left[
\sum_{j\in\mathcal S_\lambda}
\sigma_{\hat{\Gamma}_j}^{-2}
u_{j,\mathrm{BRIVW}}
\;\middle|\;
\mathcal S_\lambda
\right]
\{1+o_p(1)\}.
\end{aligned}
\]
Together with the oracle ratio consistency above,
\[
\frac{
\widehat{\operatorname{Var}}\left[
\displaystyle
\sum_{j\in\mathcal S_\lambda}
\sigma_{\hat{\Gamma}_j}^{-2}
u_{j,\mathrm{BRIVW}}
\;\middle|\;
\mathcal S_\lambda
\right]
}{
\operatorname{Var}\left[
\displaystyle
\sum_{j\in\mathcal S_\lambda}
\sigma_{\hat{\Gamma}_j}^{-2}
u_{j,\mathrm{BRIVW}}
\;\middle|\;
\mathcal S_\lambda
\right]
}
\xrightarrow{p}1.
\]
Combining the preceding numerator variance consistency with
Lemma~\ref{lem:BRIVW_denominator} yields
\[
\frac{
\widehat V_{\mathrm{BRIVW}}
}{
V_{\mathrm{BRIVW}}
}
\xrightarrow{p}1.
\]
\end{proof}

Together with Lemma~\ref{lem:BRIVW_asymptotic}, this completes the
proof of Theorem~1.

\counterwithout{figure}{section}
\setcounter{figure}{0}
\renewcommand{\thefigure}{S\arabic{figure}}
\renewcommand{\theHfigure}{S\arabic{figure}}

\section{Additional Simulation Studies}
\label{sec:app_simulation_BRIVW}

\subsection{Numerical illustration of the outcome-side post-selection shift}
\label{sec:app_BRIVW_illustrations}

To illustrate the outcome-side post-selection effect characterized in
Section~4.1, we examine the conditional distributions of
\(\hat{\Gamma}_j\) and its Rao--Blackwellized counterpart
\(\hat{\Gamma}_{j,\mathrm{RB}}\) under different levels of instrument
strength and sample structure.

We generate \(\hat{\gamma}_j/\sigma_{\hat{\gamma}_j}\) and
\(\hat{\Gamma}_j/\sigma_{\hat{\Gamma}_j}\) from a bivariate normal
distribution with respective means
\(\gamma_j/\sigma_{\hat{\gamma}_j}\) and
\(\Gamma_j/\sigma_{\hat{\Gamma}_j}\), unit marginal variances, and
correlation coefficient \(\rho\).
We consider
$
\gamma_j/\sigma_{\hat{\gamma}_j}
\in
\{0.1\lambda,\,0.5\lambda,\,2\lambda\},
$
representing weak, moderate, and strong IVs, respectively, and
set $\Gamma_j/\sigma_{\hat{\Gamma}_j}=\gamma_j/\sigma_{\hat{\gamma}_j}$.
Here,
\(\lambda=4.06\),
\(\eta=0.5\), and
\(\rho\in\{-0.3,-0.1,0,0.1,0.3\}\).
SNPs are selected using the randomized selection rule defined in Section~3.

\setcounter{figure}{0}
\renewcommand{\thefigure}{S\arabic{figure}}
\renewcommand{\figurename}{Web Figure}

\begin{figure}[!htbp]
  \centering
  \includegraphics[width=\linewidth]{FigS1.pdf}
  \caption{\small
  Post-selection distributions of
  \(\hat{\Gamma}_j\) and \(\hat{\Gamma}_{j,\mathrm{RB}}\)
  under different levels of sample structure and instrument strength.
  Weak, moderate, and strong instruments correspond to
  \(\gamma_j/\sigma_{\hat{\gamma}_j}
  \in\{0.1\lambda,0.5\lambda,2\lambda\}\), where
  \(\lambda=4.06\).
  The red dotted lines indicate the true value of \(\Gamma_j\).
  }
  \label{fig:app_postselection_distribution}
\end{figure}

Web Figure~\ref{fig:app_postselection_distribution} confirms the pattern predicted by Section~4.1.
When \(\rho=0\), the post-selection distribution of
\(\hat{\Gamma}_j\) remains centered at \(\Gamma_j\).
When \(\rho\neq0\), selection induces an outcome-side mean shift that
becomes more pronounced as \(|\rho|\) increases and is largest for weaker
instruments.
In contrast, \(\hat{\Gamma}_{j,\mathrm{RB}}\) remains well centered around
\(\Gamma_j\) across all settings.

\subsection{Sensitivity to the randomization scale \(\eta\)}

\label{sec:app_BRIVW_eta_sensitivity}

We examine the sensitivity of BRIVW to the randomization scale \(\eta\).
The data-generating mechanism follows the main simulation, with
\(\beta=0.2\), \(w_1=0\), and
\(\rho\in\{-0.3,-0.1,0,0.1,0.3\}\).
To cover a range of instrument strengths, we consider
$
\pi_x\in\{0.005,0.01,0.03,0.05\},
$
corresponding to expected exposure heritabilities of approximately
\(0.05\), \(0.10\), \(0.30\), and \(0.50\), respectively.
We fix
$
\lambda=4.06
$
and vary
$
\eta\in\{0.4,0.5,0.6,0.8,1\}.
$
All other settings follow the main simulation.
Each scenario is evaluated using \(2{,}000\) replicates,
with performance summarized by empirical bias and standard deviation.

% ============================================================
% Web Figure S2
% Sensitivity to eta
% ============================================================

\begin{figure}[!htbp]
  \centering
  \includegraphics[width=\linewidth]{FigS2.pdf}
  \caption{\small
  Sensitivity of BRIVW to the randomization scale \(\eta\).
  Empirical bias and standard deviation are shown across
  different values of \(\eta\), exposure heritability, and
  sample-structure correlation, with the selection threshold fixed at
  \(\lambda=4.06\).
  }
  \label{fig:app_eta_sensitivity}
\end{figure}

Web Figure~\ref{fig:app_eta_sensitivity} shows that BRIVW is insensitive
to the choice of \(\eta\) over the range considered.
Bias remains close to zero across settings, while the empirical standard
deviation varies little with \(\eta\) and decreases as exposure heritability
increases.

\subsection{Sensitivity to the causal effect size}

\label{sec:app_BRIVW_beta_sensitivity}

We examine sensitivity to the causal effect size by setting
\(\beta=0.05\) and \(0.5\), while keeping all other settings as in the
main simulation.
Unless otherwise stated, additional simulations that include MR-APSS are
based on \(500\) replicates because of its computational cost.

\begin{figure}[!htbp]
  \centering
  \includegraphics[width=\linewidth]{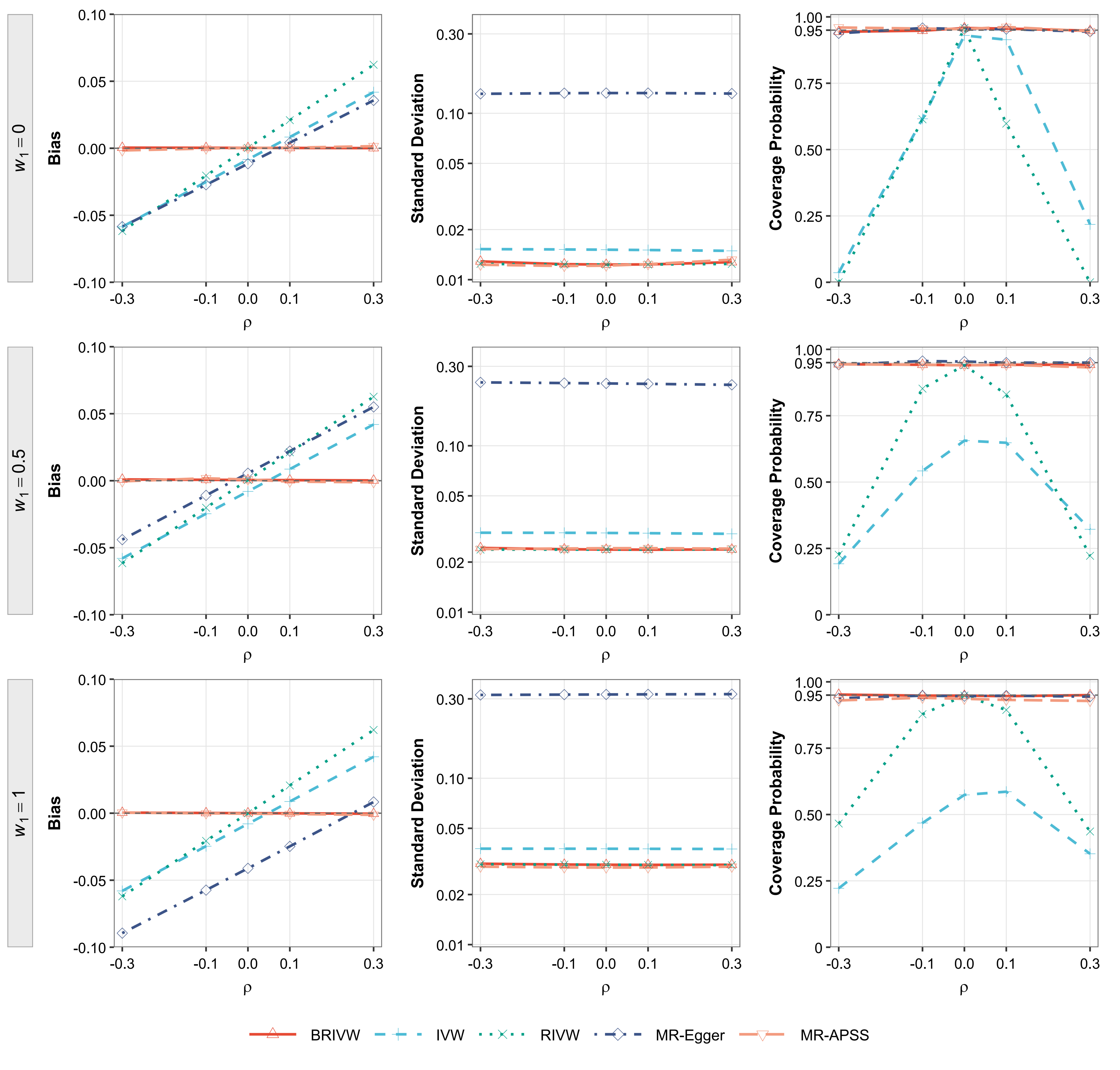}
  \caption{\small
  Sensitivity analysis with causal effect \(\beta=0.05\).
  Empirical bias, standard deviation, and coverage probability of BRIVW
    and competing methods across different values of \(\rho\) and \(w_1\).
  }
  \label{fig:app_beta005_sensitivity}
\end{figure}

\begin{figure}[!htbp]
  \centering
  \includegraphics[width=\linewidth]{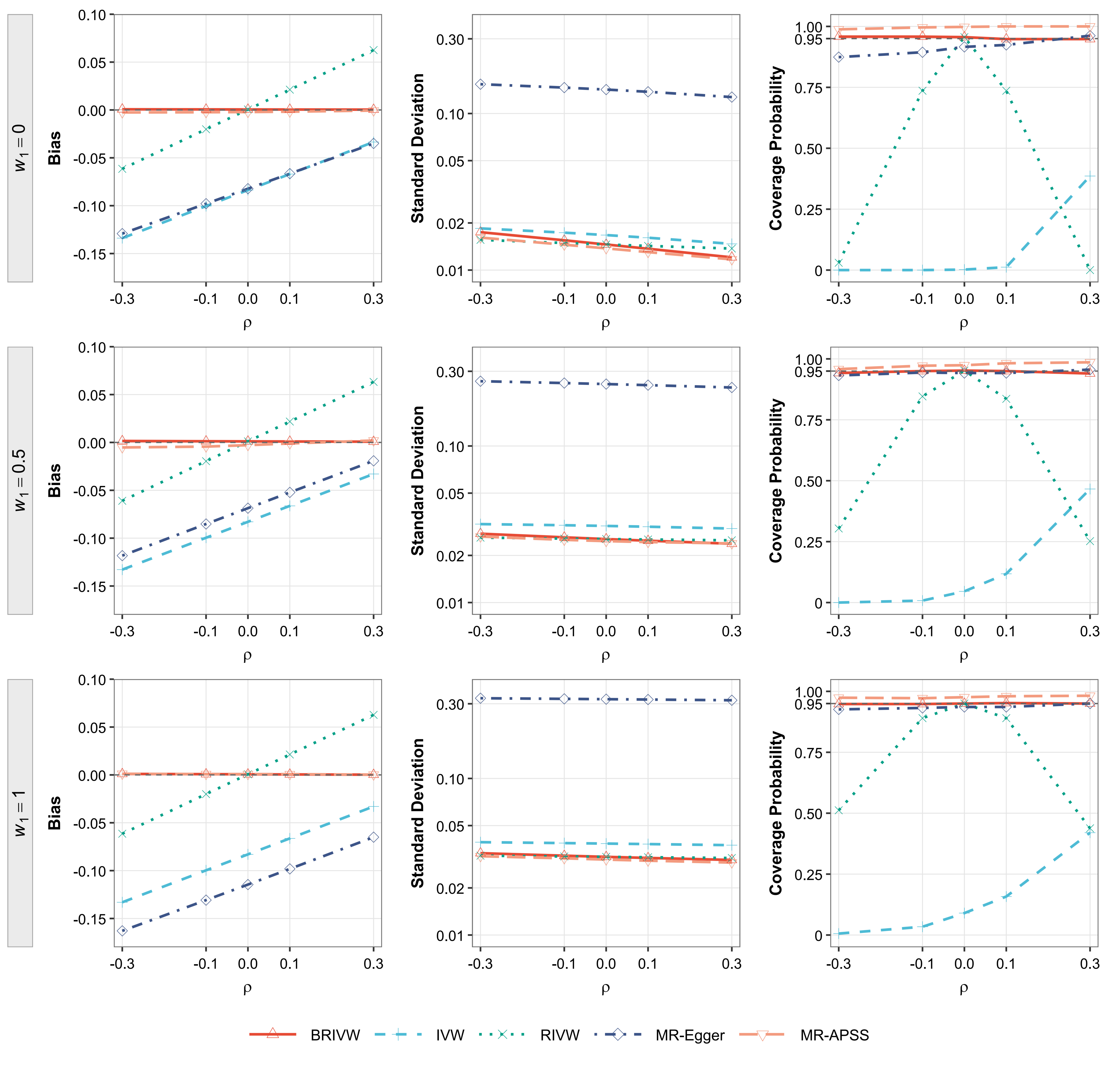}
  \caption{\small
  Sensitivity analysis with causal effect \(\beta=0.5\).
  Empirical bias, standard deviation, and coverage probability of BRIVW
    and competing methods across different values of \(\rho\) and \(w_1\).
  }
  \label{fig:app_beta05_sensitivity}
\end{figure}

Web Figures~\ref{fig:app_beta005_sensitivity} and
\ref{fig:app_beta05_sensitivity} show patterns consistent with the main
simulation.
BRIVW remains nearly unbiased with coverage close to the nominal level
across values of \(\rho\) and \(w_1\), whereas the performance of methods
that do not fully account for sample structure remains sensitive to
\(\rho\).

\subsection{Unequal GWAS sample sizes}

\label{sec:app_BRIVW_unequal_sample_sizes}

We further assess BRIVW under unequal GWAS sample sizes by considering
\[
(n_X,n_Y)=(50{,}000,100{,}000)
\quad\text{and}\quad
(100{,}000,50{,}000),
\]
while keeping all other settings as in the main simulation.

\begin{figure}[!htbp]
  \centering
  \includegraphics[width=\linewidth]{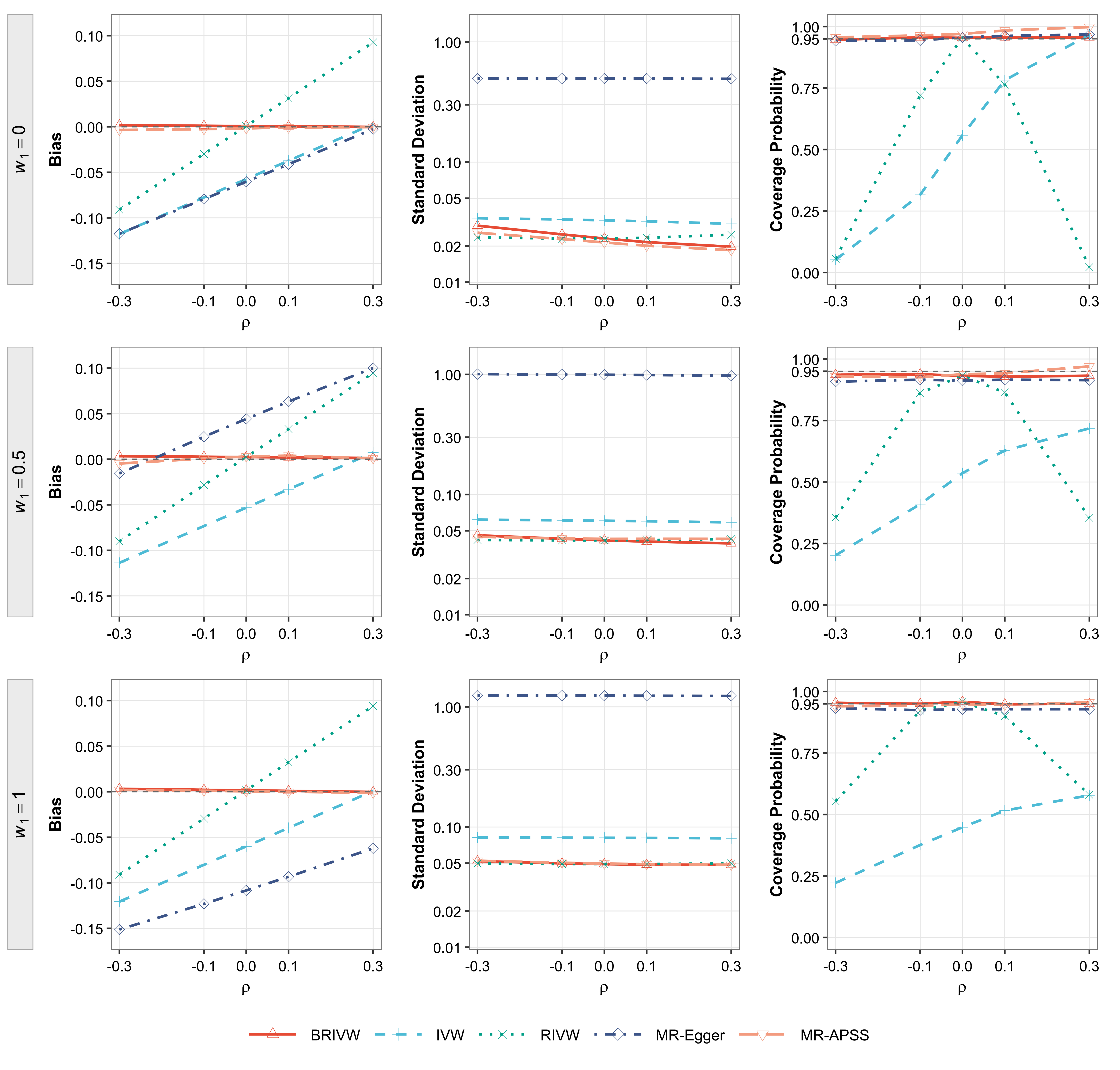}
  \caption{\small
  Simulation results with unequal GWAS sample sizes
  \(n_X=50{,}000\) and \(n_Y=100{,}000\).
  Empirical bias, standard deviation, and coverage probability of BRIVW
  and competing methods across different values of \(\rho\) and \(w_1\).
  }
  \label{fig:app_small_exposure}
\end{figure}

\begin{figure}[!htbp]
  \centering
  \includegraphics[width=\linewidth]{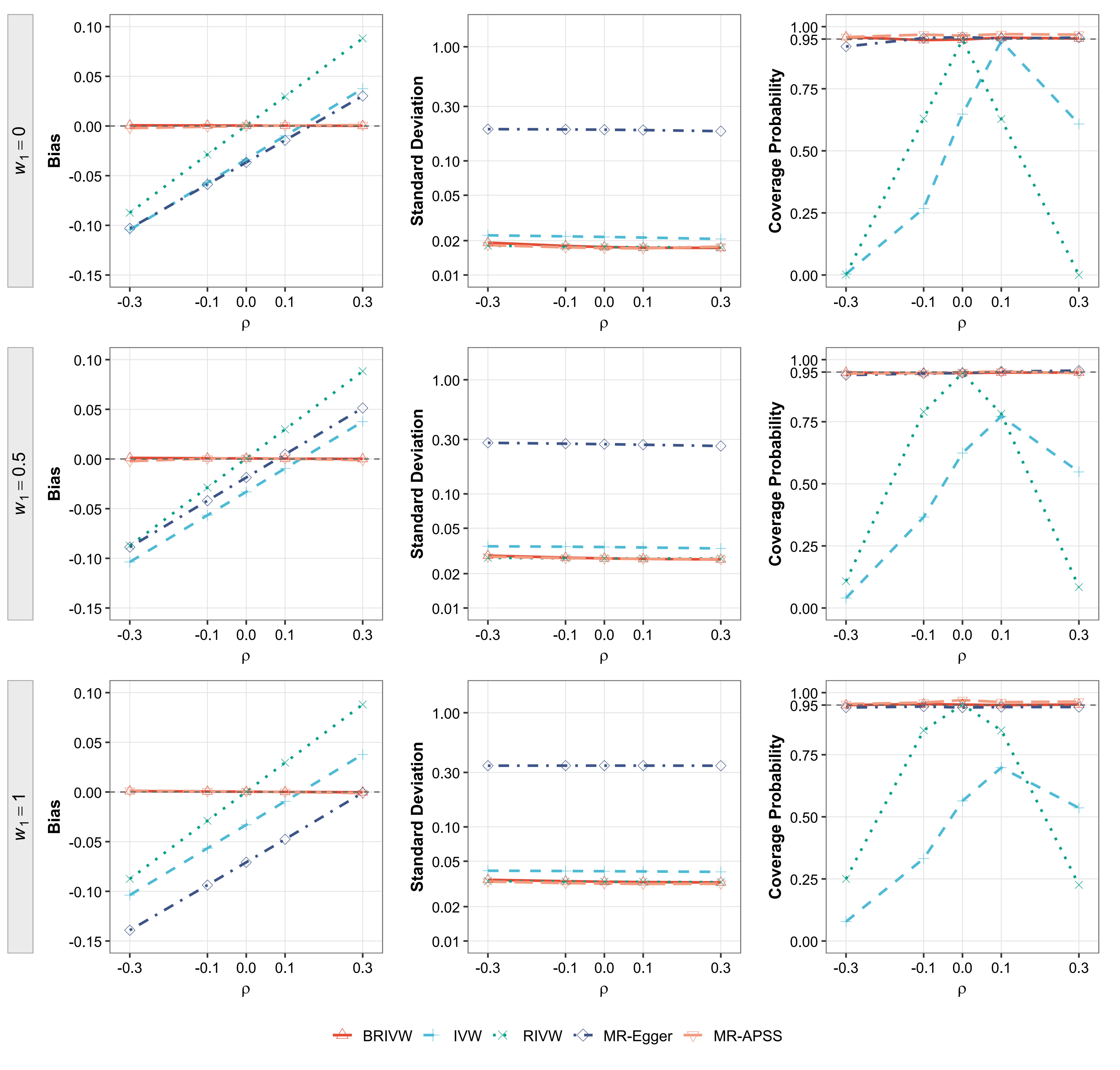}
  \caption{\small
  Simulation results with unequal GWAS sample sizes
  \(n_X=100{,}000\) and \(n_Y=50{,}000\).
  Empirical bias, standard deviation, and coverage probability of BRIVW
  and competing methods across different values of \(\rho\) and \(w_1\).
  }
  \label{fig:app_small_outcome}
\end{figure}

Web Figures~\ref{fig:app_small_exposure} and
\ref{fig:app_small_outcome} show conclusions consistent with the main
simulation.
BRIVW remains nearly unbiased with close-to-nominal coverage across
settings, although its variability increases with the reduced GWAS sample
size, particularly when the exposure GWAS is smaller.
The qualitative patterns of the competing methods remain largely unchanged.

\subsection{Robustness to the misspecification of the normal mixture model}

\label{sec:app_BRIVW_mixture_misspecification}

We further examine robustness to misspecification of the Gaussian mixture
model used by MR-APSS.
To isolate distributional misspecification, we replace one half of the Gaussian
genetic-effect components in the main simulation by uniform distributions:
\[
\begin{aligned}
\begin{pmatrix}
\gamma_j\\
\alpha_j
\end{pmatrix}
\sim\;&
\frac{\pi_x(1-w_1)}{2}
\begin{pmatrix}
N(0,\varepsilon_x^2)\\
\delta_0
\end{pmatrix}
+
\frac{\pi_x(1-w_1)}{2}
\begin{pmatrix}
\operatorname{Unif}(-0.02,0.02)\\
\delta_0
\end{pmatrix}
\\
&+
\frac{\pi_xw_1}{2}
\begin{pmatrix}
N(0,\varepsilon_x^2)\\
N(0,\tau^2)
\end{pmatrix}
+
\frac{\pi_xw_1}{2}
\begin{pmatrix}
\operatorname{Unif}(-0.02,0.02)\\
\operatorname{Unif}(-0.02,0.02)
\end{pmatrix}
\\
&+
\pi_y
\begin{pmatrix}
\delta_0\\
N(0,\tau^2)
\end{pmatrix}
+
(1-\pi_x-\pi_y)
\begin{pmatrix}
\delta_0\\
\delta_0
\end{pmatrix}.
\end{aligned}
\]
Within each balanced pleiotropic component, \(\gamma_j\) and \(\alpha_j\)
are generated independently.
We consider \(w_1\in\{0,0.5,1\}\) and
\(\beta\in\{0.05,0.2,0.5\}\), with all other settings as in the main
simulation.

\begin{figure}[!htbp]
  \centering
  \includegraphics[width=\linewidth]{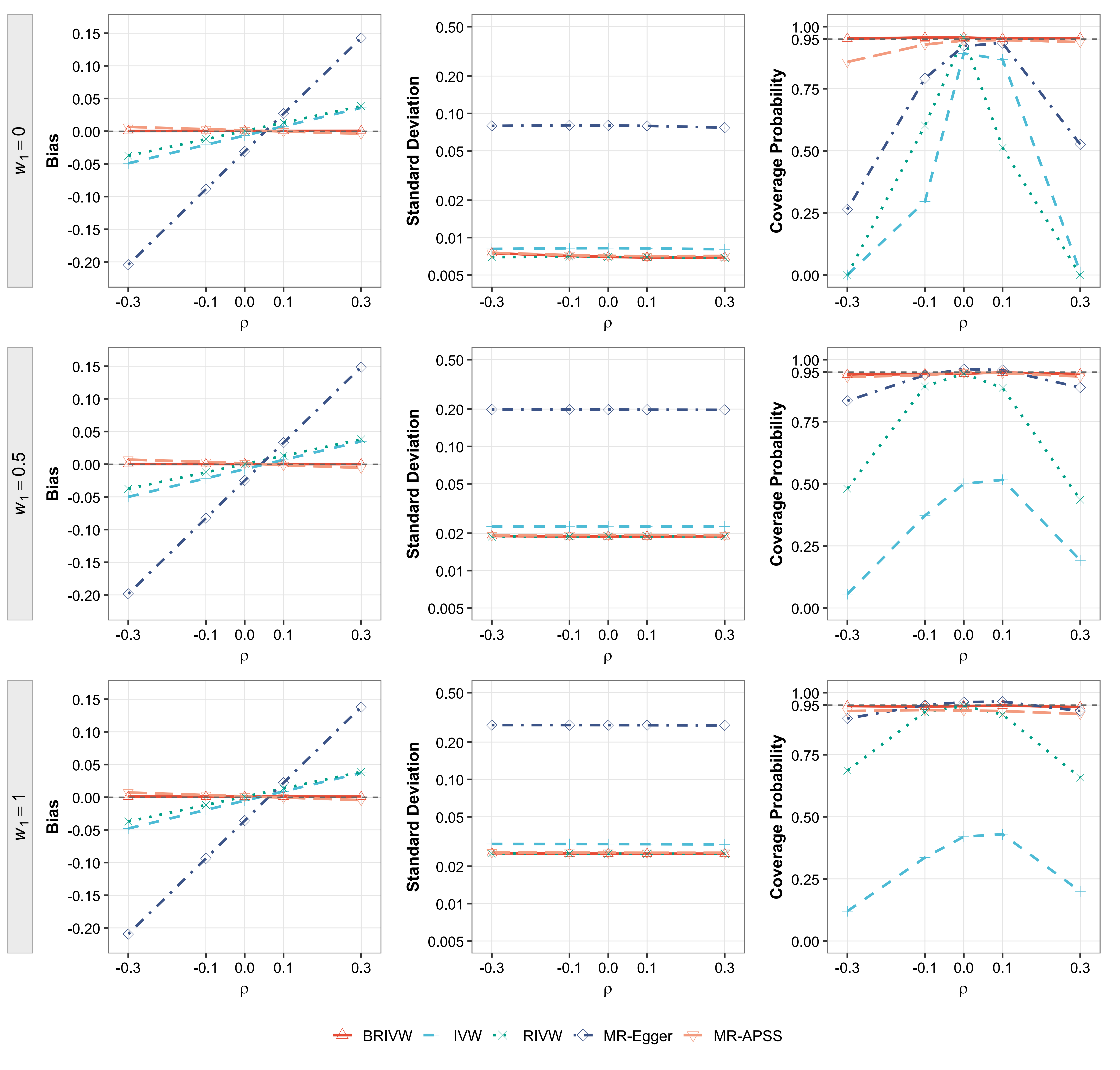}
  \caption{\small
  Simulation results under mixture-model misspecification with
  \(\beta=0.05\).
  Empirical bias, standard deviation, and coverage probability of BRIVW
  and competing methods across different values of \(\rho\) and \(w_1\).
  }
  \label{fig:app_misspec_beta005}
\end{figure}

\begin{figure}[!htbp]
  \centering
  \includegraphics[width=\linewidth]{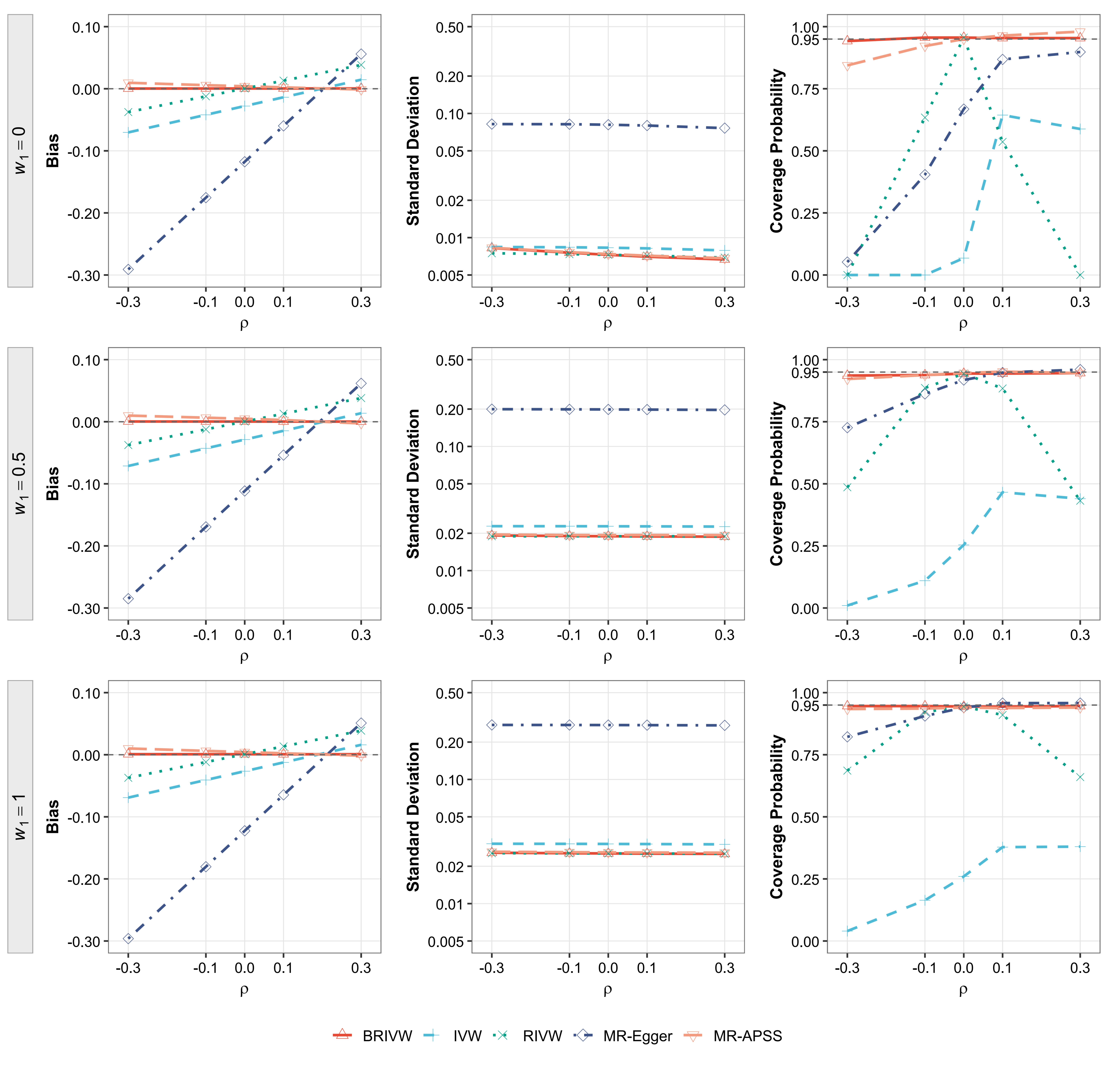}
  \caption{\small
  Simulation results under mixture-model misspecification with
  \(\beta=0.2\).
  Empirical bias, standard deviation, and coverage probability of BRIVW
  and competing methods across different values of \(\rho\) and \(w_1\).
  }
  \label{fig:app_misspec_beta02}
\end{figure}

\begin{figure}[!htbp]
  \centering
  \includegraphics[width=\linewidth]{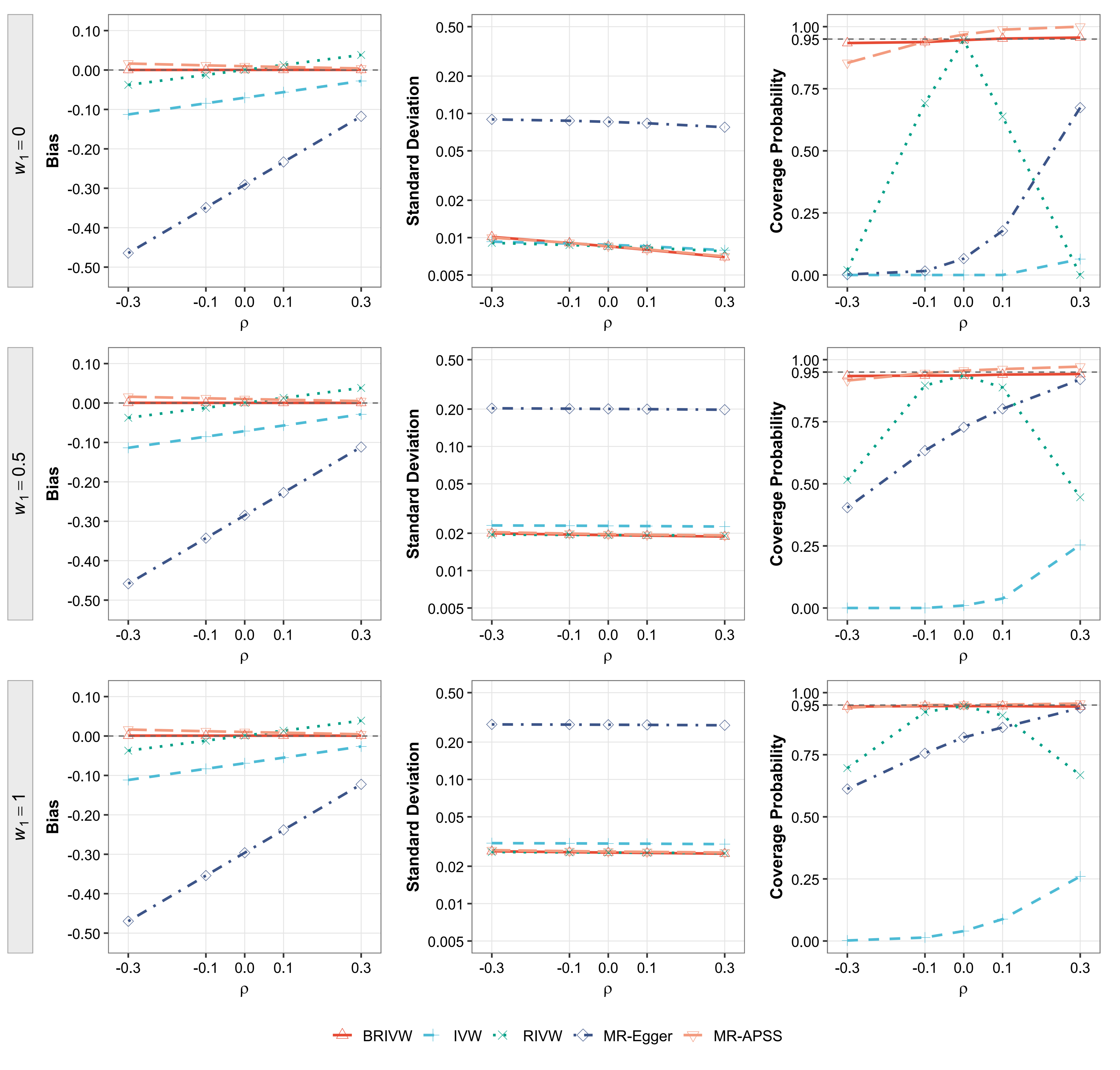}
  \caption{\small
  Simulation results under mixture-model misspecification with
  \(\beta=0.5\).
  Empirical bias, standard deviation, and coverage probability of BRIVW
  and competing methods across different values of \(\rho\) and \(w_1\).
  }
  \label{fig:app_misspec_beta05}
\end{figure}

Web Figures~\ref{fig:app_misspec_beta005}--\ref{fig:app_misspec_beta05} show that,
across all three causal-effect settings, BRIVW remains nearly unbiased
and maintains close-to-nominal coverage under mixture-model
misspecification.
Although MR-APSS exhibits relatively small bias in many settings, its
coverage deteriorates substantially in some cases, particularly for
negative \(\rho\) when \(w_1\) is small.
These results suggest that BRIVW is less sensitive to the considered mixture-model
misspecification.

\setcounter{table}{0}
\renewcommand{\thetable}{S\arabic{table}}
\renewcommand{\tablename}{Web Table}

\section{Additional details and results for the real-data analyses}

\subsection{Quality control of GWAS summary statistics}

All GWAS summary statistics were processed following the quality-control
procedures of Hu et al. (2022), as implemented in the
\texttt{format\_data} function of MR-APSS.
The applied filters are summarized in Web Table~\ref{tab:gwas_qc}.
Filters based on minor allele frequency and imputation quality were
applied only when the corresponding information was available.

\captionsetup{skip=2pt}
\begin{table}[t]
\centering
\footnotesize
\renewcommand{\arraystretch}{1.15}
\caption{Quality-control procedures for GWAS summary statistics.}
\label{tab:gwas_qc}
\begin{tabularx}{\textwidth}{@{}L{0.33\textwidth}Y@{}}
\toprule
\textbf{Quality-control step} & \textbf{Criterion or action} \\
\midrule

Reference-variant restriction
& Retain SNPs included in the HapMap 3 reference panel. \\

Minor allele frequency
& Exclude SNPs with minor allele frequency below 0.05 when allele-frequency information is available. \\

Allele validity
& Exclude SNPs containing alleles other than A, C, G, or T. \\

Ambiguous or invalid allele pairs
& Exclude strand-ambiguous SNPs with A/T or C/G alleles, as well as invalid homoallelic pairs, including A/A, C/C, G/G, and T/T. \\

Imputation quality
& Exclude SNPs with INFO scores below 0.9 when imputation-quality information is available. \\

Major histocompatibility complex
& Exclude SNPs located in the extended major histocompatibility complex region on chromosome 6, from 26 to 34 Mb. \\

Extreme association statistics
& Exclude SNPs with \(\chi_j^2>\chi_{\max,j}^2\), where
\(\chi_{\max,j}^2=\max(N_j/1000,80)\), and \(N_j\) denotes the SNP-specific GWAS sample size. \\

\bottomrule
\end{tabularx}
\end{table}

\subsection{Additional results for the negative-control outcome analysis}

The GWAS sources for the 53 exposure traits and 5 negative-control
outcomes are summarized in Web Table~\ref{tab:negative_control_traits}. 
Web Figure~\ref{fig:app_negative_control_rho} presents a heatmap of the
estimated sample-structure correlations \(\rho\) across all exposure--outcome pairs.

\setlength{\LTleft}{0pt}
\setlength{\LTright}{0pt}
\setlength{\LTpre}{4pt}
\setlength{\LTpost}{4pt}
\setlength{\tabcolsep}{4pt}
\setlength{\LTcapwidth}{\textwidth}
\renewcommand{\arraystretch}{1.10}

\captionsetup{skip=2pt}
\small
\begin{longtable}{@{}L{0.20\textwidth}L{0.50\textwidth}
C{0.14\textwidth}C{0.16\textwidth}@{}}
\caption{GWAS sources for the negative-control outcome analysis.}
\label{tab:negative_control_traits}\\
\toprule
\textbf{Trait} &
\textbf{Description} &
\textbf{Sample size} &
\textbf{Data source} \\
\midrule
\endfirsthead

\multicolumn{4}{c}{\tablename~\thetable\ (continued)}\\
\toprule
\textbf{Trait} &
\textbf{Description} &
\textbf{Sample size} &
\textbf{Data source} \\
\midrule
\endhead

\midrule
\multicolumn{4}{r}{\footnotesize Continued on next page}\\
\endfoot

\bottomrule
\endlastfoot

\multicolumn{4}{@{}l}{\textit{Exposure traits}}\\
\addlinespace[2pt]

AD
& Late-onset Alzheimer's disease
& 54,162
& \href{\detokenize{https://kp4cd.org/node/1354}}{GWAS Catalog} \\

Alcohol
& Drinks per week
& 414,343
& \href{\detokenize{https://thessgac.com/register/}}{SSGAC} \\

Alcohol with meals
& Alcohol usually taken with meals
& 197,613
& \href{\detokenize{https://atlas.ctglab.nl/ukb2_sumstats/f.1618.0.0_logistic.EUR.sumstats.MACfilt.txt.gz}}{GWAS Atlas} \\

Angina
& Self-reported angina
& 289,307
& \href{\detokenize{https://atlas.ctglab.nl/ukb2_sumstats/20002_1074_logistic.EUR.sumstats.MACfilt.txt.gz}}{GWAS Atlas} \\

Anorexia
& Anorexia nervosa
& 72,517
& \href{\detokenize{https://doi.org/10.6084/m9.figshare.14671980}}{PGC} \\

Asthma
& Asthma
& 26,475
& \href{\detokenize{https://gwas.mrcieu.ac.uk/datasets/ieu-a-44/}}{OpenGWAS} \\

Beer intake
& Average weekly beer and cider intake
& 274,556
& \href{\detokenize{https://atlas.ctglab.nl/ukb2_sumstats/f.1588.0.0_res.EUR.sumstats.MACfilt.txt.gz}}{GWAS Atlas} \\

Bipolar disorder
& Bipolar disorder
& 16,731
& \href{\detokenize{https://gwas.mrcieu.ac.uk/datasets/ieu-a-801/}}{OpenGWAS} \\

Birth weight
& Birth weight
& 143,677
& \href{\detokenize{https://gwas.mrcieu.ac.uk/datasets/ieu-a-1083/}}{OpenGWAS} \\

BMI
& Body mass index
& 385,336
& \href{\detokenize{https://atlas.ctglab.nl/ukb2_sumstats/f.21001.0.0_res.EUR.sumstats.MACfilt.txt.gz}}{GWAS Atlas} \\

BFP
& Body fat percentage
& 379,615
& \href{\detokenize{https://atlas.ctglab.nl/ukb2_sumstats/f.23099.0.0_res.EUR.sumstats.MACfilt.txt.gz}}{GWAS Atlas} \\

CAD
& Coronary artery disease
& 184,305
& \href{\detokenize{https://gwas.mrcieu.ac.uk/datasets/ieu-a-7/}}{OpenGWAS} \\

Celiac disease
& Celiac disease
& 24,269
& \href{\detokenize{https://gwas.mrcieu.ac.uk/datasets/ieu-a-1058/}}{OpenGWAS} \\

Cheese intake
& Cheese intake
& 377,082
& \href{\detokenize{https://atlas.ctglab.nl/ukb2_sumstats/f.1408.0.0_res.EUR.sumstats.MACfilt.txt.gz}}{GWAS Atlas} \\

Childhood obesity
& Childhood obesity
& 13,848
& \href{\detokenize{https://gwas.mrcieu.ac.uk/datasets/ieu-a-1096/}}{OpenGWAS} \\

Daytime sleepiness
& Daytime sleepiness
& 452,071
& \href{\detokenize{https://personal.broadinstitute.org/ryank/archive/Saxena.fullUKBB.DaytimeSleepiness.sumstats.zip}}{Broad} \\

DBP
& Diastolic blood pressure
& 757,601
& \href{\detokenize{https://gwas.mrcieu.ac.uk/datasets/ieu-b-39/}}{OpenGWAS} \\

Extreme height
& Extreme height
& 16,196
& \href{\detokenize{https://gwas.mrcieu.ac.uk/datasets/ieu-a-86/}}{OpenGWAS} \\

Extreme BMI
& Extreme body mass index
& 16,068
& \href{\detokenize{https://gwas.mrcieu.ac.uk/datasets/ieu-a-85/}}{OpenGWAS} \\

Education
& Educational qualifications
& 318,526
& \href{\detokenize{https://atlas.ctglab.nl/ukb2_sumstats/f.6138.0.0_res.EUR.sumstats.MACfilt.txt.gz}}{GWAS Atlas} \\

Fresh fruit intake
& Fresh fruit intake
& 372,901
& \href{\detokenize{https://atlas.ctglab.nl/ukb2_sumstats/f.1309.0.0_res.EUR.sumstats.MACfilt.txt.gz}}{GWAS Atlas} \\

HDL-C
& High-density lipoprotein cholesterol
& 187,167
& \href{\detokenize{https://gwas.mrcieu.ac.uk/datasets/ieu-a-299/}}{OpenGWAS} \\

Height (UKBB)
& Standing height in UK Biobank
& 385,748
& \href{\detokenize{https://atlas.ctglab.nl/ukb2_sumstats/f.50.0.0_res.EUR.sumstats.MACfilt.txt.gz}}{GWAS Atlas} \\

Income
& Household income
& 286,301
& \href{\detokenize{https://ftp.ebi.ac.uk/pub/databases/gwas/summary_statistics/GCST009001-GCST010000/GCST009523/HillWD_31844048_household_Income.txt.gz}}{GWAS Catalog} \\

Insomnia
& Insomnia symptoms
& 453,379
& \href{\detokenize{https://personal.broadinstitute.org/ryank/archive/Saxena_fullUKBB_Insomnia_summary_stats.zip}}{Broad} \\

LDL-C
& Low-density lipoprotein cholesterol
& 173,082
& \href{\detokenize{https://gwas.mrcieu.ac.uk/datasets/ieu-a-300/}}{OpenGWAS} \\

MI
& Myocardial infarction
& 171,875
& \href{\detokenize{https://gwas.mrcieu.ac.uk/datasets/ieu-a-798/}}{OpenGWAS} \\

MS
& Multiple sclerosis
& 38,589
& \href{\detokenize{https://gwas.mrcieu.ac.uk/datasets/ieu-a-1025/}}{OpenGWAS} \\

Neuroticism
& Neuroticism score
& 312,740
& \href{\detokenize{https://atlas.ctglab.nl/ukb2_sumstats/f.20127.0.0_res.EUR.sumstats.MACfilt.txt.gz}}{GWAS Atlas} \\

No. of brothers
& Number of full brothers
& 380,062
& \href{\detokenize{https://atlas.ctglab.nl/ukb2_sumstats/f.1873.0.0_res.EUR.sumstats.MACfilt.txt.gz}}{GWAS Atlas} \\

No. of sisters
& Number of full sisters
& 380,122
& \href{\detokenize{https://atlas.ctglab.nl/ukb2_sumstats/f.1883.0.0_res.EUR.sumstats.MACfilt.txt.gz}}{GWAS Atlas} \\

Obesity I
& Obesity class I
& 98,697
& \href{\detokenize{https://gwas.mrcieu.ac.uk/datasets/ieu-a-90/}}{OpenGWAS} \\

Obesity II
& Obesity class II
& 72,546
& \href{\detokenize{https://gwas.mrcieu.ac.uk/datasets/ieu-a-91/}}{OpenGWAS} \\

Overweight
& Overweight
& 158,855
& \href{\detokenize{https://gwas.mrcieu.ac.uk/datasets/ieu-a-93/}}{OpenGWAS} \\

Processed meat
& Processed meat intake
& 385,801
& \href{\detokenize{https://atlas.ctglab.nl/ukb2_sumstats/f.1349.0.0_res.EUR.sumstats.MACfilt.txt.gz}}{GWAS Atlas} \\

RA
& Rheumatoid arthritis
& 80,799
& \href{\detokenize{https://gwas.mrcieu.ac.uk/datasets/ieu-a-833/}}{OpenGWAS} \\

Reaction time
& Mean time to correctly identify matches
& 383,748
& \href{\detokenize{https://atlas.ctglab.nl/ukb2_sumstats/f.20023.0.0_res.EUR.sumstats.MACfilt.txt.gz}}{GWAS Atlas} \\

Salad intake
& Salad or raw vegetable intake
& 363,780
& \href{\detokenize{https://atlas.ctglab.nl/ukb2_sumstats/f.1299.0.0_res.EUR.sumstats.MACfilt.txt.gz}}{GWAS Atlas} \\

SBP
& Systolic blood pressure
& 757,601
& \href{\detokenize{https://gwas.mrcieu.ac.uk/datasets/ieu-b-38/}}{OpenGWAS} \\

SCZ
& Schizophrenia
& 105,318
& \href{\detokenize{https://doi.org/10.6084/m9.figshare.14681220}}{PGC} \\

Smoking
& Smoking initiation, excluding UK Biobank
& 249,171
& \href{\detokenize{https://conservancy.umn.edu/handle/11299/241912}}{GSCAN} \\

Social club
& Participation in pub or social club activities
& 385,280
& \href{\detokenize{https://atlas.ctglab.nl/ukb2_sumstats/6160_2_logistic.EUR.sumstats.MACfilt.txt.gz}}{GWAS Atlas} \\

Stroke
& Stroke
& 446,696
& \href{\detokenize{https://gwas.mrcieu.ac.uk/datasets/ebi-a-GCST005838/}}{OpenGWAS} \\

T2D
& Type 2 diabetes
& 898,130
& \href{\detokenize{https://www.diagram-consortium.org/downloads.html}}{DIAGRAM} \\

TC
& Total cholesterol
& 187,365
& \href{\detokenize{https://gwas.mrcieu.ac.uk/datasets/ieu-a-301/}}{OpenGWAS} \\

TFFM
& Trunk fat-free mass
& 379,507
& \href{\detokenize{https://atlas.ctglab.nl/ukb2_sumstats/f.23129.0.0_res.EUR.sumstats.MACfilt.txt.gz}}{GWAS Atlas} \\

TFM
& Trunk fat mass
& 379,578
& \href{\detokenize{https://atlas.ctglab.nl/ukb2_sumstats/f.23128.0.0_res.EUR.sumstats.MACfilt.txt.gz}}{GWAS Atlas} \\

TFP
& Trunk fat percentage
& 379,600
& \href{\detokenize{https://atlas.ctglab.nl/ukb2_sumstats/f.23127.0.0_res.EUR.sumstats.MACfilt.txt.gz}}{GWAS Atlas} \\

TPM
& Trunk predicted mass
& 379,469
& \href{\detokenize{https://atlas.ctglab.nl/ukb2_sumstats/f.23130.0.0_res.EUR.sumstats.MACfilt.txt.gz}}{GWAS Atlas} \\

UC
& Ulcerative colitis
& 47,745
& \href{\detokenize{https://gwas.mrcieu.ac.uk/datasets/ieu-a-970/}}{OpenGWAS} \\

WBFM
& Whole-body fat mass
& 379,203
& \href{\detokenize{https://atlas.ctglab.nl/ukb2_sumstats/f.23100.0.0_res.EUR.sumstats.MACfilt.txt.gz}}{GWAS Atlas} \\

WBWM
& Whole-body water mass
& 379,835
& \href{\detokenize{https://atlas.ctglab.nl/ukb2_sumstats/f.23102.0.0_res.EUR.sumstats.MACfilt.txt.gz}}{GWAS Atlas} \\

WHR
& Waist-to-hip ratio
& 224,459
& \href{\detokenize{https://gwas.mrcieu.ac.uk/datasets/ieu-a-72/}}{OpenGWAS} \\

\addlinespace[4pt]
\multicolumn{4}{@{}l}{\textit{Negative-control outcomes}}\\
\addlinespace[2pt]

Tanning
& Ease of skin tanning
& 378,364
& \href{\detokenize{https://atlas.ctglab.nl/ukb2_sumstats/f.1727.0.0_res.EUR.sumstats.MACfilt.txt.gz}}{GWAS Atlas} \\

Hair: black
& Natural black hair colour before greying
& 385,603
& \href{\detokenize{https://atlas.ctglab.nl/ukb2_sumstats/1747_5_logistic.EUR.sumstats.MACfilt.txt.gz}}{GWAS Atlas} \\

Hair: blonde
& Natural blonde hair colour before greying
& 385,603
& \href{\detokenize{https://atlas.ctglab.nl/ukb2_sumstats/1747_1_logistic.EUR.sumstats.MACfilt.txt.gz}}{GWAS Atlas} \\

Hair: dark brown
& Natural dark-brown hair colour before greying
& 385,603
& \href{\detokenize{https://atlas.ctglab.nl/ukb2_sumstats/1747_4_logistic.EUR.sumstats.MACfilt.txt.gz}}{GWAS Atlas} \\

Hair: light brown
& Natural light-brown hair colour before greying
& 385,603
& \href{\detokenize{https://atlas.ctglab.nl/ukb2_sumstats/1747_3_logistic.EUR.sumstats.MACfilt.txt.gz}}{GWAS Atlas} \\

\end{longtable}
\normalsize

\noindent
\footnotesize
\normalsize

\begin{figure}[!htbp]

  \centering

  \includegraphics[width=\linewidth]{FigS10.pdf}

  \caption{\small
  Estimated sample-structure correlations
  $\rho$
  for the 265 exposure--negative-control outcome pairs.
  Asterisks indicate pairs with significant \(\rho\).
  }

  \label{fig:app_negative_control_rho}

\end{figure}

\subsection{Additional results for the same trait analyses}

The GWAS datasets used in the same trait analyses are summarized in
Web Table~\ref{tab:same_trait_sources}. Because the exposure and outcome
measure the same standardized trait, the benchmark effect is 1.
Web Table~\ref{tab:same_trait_results} provides the complete numerical
results.

\begin{table}[t]
\centering
\captionsetup{skip=2pt}
\small
\setlength{\tabcolsep}{4pt}
\renewcommand{\arraystretch}{1.12}

\caption{GWAS sources for the same-trait analyses.}
\label{tab:same_trait_sources}

\begin{tabular}{@{}
L{0.10\textwidth}
L{0.40\textwidth}
C{0.17\textwidth}
C{0.15\textwidth}
C{0.15\textwidth}
@{}}
\toprule
\textbf{Dataset} &
\textbf{Description} &
\textbf{GWAS ID} &
\textbf{Sample size} &
\textbf{Data source} \\
\midrule

BMI-1
& Body mass index; UK Biobank
& \texttt{ukb-b-19953}
& 461,460
& \href{\detokenize{https://gwas.mrcieu.ac.uk/datasets/ukb-b-19953/}}
  {\mbox{OpenGWAS}} \\

BMI-2
& Body mass index; GIANT consortium
& \texttt{ieu-a-2}
& 339,224
& \href{\detokenize{https://gwas.mrcieu.ac.uk/datasets/ieu-a-2/}}
  {\mbox{OpenGWAS}} \\

HDL-1
& High-density lipoprotein cholesterol; UK Biobank
& \texttt{ieu-b-109}
& 403,943
& \href{\detokenize{https://gwas.mrcieu.ac.uk/datasets/ieu-b-109/}}
  {\mbox{OpenGWAS}} \\

HDL-2
& High-density lipoprotein cholesterol; Global Lipids Genetics Consortium
& \texttt{ebi-a-GCST002223}
& 94,595
& \href{\detokenize{https://gwas.mrcieu.ac.uk/datasets/ebi-a-GCST002223/}}
  {\mbox{OpenGWAS}} \\

\bottomrule
\end{tabular}

\end{table}

\begin{table}[!htbp]

\centering

\captionsetup{skip=2pt}

\footnotesize

\setlength{\tabcolsep}{4pt}

\renewcommand{\arraystretch}{1.10}

\caption{Complete numerical results of the same trait analyses.}

\label{tab:same_trait_results}

\begin{tabularx}{\textwidth}{@{}
L{0.18\textwidth}
C{0.16\textwidth}
C{0.08\textwidth}
C{0.28\textwidth}
C{0.06\textwidth}
C{0.18\textwidth}
@{}}

\toprule

\textbf{Trait pair} &
\textbf{Method} &
\textbf{IVs} &
\(\boldsymbol{\hat\beta}\) \textbf{(95\% CI)} &
\(\boldsymbol{\rho}\) &
\textbf{\(P\)} \\

\midrule

\multirow{5}{*}{BMI-1 \(\rightarrow\) BMI-2}
& BRIVW
& 744
& 1.024 (0.971, 1.077)
& \multirow{5}{*}{0.089}
& \multirow{5}{*}{\(1.40\times10^{-14}\)} \\

& RIVW
& 860
& 0.980 (0.934, 1.025)
& & \\

& IVW
& 419
& 0.784 (0.758, 0.811)
& & \\

& MR-Egger
& 419
& 0.950 (0.810, 1.091)
& & \\

& MR-APSS
& 752
& 0.727 (0.610, 0.845)
& & \\

\midrule

\multirow{5}{*}{BMI-2 \(\rightarrow\) BMI-1}
& BRIVW
& 509
& 1.041 (0.973, 1.109)
& \multirow{5}{*}{0.089}
& \multirow{5}{*}{\(1.40\times10^{-14}\)} \\

& RIVW
& 347
& 1.142 (1.040, 1.245)
& & \\

& IVW
& 73
& 0.875 (0.851, 0.898)
& & \\

& MR-Egger
& 73
& 0.841 (0.573, 1.109)
& & \\

& MR-APSS
& 232
& 1.227 (1.070, 1.385)
& & \\

\midrule

\multirow{5}{*}{HDL-1 \(\rightarrow\) HDL-2}
& BRIVW
& 574
& 0.990 (0.916, 1.063)
& \multirow{5}{*}{0.070}
& \multirow{5}{*}{\(3.91\times10^{-11}\)} \\

& RIVW
& 690
& 0.964 (0.901, 1.027)
& & \\

& IVW
& 301
& 0.849 (0.815, 0.883)
& & \\

& MR-Egger
& 301
& 0.867 (0.742, 0.992)
& & \\

& MR-APSS
& 572
& 0.884 (0.751, 1.018)
& & \\

\midrule

\multirow{5}{*}{HDL-2 \(\rightarrow\) HDL-1}
& BRIVW
& 295
& 1.021 (0.922, 1.120)
& \multirow{5}{*}{0.070}
& \multirow{5}{*}{\(3.91\times10^{-11}\)} \\

& RIVW
& 308
& 1.045 (0.940, 1.150)
& & \\

& IVW
& 89
& 0.884 (0.867, 0.901)
& & \\

& MR-Egger
& 89
& 0.722 (0.442, 1.002)
& & \\

& MR-APSS
& 199
& 0.929 (0.814, 1.044)
& & \\

\bottomrule

\end{tabularx}

\vspace{2pt}

\begin{minipage}{0.98\textwidth}

\footnotesize

\textit{Note:}
IVs denote the number of selected instrumental variables.
\(P\) represents the \(P\)-value of the test for \(\rho=0\).

\end{minipage}

\end{table}

\subsection{Additional details for the BMI and cardiovascular outcome analyses}

The GWAS datasets used in the BMI and cardiovascular outcome analyses
are summarized in Web Table~\ref{tab:bmi_cvd_sources}.

\begin{table}[t]
\centering
\captionsetup{skip=2pt}
\small
\setlength{\tabcolsep}{4pt}
\renewcommand{\arraystretch}{1.12}

\caption{GWAS sources for the BMI and cardiovascular outcome analyses.}
\label{tab:bmi_cvd_sources}

\begin{tabular}{@{}
L{0.10\textwidth}
L{0.40\textwidth}
C{0.17\textwidth}
C{0.15\textwidth}
C{0.15\textwidth}
@{}}
\toprule
\textbf{Trait} &
\textbf{Description} &
\textbf{GWAS ID} &
\textbf{Sample size} &
\textbf{Data source} \\
\midrule

BMI
& Body mass index; UK Biobank
& \texttt{ukb-b-19953}
& 461,460
& \href{\detokenize{https://gwas.mrcieu.ac.uk/datasets/ukb-b-19953/}}
  {\mbox{OpenGWAS}} \\

SBP
& Systolic blood pressure, automated reading; UK Biobank
& \texttt{ukb-b-20175}
& 436,419
& \href{\detokenize{https://gwas.mrcieu.ac.uk/datasets/ukb-b-20175/}}
  {\mbox{OpenGWAS}} \\

DBP
& Diastolic blood pressure, automated reading; UK Biobank
& \texttt{ukb-b-7992}
& 436,424
& \href{\detokenize{https://gwas.mrcieu.ac.uk/datasets/ukb-b-7992/}}
  {\mbox{OpenGWAS}} \\

CHD
& Coronary heart disease; CARDIoGRAMplusC4D Consortium
& \texttt{ieu-a-9}
& 194,427
& \href{\detokenize{https://gwas.mrcieu.ac.uk/datasets/ieu-a-9/}}
  {\mbox{OpenGWAS}} \\

\bottomrule
\end{tabular}

\end{table}

\end{document}